\documentclass[11pt,letterpaper]{article}

\usepackage{amsmath,amssymb,amsthm}
\usepackage{booktabs} 
\usepackage{xspace}
\usepackage[pagebackref=true,hidelinks]{hyperref}
\usepackage{accents} 
\usepackage{pgf} 
\usepackage{xparse} 
\usepackage{bbm}
\usepackage{enumitem}
\usepackage[square]{natbib}
\usepackage{algorithm}
\usepackage[noend]{algorithmic}
\usepackage{verbatim}
\usepackage[margin=1in]{geometry}
\usepackage{mathrsfs}
\usepackage{cleveref} 

\usepackage{pifont}
\usepackage{subcaption}
\newcommand{\cmark}{\ding{51}}%
\newcommand{\xmark}{\ding{55}}%

\usepackage{color}
\newcommand{\kgnote}[1]{{\color{orange}{#1}}}

\newtheorem{lemma}{Lemma}
\newtheorem{theorem}{Theorem}
\newtheorem{corollary}[theorem]{Corollary}
\newtheorem{proposition}[theorem]{Proposition}

\theoremstyle{definition}
\newtheorem{definition}{Definition}

\usepackage{todonotes}

\usepackage{accsupp}

\newtheorem*{claim}{Claim}
\newtheorem{observation}{Observation}
\newtheorem{thm}{Theorem}
\newtheorem{remark}[thm]{Remark}
\usepackage{bbm,bm}
\usepackage[normalem]{ulem}

\newcommand{\R}{\mathbb{R}}

\newcommand{\argmax}{\operatorname{arg\,max}}
\newcommand{\argmin}{\operatorname{arg\,min}}

\newcommand{\menu}{M}
\newcommand{\menuseq}{M}

\newcommand{\E}{\mathbb{E}}

\newcommand{\vareps}{\varepsilon}

\newcommand{\sgn}{\ensuremath{\mathrm{sgn}}}

\def\P{\ensuremath{\mathcal{P}}}
\def\R{\ensuremath{\mathbb{R}}}

\def\p{\ensuremath{{\bf p}}}

\def\t{\ensuremath{\tilde}}
\def\L{\ensuremath{\mathcal{L}}}
\def\children{\ensuremath{N^+}}
\def\N{\ensuremath{N^+}}
\def\parents{\ensuremath{N^-}}
\def\G{\ensuremath{\mathcal{G}}}
\def\u{\ensuremath{u}}
\def\z{\ensuremath{r}}

\def\items{\mathsf{Items}}
\def\types{\mathsf{Types}}
\def\UC{\mathtt {Uncountable}}
\def\C{\mathtt {Countable}}
\def\v{\mathfrak {v}}
\def\I{\mathcal {I}}
\def\p{\mathfrak {p}}
\def\D{\mathscr {D}}
\def\d{\mathfrak {d}}
\def\v{\mathfrak {v}}
\def\t{\mathfrak {t}}
\def\setname{coordinated items } 

\DeclareMathOperator{\MC}{\mathsf{MC}}

\defcitealias{CDW}{CDW16}
\defcitealias{FGKK}{FGKK16}
\defcitealias{DW}{DW17}
\defcitealias{DHP}{DHP17}
\defcitealias{Myerson}{Mye81}
\defcitealias{DDT15}{DDT15}
\defcitealias{SSW17}{SSW18}
\defcitealias{MV}{MV07}
\defcitealias{CheGale2000}{Che and Gale 2000}

\usepackage{aliascnt} 

\newenvironment{numberedtheorem}[1]{%
\renewcommand{\thetheorem}{#1}%
\begin{theorem}}{\end{theorem}\addtocounter{theorem}{-1}}

\begin{document}
\title{{Optimal Mechanism Design for Single-Minded Agents}}
\author{Nikhil R. Devanur%
\thanks{%
    {Amazon (\url{Iam@nikhildevanur.com}).}}
\and Kira Goldner%
\thanks{%
    {Columbia University (\url{kgoldner@cs.columbia.edu}).  Supported in part by NSF CCF-1420381 and by a Microsoft Research PhD Fellowship. Supported in part by NSF award DMS-1903037 and a Columbia Data Science Institute postdoctoral fellowship.}}
\and Raghuvansh R. Saxena%
\thanks{%
    {Princeton University (\url{rrsaxena@princeton.edu}).  Supported by NSF CAREER award CCF-1750443.}}
\and Ariel Schvartzman%
\thanks{%
    {Princeton University (\url{acohenca@princeton.edu})}. Supported by NSF CCF-1717899.}    
\and S. Matthew Weinberg%
\thanks{%
    {Princeton University (\url{smweinberg@princeton.edu}). Supported by NSF CCF-1717899.}} }

\date{}
\maketitle

\begin{abstract}

We consider optimal (revenue maximizing) mechanism design in the interdimensional setting, where one dimension is the `value' of the buyer, 
and the other is a `type' that captures some auxiliary information. 
A prototypical example of this is the FedEx Problem, 
for which ~\citet{FGKK} characterize the optimal mechanism for a single agent. 
Another example of this is when the type encodes the buyer's budget \citepalias{DW}. 
The question we address is \emph{how far can such characterizations go?}
In particular, we consider the setting of single-minded agents. A seller has heterogenous items. A buyer has a valuation $v$ for a specific subset of items $S$, and obtains value $v$ if and only if he gets all the items in $S$ (and potentially some others too). 

We show the following results.  
\begin{enumerate}
	\item Deterministic mechanisms (i.e. posted prices) are optimal for distributions that satisfy the ``declining marginal revenue'' (DMR) property. 
	 In this case we give an explicit construction of the optimal mechanism.
	\item Without the DMR assumption, the result depends on the structure of the minimal directed acyclic graph (DAG) representing the partial order among types. 
	When the DAG has out-degree at most 1, we characterize the optimal mechanism \`{a} la FedEx; this can be thought of as a generalization of the FedEx characterization since FedEx corresponds to a DAG that is a line. 
	\item Surprisingly, without the DMR assumption \emph{and} when the DAG has at least one node with an out-degree of at least 2, then we show that there is no hope of such a characterization.  
	The minimal such example happens on a DAG with 3 types. 
	We show that in this case the menu complexity is {\em unbounded} in that for any $\menu$, there exist distributions over $(v,S)$ pairs such that the menu complexity of the optimal mechanism is at least $\menu$. 
	\item  For the case of 3 types, we also show that for all distributions there exists an optimal mechanism of {\em finite} menu complexity.
	This is in contrast to the case where you have 2 heterogenous items with additive utilities for which the menu complexity could be uncountably infinite~\citepalias{MV, DDT15}. 
\end{enumerate}

In addition, we prove that optimal mechanisms for Multi-Unit Pricing (without a DMR assumption) can have unbounded menu complexity as well, and we further propose an extension where the menu complexity of optimal mechanisms can be countably infinite, but not uncountably infinite.

Taken together, these results establish that optimal mechanisms in interdimensional settings are both surprisingly richer than single-dimensional settings, yet also vastly more structured than multi-dimensional settings.

\end{abstract}

\addtocounter{page}{-1}
\newpage


\section{Introduction}
Consider the problem of selling multiple items to a unit-demand buyer. The fundamental problem underlying much of mechanism design asks how the seller should maximize their revenue. If the items are identical, then the setting is considered \emph{single-dimensional.}  In this case, seminal work of \citet{Myerson} completely resolves this question with an exact characterization of the optimal mechanism.  The optimal mechanism is a simple take-it-or-leave-it price, and the fact that there are multiple items versus just one is irrelevant. In contrast, if the items are heterogenous, then the setting is \emph{multi-dimensional} and, unlike the single-dimensional setting, optimal mechanisms are no longer tractable in any sense:~\citep{MV,BCKW15, HN13, HR12, DDT13, DDT15}.



Very recently,  \citet{FGKK} 
identify a fascinating middle-ground. Imagine that the items are neither identical nor heterogeneous, but are instead varying qualities of the same item. To have an example in mind, imagine that you're shipping a package and the  items are one-day, two-day, or three-day shipping. You obtain some value $v$ for having your package shipped, but only if it arrives by your deadline (which is one, two, or three days from now). 
We can think of the input  as being a (correlated) two-dimensional distribution over (value, deadline) pairs. 

{The FedEx Problem is a special case of single-minded valuations: a buyer has a valuation $v$ for a specific subset of items $S$, and obtains value $v$ if  he gets any superset of $S$, and 0 otherwise. 
To have an example in mind, imagine that a company offers internet, phone service, and cable TV.  You have a value, $v$, and are interested in getting internet service. So you value options such as exclusively internet service, internet/phone service, or internet/cable, and so on, at $v$. For any option that does not include internet you get a value of zero (so we again think of the input distribution as a two-dimensional distribution over (value, interest) pairs).}

{An alternative perspective to single-minded valuations is that there is a partial order on the set of possible interests a buyer may have. The partial order is just the one induced by set inclusion. The FedEx problem has \emph{totally-ordered} items: one-day shipping is at least as good as two-day shipping is at least as good as three-day shipping, and every buyer agrees.
In fact, any partial order can be induced from set inclusion, so the two settings are equivalent (see Observation~\ref{obs:SMequiv} in \Cref{sec:singleminded}).
It turns out that the partial order view is more useful from a mechanism design perspective, therefore we will use that view for the rest of the paper. 
}

The following problem can also be interpreted as a partially-ordered setting: Suppose that each buyer has a publicly visible attribute which the seller can use to price discriminate. E.g., the buyer could be a student, a senior, or general-admission. Or, the buyer could be a ``prime member'' or a ``non-prime member.''  However, buyers with certain attributes can disguise themselves as having other attributes, given by a partial order. For example, a prime member could disguise as a non-prime member, but not vice-versa. Then if item $i$ is a movie ticket redeemable by anyone who can disguise themselves as having attribute $i$, the items are partially-ordered.

\subsection{Main Results} 

\citet{FGKK} give a characterization of an optimal mechanism for the FedEx problem, and our goal is to understand the generalizability of 
this characterization, in particular to the partially ordered setting. 
Towards this, we first describe the FedEx characterization 
and what a generalization could look like. 
A \emph{deterministic} mechanism sets a posted price $p_i$ for each shipping option, and the buyer picks the option he prefers (if any).
Clearly, it makes sense for the prices to be non-increasing in $i$-day shipping. 
The FedEx solution is recursive:
start with the price on day 1 (as a variable),
and constrain the price on day 2 to be weakly lower, 
and so on. 
When the distributions satisfy the \emph{Declining Marginal Revenue (DMR)}\footnote{A one-dimensional distribution $F$ satisfies Declining Marginal Revenues  if $v(1-F(v))$ is concave. 
	See \cite{DHP} for examples and more discussion. For example, uniform distributions are DMR, along with any distribution of bounded support and monotone non-decreasing density.}
property, this strategy actually results in deterministic prices that are optimal. 
Without any distributional assumption, one might have to resort to \emph{lotteries}: the buyer gets the item only with some probability. 
The first day price is still deterministic, 
but for the second day, the mechanism offers a lottery such that the \emph{expected} price for full service
is weakly lower. 
It turns out that we only need to randomize between two options. 
Recursively, every option on day $i$ may split into two options on day $i+1$, 
so we might have at most $2^{m-1}$ options on day $m$, and $2^m -1$ options overall (and examples exist where $2^m-1$ options are necessary~\citep{SSW17}). 

So our starting point is a hope that similar recursive ideas can characterize optimal auctions beyond the totally-ordered FedEx setting. 
Some terminology is useful here to understand precisely what this might mean. We use the directed acyclic graph (DAG) 
representation of a partial order: an edge from $i$ to $j$ implies 
$j$ \emph{is preferred over} $i$. 
The DAG is \emph{minimal}: if $(i,j)$ and $(j,k)$ are edges then $(i,k)$ is \emph{not} an edge.
The DAG for the FedEx problem goes \emph{right to left}, i.e., it has edges $(i+1,i)$ for $i$ from 1 to $m-1$.  
A recursive approach for a DAG would look like this: 
start with a \emph{sink}, set a deterministic price, and use this 
to constrain the prices (either deterministically or in expectation, based
on the distributional assumption) for its predecessors and so on. 
The goal of this paper is to understand \emph{Will something like this work for partially-ordered items?} 

\paragraph{DMR:} Under the DMR assumption, this strategy for pricing works in any DAG (Theorem~\ref{thm:dmr} in \Cref{app:DMR}). 
We start from the sink nodes and recursively constrain the price of a node 
to be at most the minimum among  the prices of all its successors. Our proof that this procedure works employs LP duality, and a significantly more involved procedure to set appropriate dual variables than in~\citepalias{FGKK}. The fact that optimal mechanisms are deterministic subject to DMR matches prior work for totally-ordered settings \citepalias{CheGale2000,FGKK,DW,DHP}. 


\paragraph{Out-degree 1:}
The FedEx strategy still works when the minimal DAG has out-degree at most 1, 
without any distributional assumptions (Theorem \ref{thm:fedexext} in \Cref{sec:fedexext}).
Compared to the DMR case, we now have to deal with lotteries but when we process a node, 
there is exactly one successor that constrains the lotteries for this node, 
in exactly the same way as in FedEx.

\paragraph{3 node DAG:}
The minimal example where the out-degree is 2 is a three-node DAG with nodes $A,B$, and $C$, and edges $(C,A)$ and $(C,B)$. 
One might hope that the following recursive strategy would work (after all, the graph is still a DAG, and should be amenable to recursive arguments): set prices deterministically 
for $A$ and $B$ and use the minimum of the two to constrain the 
expected price for $C$. Note that if there were no item $B$, this would match precisely the FedEx solution.

It turns out that this idea fails horribly, for the following (very high-level) reason. With just two items ($C$ and $A$), the price of $A$ transparently constrains what prices we can set for $C$ (the expected price for $C$ must be lower). So when optimizing the price of $A$, we can take this into account. With three items, it's no longer clear how the price of $A$ constrains the price of $C$. Certainly, the expected price for $C$ must be lower, but perhaps a stronger constraint is already implied by the price of $B$. Therefore, one cannot separately optimize the price of $A$ without knowing the price of $B$.

Indeed, this intuition actually manifests into a lower bound: it is not only challenging to jointly optimize the prices of $A,B$ together, but the optimum may no longer be deterministic at all! Specifically, for any integer $M$, there exist value distributions for this 3-node DAG for which the unique optimal mechanism presents $M$ different lotteries to the buyer (Theorem~\ref{thm:unbounded}). Essentially, \emph{there is no hope for a FedEx style solution even for this minimal case}. We focus the technical presentation of our paper on this result.



\paragraph{Finiteness of menu complexity:}
 The use of menu complexity lower bounds to ascertain complexity of 
 mechanisms is not new: 
 \citet{MV,DDT15} show that the optimal mechanism for the multi-dimensional setting might have \emph{uncountable} menu complexity---this holds even for just two items with \emph{additive} valuations, and even when the item values are drawn independently from absolutely bounded distributions. This dichotomy serves as one fundamental difference between 
single-dimensional and multi-dimensional settings.

Within the context of these results, we ask if we can get an infinite (uncountable or countable) menu complexity for the partially-ordered setting as well. 
A natural strategy is to take the limit of our construction 
as the number of randomizations goes to infinity. 
Somewhat surprisingly, the example then collapses 
and has a deterministic price as optimal. 
We show that this is no coincidence: that the menu complexity for the three item case is \emph{always finite} (\Cref{thm:allocalg}).

\paragraph{Summary:} 
The main technical takeaway from our results is a thorough understanding of optimal mechanisms in interdimensional settings beyond FedEx through broadly applicable tools. Our theorem statements use the language of menu complexity, but only to distinguish among mechanisms with bounded, unbounded, or infinite menu complexity.  The main conceptual takeaway is that optimal auctions for single-minded valuations lie in a space of their own: significantly more complex than optimal single-dimensional auctions, or even optimal auctions for totally-ordered valuations, yet  more structured than optimal multi-dimensional auctions.

\begin{figure}
{\footnotesize{
\begin{center}
Known Menu Complexity Results for Optimal Mechanisms with One Buyer \\
\begin{tabular}{|c|c|c|c|c|c|c|} \hline
& One Item & FedEx & Single-Minded, 3 Items & Multi-Unit & Coordinated, 3 Items & Additive \\ \hline
Det. under DMR & N/A & \checkmark & \textbf{\cmark} & \checkmark & \xmark & N/A \\ \hline
Lower Bound &  1 & $2^m-1$  & \textbf{unbounded} & \textbf{unbounded} & \textbf{countably infinite} & uncountable \\ \hline
Upper Bound &  1 & $2^m-1$ & \textbf{finite} & --- & \textbf{countably infinite} & uncountable  \\ \hline
\end{tabular} \\
Bold results are from this paper.
\end{center}}}
\end{figure}

\subsection{Additional results} 
We postpone all details about our proofs to the technical sections, but highlight one result of independent interest that we develop en route. Our problem can be phrased as a continuous linear program, and all of our proofs require reasoning about the dual. In particular, developing our lower bound construction (instances with unbounded menu complexity) consists of two parts: First, we construct a candidate dual $\lambda$ for which a primal exists satisfying complementary slackness, and for which \emph{every} primal satisfying complementary slackness has menu complexity $\geq \menu$. Second, we prove that there exists a distribution for which $\lambda$ is a feasible dual (and combining these two claims means that every optimal mechanism for this input has menu complexity $\geq \menu$).  Analyzing $\lambda$ through complementary slackness is technically interesting, and captures all of the insight one would hope to gain from the construction. Reverse engineering an instance for which $\lambda$ is feasible, however, is technically challenging yet unilluminating. On this front, we prove a ``Master Theorem,'' stating essentially that every candidate dual is feasible for some input distribution (Theorem~\ref{thm:multimaster}). This allows the user (of the theorem) to reason exclusively about primals and duals, letting the Master Theorem map the candidate pair back to an instance for which they are feasible. In some sense, the Master Theorem formally separates the insightful analysis from the tedious parts.

Of course, one should not expect this theorem to hold in general multi-dimensional settings (in particular, one key property that enables our Master Theorem is a ``payment identity,'' which general multi-dimensional settings notoriously lack---this is a further example of how our setting lies in-between single- and multi-dimensional), but the Master Theorem is quite generally applicable for problems in this intermediate range. In addition, because the Master Theorem takes care of guaranteeing that distributions corresponding to some dual will exist, this result also emphasizes the strength of reasoning about duals in similar settings.

Finally, beyond our main results, we prove {two} additional results using the same tools. 
First, we apply our lower bound techniques to show that the menu complexity of the Multi-Unit Pricing problem~\citepalias{DHP} is also unbounded (Theorem~\ref{thm:MUPunbounded} in Appendix~\ref{sec:MUPLB}). Multi-Unit Pricing is also a totally-ordered setting, where the items correspond to copies of a good (item one is one copy, item two is two copies, item three is three copies). The difference from FedEx is that if the buyer is interested in two copies but gets one, they get half their value (versus zero). {Second, we propose a generalization beyond totally-ordered settings which we call \emph{coordinated} valuations, and again characterize the menu complexity of optimal mechanisms for one instance of three items (which can be countably infinite, but not uncountable, see Appendix~\ref{sec:raghuvansh}).

\subsection{Related Work}\label{sec:related}


Single-minded valuations are a well-known model (e.g. \citep{LOS02}).  Most work in this model pertains to welfare maximization in more complex settings, such as combinatorial auctions.  Other work assumes that the buyer's interest is publicly known; in this case, the buyer is single-parameter, and a single-buyer revenue maximization problem reduces to Myerson.  

The most related line of works has already mostly been discussed. The FedEx Problem considers totally-ordered items (in our language), as does Multi-Unit Pricing and Budgets~\citep{CheGale2000,FGKK, DHP, DW}. The present paper is the first to consider partially-ordered items. In terms of techniques, we indeed draw on tools from prior work. All three prior works employ some form of duality. Our approach is most similar to that of~\citet{DW} in that (1) both are the only works to use  the analysis from \citepalias{CDW} to characterize optimal mechanisms rather than obtain approximations, and (2) we also perform ``dual operations'' rather than search for a closed form.  However, as the single-minded setting is much more complicated, we extend the techniques to handle this setting.

Also related is a long line of work which aims to characterize optimal mechanisms beyond single-dimensional settings. Owing to the inherent complexity of mechanism design for heterogeneous items, results on this front necessarily consider restricted settings~\citep{LMR,GK-uniform,MM,DDT13,DDT15,HH,malakhov2009optimal}. From this set, the most related are~\cite{HH, malakhov2009optimal}, who also considered settings where all consumers prefer (e.g.) item $a$ to item $b$, but there are no substantial technical connections.




There is also a quickly growing body of work regarding the menu complexity of multi-item auctions. Much of this work focuses on settings with heterogeneous items~\citep{BCKW15,HN13,BGN17,WangT14,DDT15, Gonczarowski18}. Very recent work of~\citepalias{SSW17} considers the menu complexity of approximately optimal mechanisms for the FedEx Problem (for which~\citepalias{FGKK} already characterized the menu complexity of exactly optimal mechanisms). On this front, our work places partially-ordered items (where the menu complexity is finite but unbounded)  distinctly between totally-ordered items (where the menu complexity is bounded)~\citepalias{FGKK}, and heterogeneous items (uncountable)~\citepalias{DDT15}. Previously, no settings with this property were known.

\subsection{Roadmap}
Our paper contains four main results, although we view the primary contributions as (3) and (4):
\begin{enumerate}
\item In~\Cref{app:DMR}, we prove Theorem~\ref{thm:dmr}, which explicitly constructs a deterministic optimal auction for partially-ordered items when all marginals are DMR.
\item In~\Cref{sec:fedexext}, we prove Theorem~\ref{thm:fedexext}, which extends the recursive FedEx algorithm for totally-ordered items to partially-ordered items when minimal DAGs with outdegree at most one.
\item We focus our technical presentation on the ideas necessary for Theorem~\ref{thm:unbounded}, which establishes that any partially-ordered instance for which some node in the minimal DAG has outdegree at least two, the menu complexity of the optimal mechanism may be unbounded. In~\Cref{sec:prelim} we provide the minimal preliminaries to understand the main ideas behind our proof of this result (full preliminaries in~\Cref{sec:addtlprelims}). In~\Cref{sec:highlevel} we overview the key duality aspects. In~\Cref{sec:menucomplexity} we give a brief overview of the proof of Theorem~\ref{thm:unbounded}. The full proof is in~\cref{app:LB}.
\item Finally, we also establish that the menu complexity of optimal mechanisms for this minimal $3$-item instance is always finite. The main ideas appear in~\Cref{sec:menucomplexity}, and a full proof of Theorem~\ref{thm:allocalg} appears in~\Cref{app:optvars}. 
\end{enumerate}

Outside of our main results,~\Cref{sec:masterthmApp} presents our ``Master Theorem'' (Theorem~\ref{thm:multimaster}), which is of independent interest for future work on mechanism design with totally- or partially-ordered items. In \Cref{sec:MUPLB} and \Cref{sec:raghuvansh} we display the applicability of our techniques for related settings such as Multi-Unit Pricing (Theorem~\ref{thm:MUPunbounded}) and coordinated values (Theorems~\ref{thm:LB1}, \ref{thm:LB2}, \ref{thm:LB-gen}), respectively.~\Cref{sec:conclusion} presents our conclusions and discusses future work.


\section{Preliminaries} \label{sec:prelim}
In the interest of presentation, we'll provide the minimum preliminaries here for the reader to understand the key ideas. In \Cref{sec:addtlprelims}, we provide full preliminaries, including additional intuition, and covering prior work (such as~\citepalias{FGKK, DW}). Many of the facts we will use are stated here without proof (proofs are given in \Cref{sec:addtlprelims}). 
\subsection{A Minimal Instance}

We focus on the three-item case with items $\G = \{A,B,C\}$ where $A \succ C$ and $B \succ C$, but $A \not \succ B$ and $B \not \succ A$.  That is, if a buyer is interested in item $C$, they are content with $A$ or $B$. If they are interested in $A$, they are content only with $A$ (ditto for $B$). There is a single buyer with a (value, interest) pair $(v,G)$, who receives value $v$ if they are awarded an item $\succeq G$ (that is, $G' \succ G$ or $G' = G$). This is the minimal non-trivial example of a partially-ordered setting.  A menu-complexity lower bound for this example applies to any partially-ordered setting that contains an item $G$ with at least two incomparable items that dominate $G$ (which includes every single-minded valuation setting with at least $3$ items). 

An instance of the problem consists of a joint probability distribution over $[0,H]\times \G$, where $H$ is the maximum possible value of any bidder for any item.\footnote{Note that the multi-dimensional instances with uncountable menu complexity are also supported on a compact set: $[0,H]^2$. So our results are not merely a product of compactness.} We will use $f$ to denote the density of this joint distribution, with $f_G(v)$ denoting the density at $(v, G)$. We will also use $F_G(v)$ to denote $\int_0^v f_G(w)dw$, and $q_G$ to denote the probability that the bidder's interest is $G$. 

We'll consider (w.l.o.g.) direct truthful mechanisms, where the bidder reports a (value, interest) pair and is awarded a (possibly randomized) item. Further, as observed in~\citet{FGKK}, it is without loss of generality to only consider mechanisms that award bidders their declared item of interest with probability in $[0,1]$, and all other items with probability $0$.
\footnote{To see this, observe that the bidder is just as happy to get nothing instead of an item that doesn't dominate their interest. See also that they are just as happy to get their interest item instead of any item that dominates it. It will also make this option no more attractive to any bidder considering misreporting. So starting from a truthful mechanism, modifying it to only award the item of declared interest or nothing cannot possibly violate truthfulness. Note also that this modification maintains optimality, but could impact the menu complexity up to a factor of \# items. As we only consider distinctions between bounded, unbounded, and infinite, this is still w.l.o.g.}
For a direct mechanism, we'll define $a_G(v)$ to be the probability that  item $G$ is awarded to a bidder who reports $(v,G)$. Our goal is to find the revenue-optimal allocation rule---$a_G(v)$ defined for all $G \in \G, v \in [0,H]$ with payment determined by the allocation rule---such that the mechanism is incentive-compatible. The menu complexity of a direct mechanism refers to the number of distinct  pairs $(G,q)$ such that there exists a $v$ with $a_G(v) = q$. 



\subsection{Incentive Compatibility, Revenue Curves, and Ironing} \label{subsec:rev-iron}


As observed in~\citepalias{FGKK}, it is without loss of generality to only consider mechanisms that award bidders their declared item of interest with probability in $[0,1]$, and all other items with probability $0$. Also observed in~\citepalias{FGKK} is that Myerson's payment identity holds in this setting as well, and any truthful mechanism must satisfy $p_G(v) = v a_G(v) - \int_0^v a_G(w)dw$ (this also implies that the bidder's utility when truthfully reporting $(v,G)$ is $u_G(v) = \int_0^v a_G(w)dw$). This allows us to drop the payment variables, and follow Myerson's analysis.  Fiat et al. observe that many of the truthfulness constraints are redundant, and in fact it suffices to only make sure that when the bidder has (value, interest) pair $(v, G)$ they:
\begin{itemize}
\item Prefer to tell the truth rather than report any other $(v', G)$. This is accomplished by constraining $a_G(\cdot)$ to be monotone non-decreasing (exactly as in the single-item setting).
\item Prefer to tell the truth rather than report any other $(v, G' \in \children(G))$. By $N^+(G)$, we mean all items $G'$ such that $G' \succeq G$, but there does not exist a $G''$ with $G' \succeq G''\succeq G$. This is accomplished by constraining $\int_0^v a_G(w)dw \geq \int_0^v a_{G'}(w)dw$ (as the LHS denotes the utility of the buyer for reporting $(v,G)$ and the RHS denotes the utility of the buyer for reporting $(v, G')$).  Note that this is equivalent to saying that the area under $G$'s allocation curve should be at least as large \emph{at every $v$} as the area under $G'$'s allocation curve.
\end{itemize}
All of these constraints together imply that $(v,G)$ also does not prefer to report any other $(v', G')$.\footnote{
For example, if $(v,G)$ prefers truthful reporting to reporting $(v,G')$ where $G' \succ G$, and $(v,G')$ prefers truthful reporting to reporting $(v',G')$, then since $(v,G)$ gets the same utility for reporting $(v,G')$ as type $(v,G')$ does for truthfully reporting, $(v,G)$ prefers truthful reporting to reporting $(v',G')$.} We conclude this section with some standard definitions and observations. 
\begin{definition}[Revenue Curve]
	\label{def:revcurve}
	The \emph{revenue curve} for an item $G$ with CDF $F_G$ is a function $R_G$ that maps a value $v$ to the revenue obtained by posting a price of $v$, for a single item, when buyer values are drawn from the distribution $F_G.$  
	Formally,  $R_G (v) := v \cdot [1-F_G(v)]\cdot q_G$. We say that a revenue curve is feasible if there exists a distribution that induces it.   The \emph{monopoly reserve price} $r_G$ of the revenue curve is $r_G \in \argmax _p R_G(p)$.
\end{definition}

\begin{definition}[Virtual Value] \label{def:myevval} Myerson's virtual valuation function $\varphi_G(\cdot)$ is defined so that $\varphi_G(v):= v - \frac{1-F_G(v)}{f_G(v)}$. Observe that $R'_G(v) = 1 - F_G(v) - v f_G(v) = -\varphi_G(v)f_G(v)$. When clear from context we will omit the subindex $G$. 

\end{definition}

\begin{definition}[DMR] We say that a marginal distribution of values $F_G$ satisfies \emph{declining marginal revenues} (DMR) if $R_G (v)$ is concave, or equivalently, if $\varphi_G(v) f_G(v)$ is monotone non-decreasing. \end{definition}

When the marginal distributions do not all satisfy the DMR assumption, we instead need to iron the distribution, an analogue to Myersonian ironing.

\begin{definition}[Ironing]\label{def:ironprelim} The \emph{ironed revenue curve} denoted $\hat{R}(\cdot)$ for a revenue curve $R(\cdot)$ is the least concave upper bound on the revenue curve $R(\cdot)$.\footnote{We emphasize that this work irons the revenue curve with \emph{values} on the $x$-axis. Classical one-dimensional ironing (to yield Myersonian ironed virtual values) is done on the revenue curve with quantiles on the $x$-axis.}  A point $v$ is \emph{ironed} if $\hat{R}(v) \neq R(v)$.  We say that $[a,b]$ is an \emph{ironed interval} if $\hat{R}(a) = R(a)$, $\hat{R}(b) = R(b)$, and $\hat{R}(v) \neq R(v)$ for all $v \in (a,b)$, where if $v \in (a,b)$, then $a$ and $b$ are the lower and upper endpoints of the ironed interval, respectively.
\end{definition}

An ironed revenue curve is depicted in Figure~\ref{fig:ironedrevcurve}.  By the definition of concavity, if $z$ is ironed, then $\hat{R}(z) = \beta R(a) + (1-\beta) R(b)$ where $z \in (a,b)$, $\beta a + (1- \beta) b = z$, and $a,b$ are unironed. Importantly, observe that setting price $z$ to a consumer drawn from $F_G$ yields revenue $R(z) < \hat{R}(z)$. Yet, if we set price $a$ with probability $\beta$ and $b$ with probability $(1-\beta)$, we will get revenue $\beta R(a) +(1-\beta)R(b) = \hat{R}(z)$. One can check that this is precisely the allocation and payment
$$a(v) = \begin{cases} 0 & v < a \\ \beta & v \in [a, b) \\ 1 & v \geq b \end{cases} \quad\quad \text{and} \quad \quad p(v) = \begin{cases} 0 & v < a \\ \beta a & v \in [a, b) \\ \beta a + (1-\beta) b & v \geq b \end{cases}. $$

\begin{figure}[h!]
\centering
\includegraphics[scale=.42]{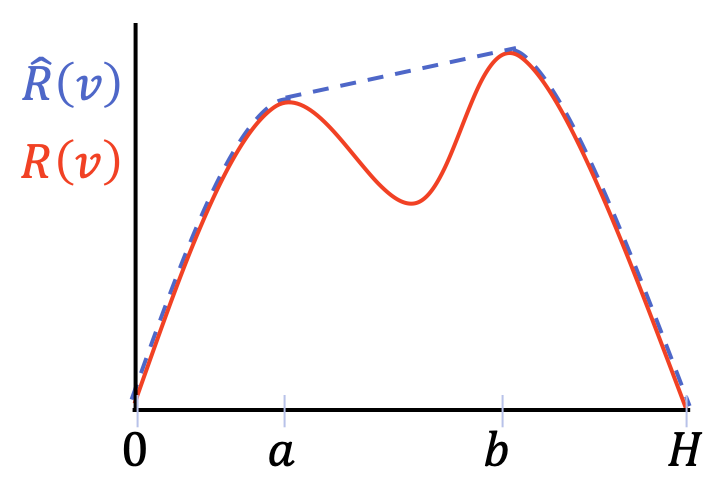}  
\caption{For some implicit distribution $F$, the revenue curve $R(v) = v \cdot [1-F(v)]$ is depicted, as is the ironed revenue curve, or the revenue curve's least concave upper bound. }
\label{fig:ironedrevcurve}
\end{figure}

\vspace{-.5cm}

\section{Duality}
\label{sec:highlevel}

In this section, we briefly overview the bare minimum duality preliminaries required.  Full duality preliminaries are provided in \Cref{subsec:dualderiv}-\ref{subsec:regvars}.


\subsection{Dual Terminology.} \label{sec:dualtermsapp}

In this section, we introduce pictorial representations (Figures~\ref{fig:dualvarsdef} and \ref{fig:dualvarsdefalpha})
of key aspects of a dual solution and define terminology relevant to the dual.

\begin{figure}[h!] 
	\centering
	\includegraphics[scale=.3]{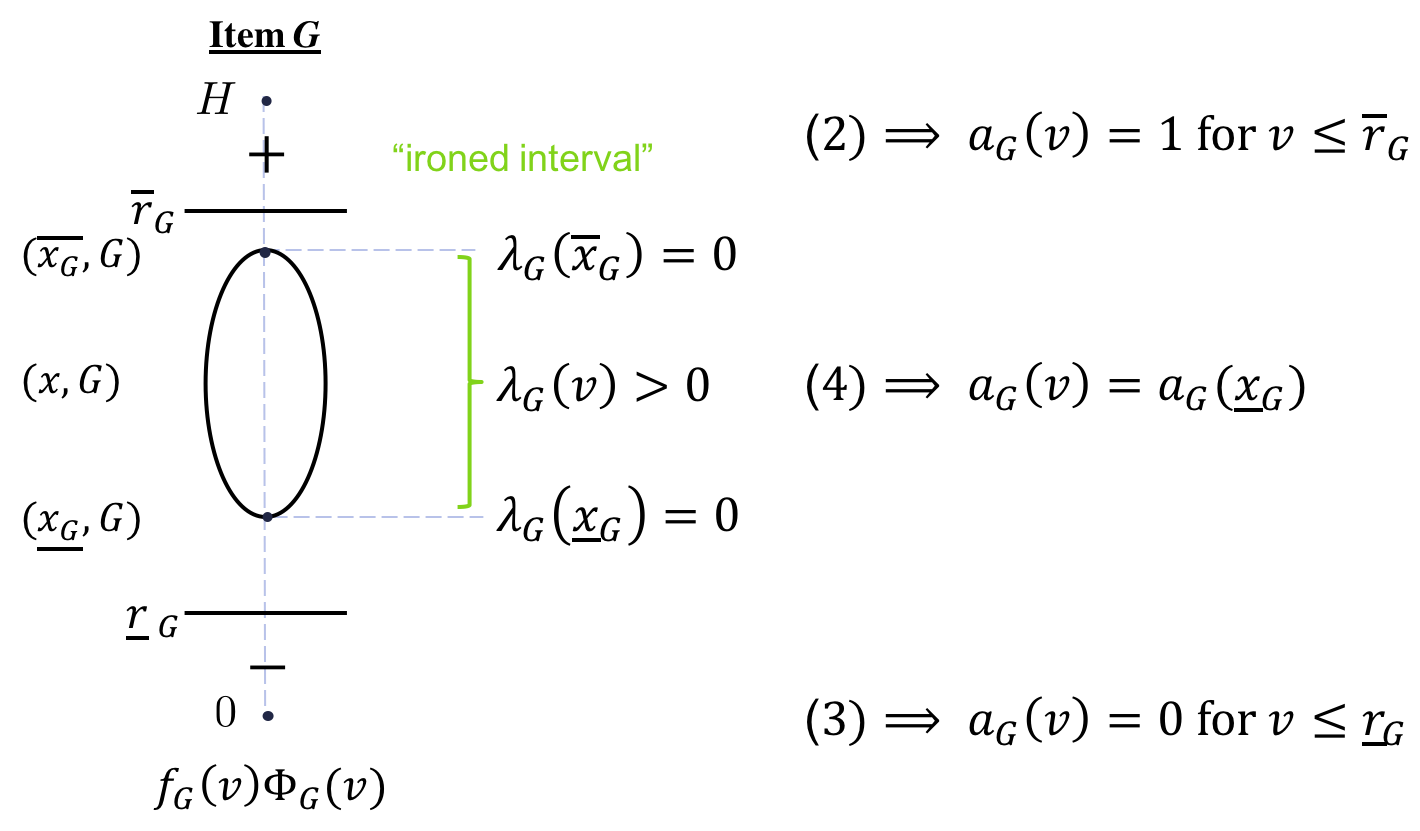} 
	\caption{A pictorial interpretation of virtual values $f_G(v) \Phi_G^{\lambda,\alpha}(v)$ and the dual variable $\lambda_G(v)$, in addition to the concepts of endpoints of the zero region, ironing, an ironed interval, and the allocation in response.}
	\label{fig:dualvarsdef} 
\end{figure}

The primal variables are $a_G(v)$ for all $G \in \G$, $v \in [0,H]$. 
Recall that we use $u_G(v) = \int_0^v a_G(w)dw$ to refer to the utility of $(v,G)$.
The dual variables are $\lambda_G(v)$, $\alpha_{G, G'}(v)$
 for all $G, G' \in \G$, and $v \in [0,H]$.
We first explain the role of these dual variables, 
 and then describe the Lagrangian relaxation 
 obtained using these dual variables. 
\paragraph{Dual Variable $\lambda$.}
The $\lambda$ dual variables correspond to incentive constraints   
 between types of the same interest but different value. 
This dual controls \emph{ironing}, as explained below. 
This really does correspond to ironing in the classical Myerson sense, only in value space. 
	
An oval (as depicted in Figure~\ref{fig:dualvarsdef}) represents an \emph{ironed interval}, a region where the dual variable $\lambda_G(\cdot)$ is non-zero.

\begin{itemize}
	\item (Ironing) \label{def:ironing} We say a type $(v,G)$ is \emph{ironed}, or that $v$ is ironed in item $G$, if $\lambda_G(v) > 0$.
	
	\item (Ironed Intervals) \label{def:ironedintervalendpoints} For any type $(x,G)$, the ironed interval containing $x$ in $G$ is defined by the bottom end point $\underline{x}_G = \sup \{v \leq x \mid \lambda_G(v) = 0\}$ and the top end point $\bar{x}_G = \inf \{v \geq x \mid \lambda_G(v) = 0\}$.  Then for all $v \in (\underline{x}_G, \bar{x}_G)$, type $(v,G)$ is ironed, $\bar{v}_G = \bar{x}_G$, and $\underline{v}_G = \underline{x}_G$.
	
\end{itemize}

As we will see later, dual best response (condition (\ref{CS3})) requires that if $\lambda_G(v) > 0$ then $a_G'(v) = 0$.  
In other words, the allocation rule $a_G$ must be constant over ironed intervals.  
For any value $x$, an optimal allocation must satisfy that $a_G(x) = a_G(\underline{x}_G)$.

\paragraph{Dual Variable $\alpha$.}
The $\alpha$ dual variables correspond to incentive constraints between types of the same value but different interest.}
\begin{figure}[h!] 
	\centering 
	\includegraphics[scale=.36]{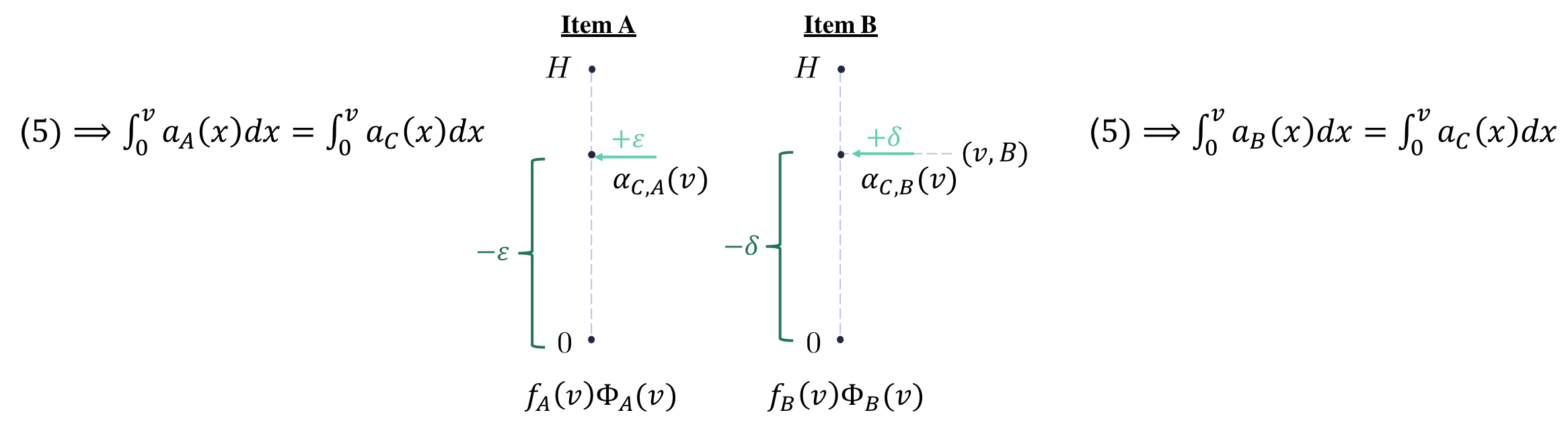} 
	\caption{A pictorial representation of the dual variable $\alpha$, in addition to the concepts of flow, preferable items, and equally preferable items.  Flow is assumed to be coming from item $C$.}  
	\label{fig:dualvarsdefalpha}
\end{figure}

In Figure~\ref{fig:dualvarsdefalpha}, a horizontal arrow into item $A$ (or $B$) at $v$ indicates that $\alpha_{C,A}(v)$ (or $\alpha_{C,B}(v)$) is non-zero.  We write the following statements for $G \in \{A,B\}$.

\begin{itemize}
	\item (Flow) \label{colloquial:flow} We will call the value of $\alpha_{G', G}(v)$ the ``flow into $(v,G)$'' or the ``flow into $G$ at $v$.''  When we focus on the minimal partial-order example, we infer that flow into $A$ or $B$ comes from $C$ in our figure.
	
\end{itemize}

Dual best response (condition (\ref{CS4})) requires that 
for $G \in \{A,B\}$, if $\alpha_{C,G}(v) > 0$ then $\int _0 ^v a_G(x) dx = \int _0 ^v a_C(x) dx$, or equivalently, 
$u_G(v) = u_C(v)$: a type with value $v$ should have the same utility in $C$ and $G$.  
Sending flow across interests forces the corresponding utilities to be the same.

\paragraph{Virtual Values.}
We will define a new variable, $\Phi^{\lambda, \alpha}(v)$ for all $v \in [0,H]$, 
and we will call the product 
$f(v) \Phi^{\lambda, \alpha} (v)$  
the \emph{virtual value}.\footnote{
	Whether we refer to $\Phi$ as the virtual value or $\Phi f$ reflects whether we iron in the quantile space or the value space. 
} 
Once again, this is a generalization of Myerson's virtual value function to this more general setting.  

Figure~\ref{fig:dualvarsdef} has a vertical axis ranging over values from $0$ (at the bottom) to $H$ (at the top), with a label of the item of focus $G$ at the top.  The point on the axis for any $v$ represents the virtual value $f_G(v) \Phi^{\lambda,\alpha}_G(v)$.

Of particular interest to us is the region where the virtual value is 0
because this is the region (and the only region) for which a primal satisfying complementary slackness can have a randomized allocation.  
This is an interval if $(f_G \Phi_G^{\lambda,\alpha})(\cdot)$ is monotone in $v$ (our solution ensures it is; details in Appendix~\ref{subsec:regvars}).
\begin{itemize}
	\item (Endpoints of Zero Region) \label{def:endpointsofzeroes} We define the bottom end point of the zero virtual value region in $G$ by $\underline{r}_G = \inf \{v \mid f_G(v) \Phi^{\lambda, \alpha}_G(v) \geq 0\}$ and the top end point $\bar{r}_G = \sup \{v \mid f_G(v) \Phi^{\lambda, \alpha}_G(v) \leq 0\}$.
\end{itemize}
In Figure \ref{fig:dualvarsdef} the horizontal black lines  and signs indicate where the virtual values shift from positive sign to zero, $\overline{r}_G$, and from zero to negative sign, $\underline{r}_G$.  Primal best response requires the allocation to satisfy $a_G(v) = 0$ for $v \leq \underline{r}_G$ (condition (\ref{CS1})) and $a_G(v) = 1$ for $v \geq \overline{r}_G$ (condition (\ref{CS2})).

\subsection{The Lagrangian Dual.} \label{subsec:lagrangian}

 The quality of a primal solution is measured by how well it solves the following Lagrangian relaxation induced by $(\lambda, \alpha)$. The quality of a dual solution is measured by the value of its induced Lagrangian relaxation. A dual is \emph{better} if the value of its induced Lagrangian relaxation is \emph{smaller}. 
\vspace{-1mm}
\begin{align*}
\noindent \text{Variables: }\quad\quad\quad &a_G(v) \quad \forall G \in \G,\ v\in [0,H]\\
      \text{Maximize } \quad\quad\quad    &\sum_{G \in \G} \int _0 ^{H} f_G(v) \cdot a_G(v) \cdot \Phi^{\lambda,\alpha}_G(v) dv \\
\text{subject to}\quad\quad\quad &a_G(v) \in [0,1] 
\end{align*}
\begin{multline} \label{def:vvalgen}
\text{ where } \quad\quad \varphi_G(v) = v - \frac{1-F_G(v)}{f_G(v)} \quad\quad  \text{ and where }  \quad\quad \Phi^{\lambda,\alpha}_G(v)  \\
:= \varphi_G(v) +\frac{1}{f_G(v)} \left[- \lambda'_G(v) + \sum _{G' \in \children(G)} \int _v ^H \alpha_{G,G'}(w) dw - \sum _{G': G \in \children(G')} \int _v ^H \alpha_{G',G}(w) dw \right].
\end{multline}

Before continuing, lets parse the Lagrangian relaxation. The only remaining constraints are that $a_G(v) \in [0,1]$, and the objective is a linear function of these variables. This immediately implies that the solution to this LP relaxation will set $a_G(v) = 1$ whenever $\Phi_G^{\lambda,\alpha}(v)> 0$, and  $a_G(v) = 0$ whenever $\Phi_G^{\lambda,\alpha}(v) < 0$.
This implies that if there is any randomization, i.e., $a_G(v) \in (0,1)$ then it must be that $\Phi_G^{\lambda,\alpha}(v) = 0$. 
The details of the definition of $\Phi$ are not so important here. 
(However, note that in the definition of $\Phi$, 
 the term $\lambda'$ refers to the derivative of $\lambda$.) 


\subsection{Complementary Slackness.}

Under strong duality, a (primal, dual) pair is optimal if and only if the primal and dual satisfy complementary slackness. In addition, if a dual $(\lambda ,\alpha)$ is optimal, i.e. satisfies complementary slackness with some primal, then \emph{any} primal is optimal if and only if it satisfies complementary slackness with $(\lambda, \alpha)$. Let's review complementary slackness in our setting. A primal $a$ and dual $(\lambda, \alpha)$ satisfy complementary slackness if and only if:\footnote{One can interpret these conditions as saying that the primal is an optimal solution to the Lagrangian relaxation, and the dual is the \emph{worst} possible Lagrangian relaxation for the primal.}

\begin{align}
\text{(Primal best response)}\quad\quad& \Phi^{\lambda, \alpha}_G(v) > 0 &&\Rightarrow \quad\quad a_G(v) =1 \label{CS1}\\
&\Phi^{\lambda, \alpha}_G(v) < 0 &&\Rightarrow \quad\quad  a_G(v) = 0 \label{CS2}\\ \nonumber\\
\text{(Dual best response)}\quad\quad &\lambda_G(v) > 0 &&\Rightarrow \quad\quad  a'_G(v) = 0 \label{CS3}\\
&\alpha_{G, G'}(v) > 0 &&\Rightarrow \quad\quad  \int_0^v a_G(x)dx - \int _0 ^v a_{G'}(x)dx = 0  \label{CS4}
\end{align}

That is, a primal is a best response to a dual if all $(v,G)$ with positive virtual value are awarded the item, and all $(v,G)$ with negative virtual value are not. A dual is a best response to a primal if whenever a dual variable is non-zero, the corresponding local IC constraint is tight. {The entire technical aspect of this paper is using the constraints imposed by complementary slackness in (\ref{CS1}-\ref{CS4}) to reason about optimal mechanisms and their menu complexity.}

\section{Menu Complexity}

We provide here the key ideas behind the construction that forms our lower bound and the proof of our upper bound.  Full details are provided in \Cref{app:LB} and \Cref{app:optvars} respectively.

\label{sec:menucomplexity}
\subsection{Menu Complexity is Unbounded: A Gadget and Candidate Instance} \label{sec:highlevelLB}
In this section, we provide a gadget that will be used in our menu complexity lower bound, and successively chain copies of it together to build our full construction. For one instance of our gadget, we provide a concrete potential dual, and prove that any allocation rule satisfying complementary slackness with it must have two distinct allocation probabilities. In order for this example to establish a menu complexity lower bound of two, we must additionally:
\begin{itemize}
\item Establish that there exists a distribution $F$ for which our dual is feasible. This is not covered in this section, and is deferred to our Master Theorem (Theorem~\ref{thm:multimaster}).
\item Establish that there exists an allocation rule which satisfies complementary slackness with this dual, thereby establishing that the dual is optimal (and any optimal allocation rule must satisfy complementary slackness with it). This is also not covered in this section, and is deferred to~\Cref{app:optvars}. 
\end{itemize}

We begin below with our gadget, then successively chain copies together to establish a menu complexity lower bound of $M$ for any $M > 0$. We recall the following facts established in the previous section:

\begin{enumerate}
\item A $+$ in any graphics at $(x,G)$ represents a strictly positive Virtual Value, which implies that $a_A(x) = 1$ in any allocation rule satisfying CS. A $-$ in any graphics at $(x,G)$ represents a strictly negative Virtual Value, which implies that $a_G(x) = 0$. (CS\ref{CS1}-\ref{CS2}) \label{pos}
\item A $\leftarrow$ in any graphics into $A$ at $x$ represents flow in. When there is flow into both $A$ and $B$ at the same point $x$, this implies that $u_A(x) = u_B(x)$. (CS\ref{CS4})  \label{alpha}
\item A point $x$ in the middle of an oval in any graphics represents that $x$ is contained in the interior of an ironed interval, and implies that $a(x) = a(y)$ where $y$ is the bottom of the oval. (CS\ref{CS3}) \label{ironed}
\end{enumerate}

\subsubsection{Step One: the base gadget and a lower bound of $M=2$.}
\label{sec:basecase}
\begin{figure}[h!] 
\centering 
\includegraphics[scale=.36]{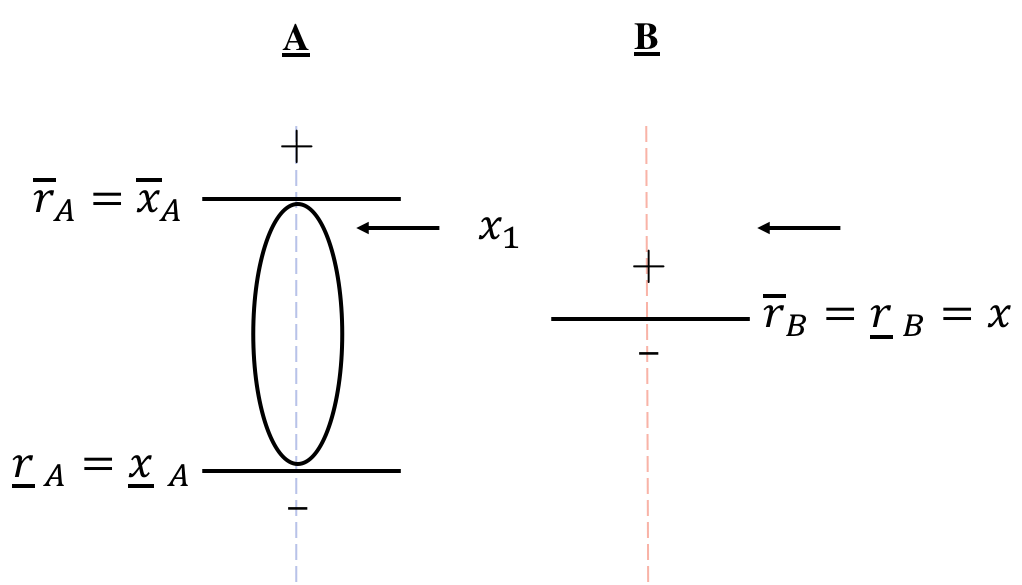} \hspace{1cm} \includegraphics[scale=.30]{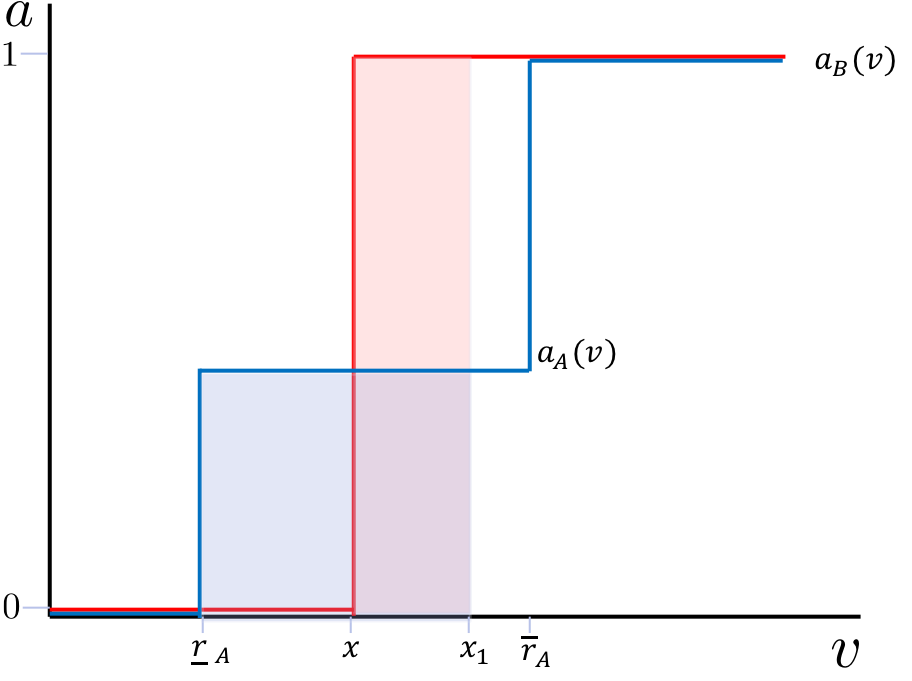} 
\caption{Left: Our first example that requires randomizing on $A$, containing an ironed interval $[\underline{r}_A, \overline{r}_A]$ (so $a_A(x_1) = a_A(\underline{r}_A)$) and flow into both $A$ and $B$ at $x_1$ (so $u_A(x_1) = u_B(x_1)$). 
Right: Primal best response dictates a price of $x$ for item $B$, while $A$'s allocation is 0 until $\underline{r}_A$ and $1$ after $\overline{r}_A$.  Equal preferability at $x_1$ forces $u_A(x_1)$ (the red area) equals $u_B(x_1)$ (the blue area); the ironed interval $[\underline{r}_A, \overline{r}_A]$ requires $a_A(\cdot)$ to be constant in this region, hence we must have $a_A(\underline{r}_A) \in  (0, 1)$.}
\centering
\label{fig:LBexample1}
\end{figure}

Our base case example is depicted in Figure~\ref{fig:LBexample1}.  We note each feature, and how it ties our hands with respect to the allocation rule via complementary slackness.
\begin{itemize}
\item In item $B$, there is a single point $x < \overline{r}_A$ for which $f_B(x) \Phi^{\lambda,\alpha}_B(x) = 0$.  That is, $\overline{r}_B = \underline{r}_B = x$.  Then (CS\ref{CS1}) implies that $a_B(v) = 1$ for $v > x$.
\item There is flow into both items $A$ and $B$ at $x_1 > x$.  That is, $\alpha_{C,A}(x_1), \alpha_{C,B}(x_1) > 0$.  (CS\ref{CS4}) implies that $A$ and $B$ must be equally preferable at $x_1$, that is, $\int _0 ^{x_1} a_A(w)dw = \int _0 ^{x_1} a_B(w)dw$.  Note that $a_B(w) > 0$ for $w \in (x, x_1]$, hence $\int _0 ^{x_1} a_B(w)dw > 0$.  Then to have $\int _0 ^{x_1} a_A(w)dw > 0$, because $a_A(\cdot)$ is monotone, it must be the case that $a_A(x_1) > 0$.
\item The point $x_1$ has $f_A(x_1) \Phi^{\lambda,\alpha}_A(x_1) = 0$ and is in an ironed interval $[\underline{r}_A, \overline{r}_A]$ where $\underline{r}_A < x$, that is, this ironed interval is the entire region of values that have virtual value zero in item $A$ and it contains both $x_1$ and $x$.  Because $x_1$ is in an ironed interval in $A$, then the allocation is constant, so $ a_A(\underline{r}_A) = a_A(x_1)$, which we have already established must be positive.
\item For whatever value that $a_A(\underline{r}_A)$ takes on, because $\underline{r}_A < x$, to satisfy equal preferability at $x_1$ (again, that $\int _0 ^{x_1} a_A(w)dw = \int _0 ^{x_1} a_B(w)dw$), we must have $a_B(x) > a_A(\underline{r}_A) (> 0)$, resulting in at least two distinct non-zero probabilities of allocation.
\end{itemize}

To complete the example, (1) there is no other flow: for all $v \neq x_1$, $\alpha_{C,A}(v) = \alpha_{C,B}(v) = 0$, and (2) item $C$ is unironed everywhere: $\lambda_C(v) = 0$ for all $v$. This base gadget forces randomization for the allocation of item $A$ because the utility of $x_1$ must be equal at $A$ and $B$, but the allocation of item $B$ must be zero below $x$, while the allocation of item $A$ must be non-zero.   

\subsubsection{Step Two: two chains and a lower bound of $M=3$.} Our second example (see Figure~\ref{fig:LBexample2}) contains the relevant features from the first example, but extends it to add an additional constraint: we replace the condition $\overline{r}_B = \underline{r}_B = x$ with an ironed interval $[\underline{r}_B, \overline{r}_B]$ where $\underline{r}_B < \underline{r}_A < \overline{r}_B < \overline{r}_A$. We claim that this example requires us to randomize on both items. Intuitively, this is because we now have two constraints on utilities that must be satisfied, so two degrees of freedom seems necessary.
\begin{itemize}
\item There is flow into both items $A$ and $B$ at $x_1 \in (\overline{r}_B, \overline{r}_A)$: $\alpha_{C,A}(x_1), \alpha_{C,B}(x_1) > 0$.  (CS\ref{CS1}) implies that $a_B(v) = 1$ for $v>\overline{r}_B$, so to satisfy equal preferability, we must have $a_A(x_1) > 0$.
\item The point $x_1$ has $f_A(x_1) \Phi^{\lambda,\alpha}_A(x_1) = 0$ and is in an ironed interval $[\underline{r}_A, \overline{r}_A]$ where $\underline{r}_A < \overline{r}_B$.
As $x_1$ is in an ironed interval in $A$, then the allocation is constant, so $a_A(\underline{r}_A) = a_A(x_1) > 0$.
\item There is flow into both items $A$ and $B$ at $x_2 \in (\underline{r}_A, \overline{r}_B)$: $\alpha_{C,A}(x_2), \alpha_{C,B}(x_2) > 0$.  Since $a_A(x_2) > 0$---it lies in the ironed interval in $A$, so $a_A(x_2) = a_A(\underline{r}_A$)---then to satisfy equal preferability at $x_2$, we must have $a_B(x_2) > 0$.
\item The point $x_2$ has $f_B(x_2) \Phi^{\lambda,\alpha}_B(x_2) = 0$ and is in an ironed interval $[\underline{r}_B, \overline{r}_B]$ where $\underline{r}_B < \underline{r}_A$.
As $x_2$ is in an ironed interval in $B$, then the allocation is constant, so $a_B(\underline{r}_B) = a_B(x_2) > 0 $.
\item For whatever value that $a_B(\underline{r}_B)$ takes on, because $\underline{r}_B < \underline{r}_A$, then to satisfy equal preferability at $x_2$ (that $\int _0 ^{x_2} a_A(w)dw = \int _0 ^{x_2} a_B(w)dw$), we must have $a_A(\underline{r}_A) > a_A(\underline{r}_B) (> 0)$.
\item For whatever value that $a_A(\underline{r}_A)$ takes on, because $\underline{r}_A < \overline{r}_B$, then to satisfy equal preferability at $x_1$ ($\int _0 ^{x_1} a_A(w)dw = \int _0 ^{x_1} a_B(w)dw$), we must have $a_B(\overline{r}_B) > a_A(\underline{r}_A) (> a_A(\underline{r}_B) > 0)$, resulting in at least \emph{three} distinct non-zero probabilities of allocation.
\end{itemize}
Again, (1) there is no other flow: for all $v \neq x_1, x_2$, $\alpha_{C,A}(v) = \alpha_{C,B}(v) = 0$, and (2) item $C$ is unironed everywhere: $\lambda_C(v) = 0$ for all $v$.

Observe that in both examples, we reason from where we have one item with positive virtual value and the other with virtual value zero downward that, in order to satisfy a number of equal preferability constraints, because ironed intervals force the allocation to be constant, then at every point, the allocation must be non-zero.  Then, we reason upward that, because the ironed intervals are interleaving between the items and never aligned, the allocation must strictly increase at each point of interest in order to satisfy equal preferability.  This is precisely the reasoning we will use to construct and prove an arbitrarily large instance and menu.

\begin{figure}
\centering
\includegraphics[scale=.36]{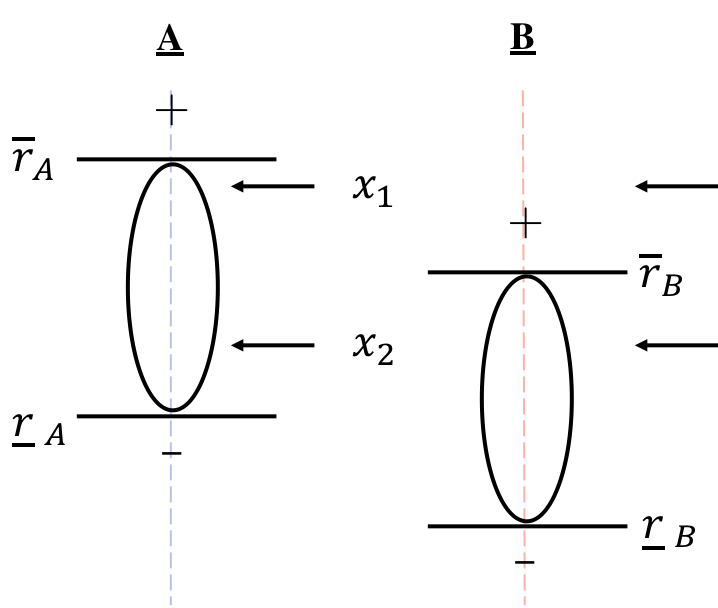} \hspace{1cm} \hspace{.25cm}\includegraphics[scale=.3]{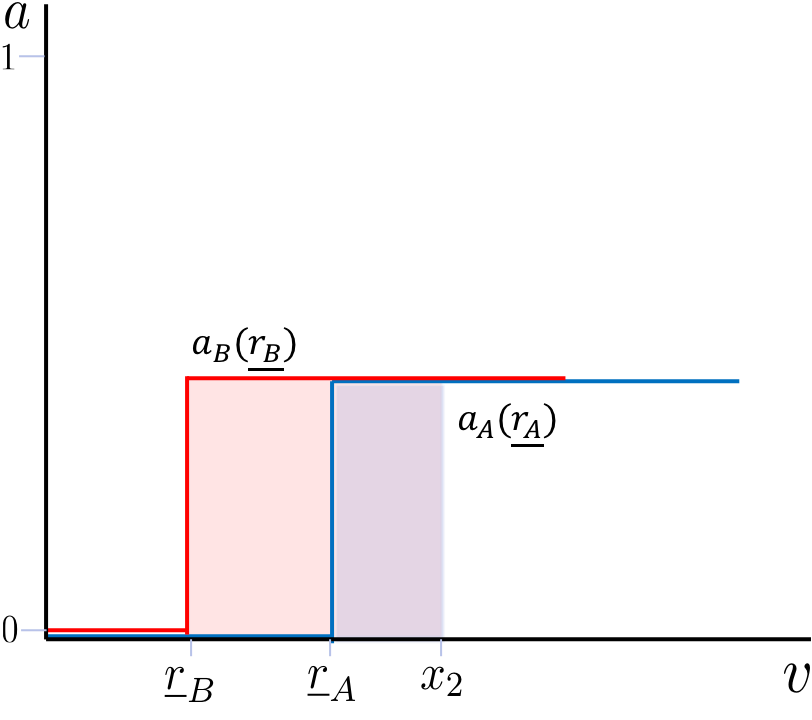} 
\caption{Left: Our second example, which requires randomization on both $A$ and $B$.  \newline Right: If $a_A(\underline{r}_A) \leq a_B(\underline{r}_B)$, then $u_A(x_2) < u_B(x_2)$ (the blue region is smaller than the red), which violates complementary slackness. }
\centering
\label{fig:LBexample2}
\end{figure}

\subsubsection{Step Three: four chains and a lower bound of $M=4$.} In this section, we take one more step towards our general construction. The first example presents our base gadget, and the second example chains two copies together. In this section, we simply confirm how the gadgets interact as we chain more and more together, bouncing back and forth from $A$ to $B$.

\begin{figure}[h!] 
\centering 
\includegraphics[scale=.47]{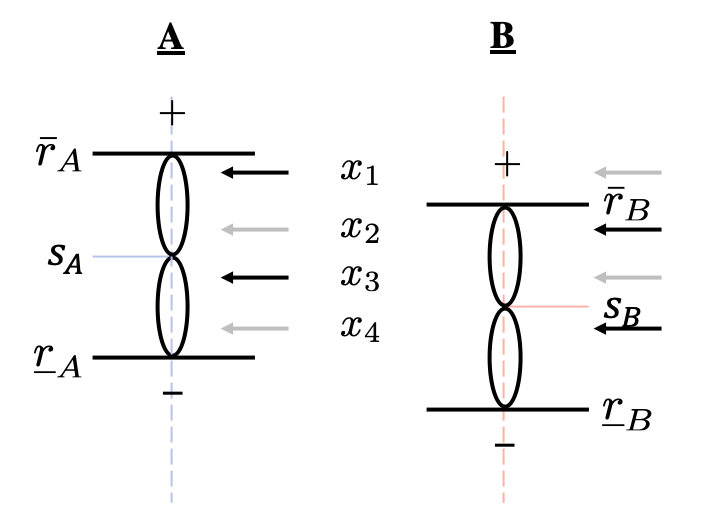} \hspace{1cm} \hspace{.25cm} \includegraphics[scale=.32]{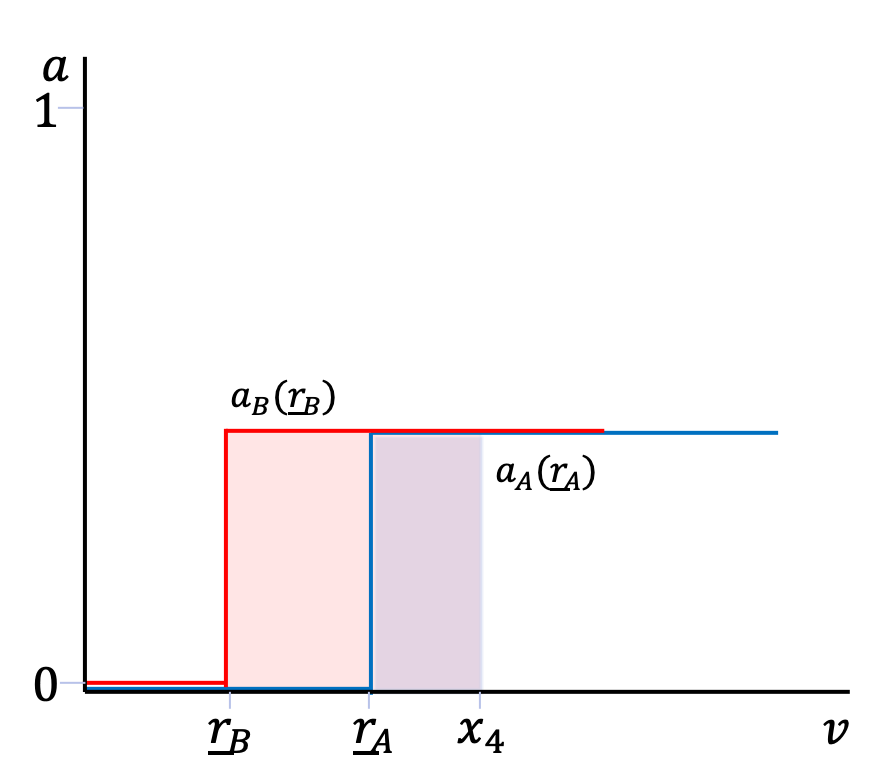} 
\caption{Left: An optimal dual for our example distributions, which will require at least 4 distinct allocation probabilities. 
\newline Right: If $a_A(\underline{r}_A) \leq a_B(\underline{r}_B)$, then $u_A(x_2) < u_B(x_2)$ (the blue region is smaller than the red), which violates complementary slackness via Fact~\ref{alpha} at $x_4$.}
\centering
\label{fig:chain-length-4}
\end{figure}

\paragraph{Nonzero allocation probabilities.} First, we see that the allocation at every ironed value $v$ such that $\Phi^{\lambda,\alpha}_G(v) = 0$ must be nonzero: $a_G(v) > 0$.  The argument holds for each of $(x_1, A), (x_2, B)$, $(x_3,A)$, and $(x_4,B)$. Below we iterate the same argument made in the two previous sections, skipping some details.
\begin{itemize}
\item Note that $a_B(x_1) > 0$ by Fact~\ref{pos}, and thus $u_B(x_1) > 0$.
\item By Fact~\ref{alpha}, $u_A(x_1) = u_B(x_1) > 0$.  Then $a_A(x_1) > 0$.
\item By Fact~\ref{ironed}, $a_A(s_A) = a_A(x_1) > 0$.  This also implies that $u_A(x_2) > 0$.
\item Now, again by Fact~\ref{alpha}, $u_B(x_2) = u_A(x_2) > 0$, so $a_B(x_2) > 0$.
\item Now, again by Fact~\ref{ironed}, $a_B(x_3) = a_B(x_2) > 0$, so $u_B(x_3) > 0$.
\item Again by Fact~\ref{alpha}, $u_A(x_3) = u_B(x_3) > 0$, so $a_A(x_3) > 0$.
\item By Fact~\ref{ironed}, $a_A(x_4) = a_A(x_3) > 0$, so $u_A(x_4) > 0$.
\item Finally by Fact~\ref{alpha}, $u_B(x_4) = u_A(x_4) > 0$.
\end{itemize}


Essentially, if any of these allocations must be positive, it forces the rest of them, working downwards, to be positive.  And, by Fact~\ref{pos}, $a_B(x_1) = 1$, so $u_B(x_1) > 0$.  Hence the rest of the implications follow, so the allocation must be nonzero throughout this region.

\paragraph{Distinct allocation probabilities.} Now, given that the allocation must be nonzero at every point in this range, we argue that it must be distinct at all of the points of interest.  Fix some nonzero $a_B(\underline{r}_B)$, and note by Fact~\ref{ironed} that $a_B(v) = a_B(\underline{r}_B)$ for all $v \in [\underline{r}_B, s_B]$.  By Fact~\ref{pos}, $a_G(v) = 0$ for $v < \underline{r}_G$.  Because $\underline{r}_B < \underline{r}_A$, then to have $u_A(x_4) = u_B(x_4)$, since $u_A(x_4) = \int _{\underline{r}_A} ^{x_4} a_A(w) dw = (x_4 - \underline{r}_A) a_A(\underline{r}_A)$ and $u_B(x_4) = \int _{\underline{r}_B} ^{x_4} a_B(w) dw = (x_4 - \underline{r}_B) a_B(\underline{r}_B)$, then we must have a distinct $a_A(\underline{r}_A) > a_B(\underline{r}_B)$.  This is depicted on the right side in Figure~\ref{fig:chain-length-4}.  Then, by Fact~\ref{ironed}, $a_A(x_3) = a_A(\underline{r}_A) > a_B(\underline{r}_B)$.

The argument extends inductively for $(x_3, B), (x_2, A)$, and $(x_1, B)$: we show it with $(x_3, B)$.  Note that $u_A(x_4) = u_B(x_4)$ and suppose the inductive hypothesis of $a_A(x_4) > a_B(x_4)$, where $a_A(x_3) = a_A(x_4)$ and $a_B(x_4) = a_B(\underline{r}_B)$ by Fact~\ref{ironed}.  Hence $u_A(s_B) > u_B(s_B)$.  Then in order to have $u_A(x_3) = u_B(x_3)$, we must have $a_B(x_3) > a_A(x_3)$.


The result is four distinct allocation probabilities in these four regions, and five in total (including the deterministic option to get the item w.p. one).  Essentially, this example only has two ironed intervals in $A$ and $B$ each with four points of interest.  Our full construction below lets the number of ironed intervals grow with $M$.

\subsubsection{Final Step: $M$ chains and a lower bound of $M$.}

\begin{figure}
\centering
\includegraphics[width=.45\linewidth]{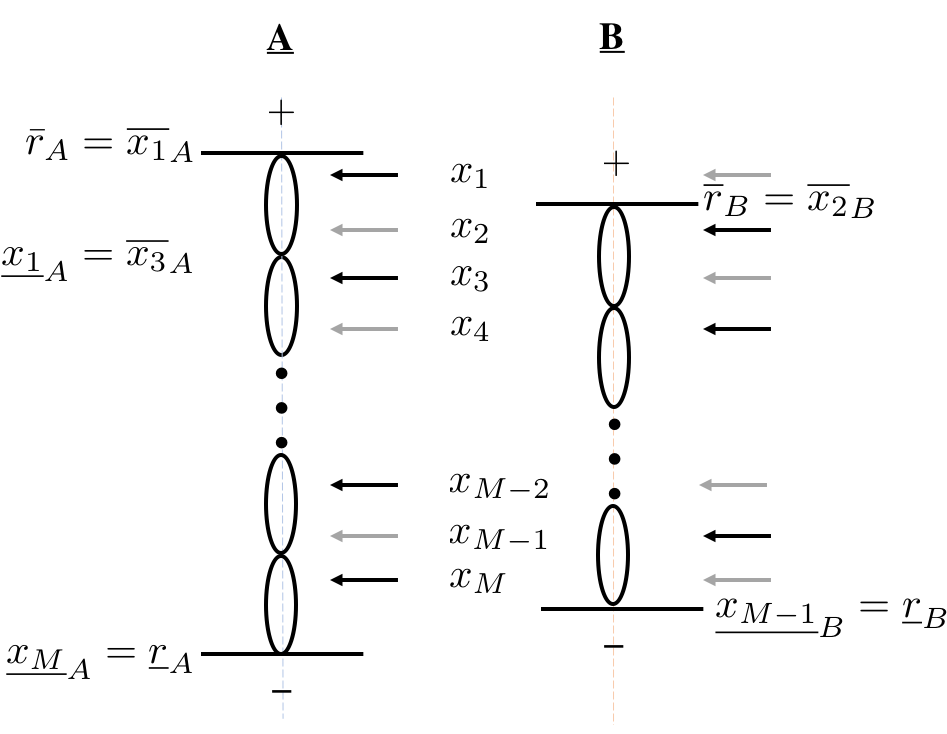}
\caption{Our candidate dual instance: a top chain that spans the entire region of zero virtual values for both $A$ and $B$ with no gaps between the ironed intervals that comprise the chain.  There is flow into $A$ and $B$ at every point $x_i$ in the chain.}
\label{fig:candidatedual8}
\centering
\end{figure}
It is possible to extend the examples above by continuing to interleave ironed intervals with flow coming in.  The combination of the equal preferability constraints and the inability to increase the allocation in the middle of an ironed interval is what requires us to randomize differently within each interval, forcing any number of menu options.  Details are given in \Cref{app:LB},
where we formally define this ``top chain" structure (Definition~\ref{def:chain}) and construct the candidate dual instance, which is depicted in Figure~\ref{fig:candidatedual8}. For example, our first example has a top chain of length one, the second of length two, and the third of length four. Theorem~\ref{thm:allocalg} proves that there exists a primal instance that satisfies  complementary slackness with the defined dual. This proves both that our dual is optimal, and thus {\em any} optimal primal must satisfy complementary slackness with it, giving us Theorem~\ref{thm:unbounded}.


\begin{theorem}\label{thm:unbounded}
Mechanisms that satisfy complementary slackness with a dual containing a top chain of length $M$ have menu complexity at least $M$. Moreover, for all $M$, there exists a distribution $F$ over three partially-ordered items for which a dual with top chain of length $M$ is feasible.
\end{theorem}

The ``Moreover, \ldots'' part of the theorem is due to our Master Theorem (Theorem~\ref{thm:multimaster}). The formal statement is a bit technical, and can be found in \Cref{sec:masterthmApp}.  

\subsection{For Three Items, Menu Complexity is Finite: Brief Highlight} \label{sec:minioptvar}

In \Cref{app:optvars}, we discuss our approach for characterizing the optimal mechanism for our 3-item minimal instance.  We prove essentially that the interleaving of ironed intervals used in the construction of the previous section is the worst case (in terms of menu complexity).  
We do this by specifying a subclass of optimal duals (that we call \emph{best duals}) using two new dual operations, \emph{double swaps} and \emph{upper swaps}.  We then leverage the structure of the best duals to give an algorithm that recovers the optimal primal from any best dual, and prove that the resulting mechanism has finite menu complexity.

\begin{theorem} \label{thm:allocalg} For any best dual solution, the primal recovery algorithm returns a primal with finite menu complexity that satisfies complementary slackness (and is therefore optimal). \end{theorem}


We conclude with one vignette regarding how the menu complexity can be unbounded but not infinite. Two crucial aspects of the ``top chain'' structure from our examples (generalized in Figure~\ref{fig:candidatedual8}) are that: (1) the ironed intervals for $A$ and $B$ are interleaving---this is what ``keeps the chain going'' and (2) the sequences for $A$ and $B$ terminate at \emph{different} bottom endpoints. The latter is a bit subtle, but the idea is that if the two chains terminate at the same bottom endpoint $v$, then this entire process can be aborted and simply setting $v$ as the reserve for all items satisfies complementary slackness. So while in principle, this top chain structure could indeed be countably infinite, it cannot also satisfy (1) and (2). This is because the monotone convergence theorem states that both chains do indeed converge to some bottom endpoint, and interleaving then guarantees that this bottom endpoint must be the same.



\subsection{One Last Example} \label{subsec:3nodmrex}

 In this section, we construct an example by applying the Master Theorem (Theorem~\ref{thm:multimaster}) to the dual in Figure~\ref{fig:chain-length-4}.  The customer prior distribution in the example consists of the marginal distributions depicted in Figure~\ref{fig:nondmrex}.  The distributions for $A$ and $B$ do not satisfy DMR, and, using the ideas from the previous subsections, we will see that the optimal mechanism is randomized.

\begin{figure}[h!] 
	\centering 
	\begin{subfigure} {.41\textwidth}
		\includegraphics[width=\textwidth]{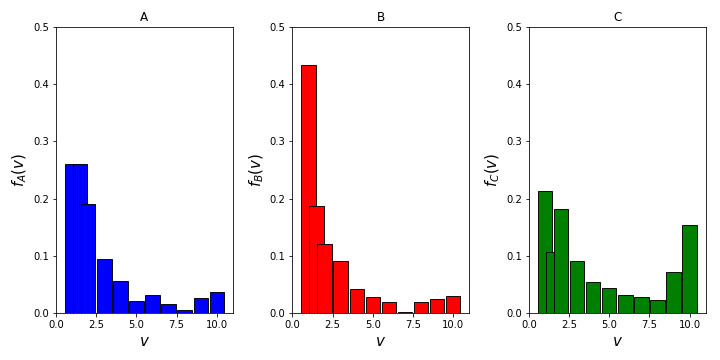}
		\caption{Probability densities}
		\label{fig:nondmrpdf} 
	\end{subfigure}
	\begin{subfigure}{.41\textwidth}
		\centering
		\includegraphics[width=\textwidth]{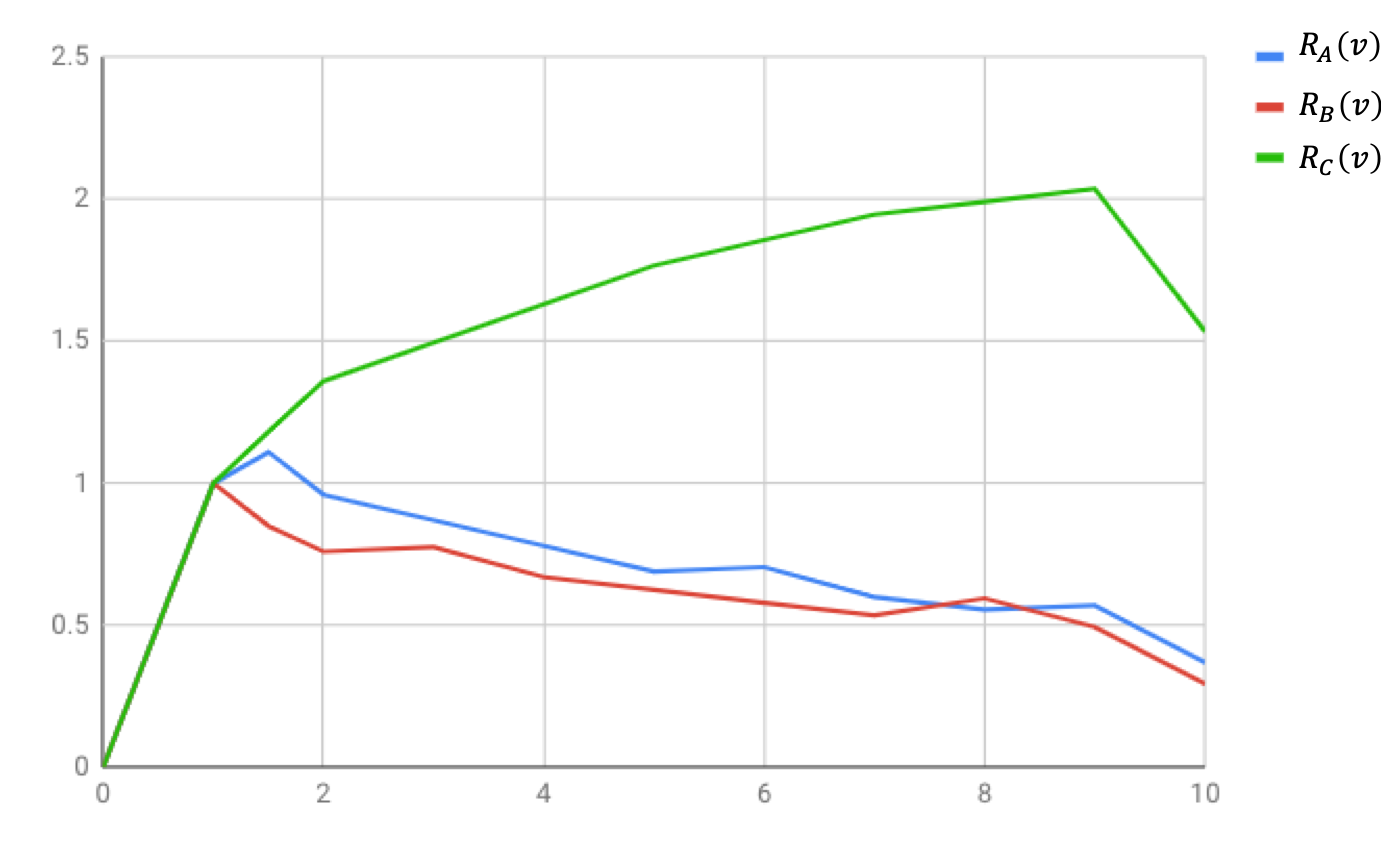}
		\caption{Revenue curves}
		\label{fig:nondmrrevcurves} 
	\end{subfigure}	
	\caption{The value distributions for items $A$, $B$, and $C$, that do not satisfy DMR.}
	\centering
	\label{fig:nondmrex}
\end{figure}

We can use the revenue curve procedure from \Cref{sec:example} to determine the optimal pricing for this example.  It produces the curves in Figure~\ref{fig:nondmrex}, telling us that 
the optimal price to set on item $C$ is $8$, which will result in prices of $9$ on item $A$ and 8 on item $B$. 
This gives $R_{ABC}(8) = 3.155$.  However, as we have seen in \Cref{sec:highlevelLB}, for the dual in Figure~\ref{fig:chain-length-4} (which corresponds to this distribution) to satisfy complementary slackness with a mechanism, the mechanism must have a good deal of randomization.

In \Cref{sec:highlevelLB}, we reasoned that the allocation probability must be distinct at each of the points $(x_1, A), (x_2, B)$, $(x_3,A)$, and $(x_4,B)$.  We also saw that if we fixed the allocation at $(x_4,B)$, there was only one way to satisfy the rest of the complementary slackness constraints, forming a system of equations.  The primal recovery algorithm described in the proof of \Cref{thm:allocalg} goes through solving this system of equations, ensuring that any other additional complementary slackness constraints are met, and that no pathological structures that might prevent a solution from existing can arise.  Applying this algorithm to our example results in the following optimal randomized mechanism:  
$$a_A(v) = \begin{cases} 0 & v < 1.5 \\ \frac{4}{7} & v \in [1.5, 6) \\ \frac{6}{7} & v \in [6, 10) \\ 1 & v \geq 10 \end{cases} \quad \quad a_B(v) = \begin{cases} 0 & v < 1 \\ \frac{2}{7} & v \in [1, 3) \\ \frac{5}{7} & v \in [3, 8) \\ 1 & v \geq 8 \end{cases} \quad\quad a_C(v) = \begin{cases} 0 & v < 1 \\ \frac{2}{7} & v \in [1, 2) \\ \frac{4}{7} & v \in [2, 5) \\ \frac{5}{7} & v \in [5, 7) \\ \frac{6}{7} & v \in [7, 9) \\ 1 & v \geq 9 \end{cases}.$$
The mechanism achieves a revenue of 3.2, which is slightly more than that of the best deterministic mechanism.

\section{Conclusions}
\label{sec:conclusion}
We study optimal mechanisms for single-minded bidders, and show that the menu complexity of optimal mechanisms is unbounded but finite for three items. Recall that for three identical items, the menu complexity is 1, for totally-ordered items the menu complexity is at most 7, and for heterogeneous items the menu complexity is uncountable. So our setting fits nicely ``in between'' totally-ordered and heterogeneous by this measure. By fuzzier measures of complexity, the same is true too: for identical items, the optimal mechanism has a clean closed-form description. For totally-ordered items, the optimal dual has a closed form, and the primal can be recovered by a simple algorithm as a function of this dual. For partially-ordered items, the optimal dual is unlikely to have a closed form, but can be characterized in terms of properties it must satisfy, and the primal can still be recovered algorithmically\footnote{Contrast this with \citep{CDW12a}, which only claims that by solving a linear program, an optimal mechanism for heterogenous settings can be found in time polynomial in the type space.} as a function of this dual. For heterogeneous items, optimal mechanisms are pure chaos.  And, like other settings that can be placed fundamentally in between single- and multi- dimensional settings (e.g., FedEx and MUP), we prove that the optimal mechanism is deterministic under DMR in the partially-ordered setting.


We also provide extensions---menu complexity of MUP (Theorem~\ref{thm:MUPunbounded}, Appendix~\ref{sec:MUPLB}) and of coordinated values (Theorems~\ref{thm:LB1}, \ref{thm:LB2}, \ref{thm:LB-gen}, \Cref{sec:raghuvansh})---proving the usefulness of our techniques beyond our setting.


Many interesting open directions remain.  First, general menu complexity upper bounds---for the single-minded setting, the Multi-Unit Pricing setting, and the coordinated valuations setting.  The techniques we use in this paper focus on characterizing the optimal dual and recovering the optimal mechanism for the three-item single-minded setting; this approach appears to be far too detailed and focused on characterizations to be extended.  We expect new ideas to be needed.

Second, the question of menu-complexity lower bounds for any of these three settings for \emph{approximately}-optimal mechanism are wide-open.  Is the separation from FedEx still as large when we only require approximately-optimal revenue?

Both directions of research would further fill out this rich spectrum, which until only recently was but thought to be a dichotomy between single-dimensional and heterogenous.

\bibliographystyle{plainnat}
\bibliography{dags}

\newpage
\appendix
\section{Full Preliminaries}
\label{sec:addtlprelims}

While this paper focuses on the three-item case, it's illustrative (and perhaps cleaner) to provide notation for general partially-ordered items. In general, there are $m$ partially-ordered items. Item $G$ can be better than, worse than, or incomparable to item $G'$, and we'll use the relation $G \succ G'$ to denote that $G$ is better than $G'$. We refer to the set of items as $\G$, and use $\children(G)$ to denote the set of items $G'\in \G$ for which $G' \succ G$, but there is no $G''$ with $G' \succ G'' \succ G$ (i.e. the items ``immediately better'' than $G$, or the 1-out-neighborhood of $G$ in a graphic representation). There is a single buyer with a (value, interest) pair $(v,G)$, who receives value $v$ if they are awarded an item $\succeq G$. An instance of the problem consists of a joint probability distribution over $[0,H]\times \G$, where $H$ is the maximum possible value of any bidder for any item. We will use $f$ to denote the density of this joint distribution, with $f_G(v)$ denoting the density at $(v, G)$. We will also use $F_G(v)$ to denote $\int_0^v f_G(w)dw$, and $q_G$ to denote the probability that the bidder's interest is $G$. 
  Note that $F_G(H) = q_G<1$, so $F_G(\cdot)$ is not the CDF of a distribution
(although $F_G(\cdot)/q_G$ is the CDF of the marginal distribution of $v$ conditioned on interest $G$). 

We'll consider (w.l.o.g.) direct truthful mechanisms, where the bidder reports a (value, interest) pair and is awarded a (possibly randomized) item. For a direct mechanism, we'll define $a_G(v)$ to be the probability that  item $G$ is awarded to a bidder who reports $(v,G)$, and $p_G(v)$ to be the expected payment charged. Then a buyer's utility for reporting any $(v', G')$ where $G'$ doesn't dominate $G$ is $- p_{G'}(v')$, and the utility for reporting any $(v', G')$ where $G'$ dominates $G$ is $v\cdot a_{G'}(v') - p_{G'}(v')$. 

At this point, one can write a primal LP that maximizes expected revenue subject to incentive constraints, manipulate the LP, and consider a Lagrangian relaxation (and all of this is done in~\citet{FGKK, DW}). 

\subsection{Formulating the Optimization Problem}


The ``default'' way to write the continuous LP characterizing the optimal mechanism would be to maximize $\sum_{G \in \G} \int_0^H f_G(v)p_G(v)dv$ (the expected revenue) such that everyone prefers to tell the truth than to report any other type.   As observed in~\citet{FGKK}, it is without loss of generality to only consider mechanisms that award bidders their declared item of interest with probability in $[0,1]$, and all other items with probability $0$.\footnote{To see this, observe that the bidder is just as happy to get nothing instead of an item that doesn't dominate their interest. See also that they are just as happy to get their interest item instead of any item that dominates it. It will also make this option no more attractive to any bidder considering misreporting. So starting from a truthful mechanism, modifying it to only award the item of declared interest or nothing cannot possibly violate truthfulness.}
Also observed in~\citet{FGKK} is that Myerson's payment identity holds in this setting as well, and any truthful mechanism must satisfy $p_G(v) = v a_G(v) - \int_0^v a_G(w)dw$ (this also implies that the bidder's utility when truthfully reporting $(v,G)$ is $u_G(v) = \int_0^v a_G(w)dw$). This allows us to drop the payment variables, and follow Myerson's analysis to recover:\footnote{For the familiar reader, this derivation is routine, so we omit it. The unfamiliar reader can refer to~\citep{Myerson,HartlineBook} for this derivation.}

$$\E[\text{revenue}] = \sum_{G \in \G}\int_0^H f_G(v)\cdot p_G(v)dv = \sum _{G \in \G} \int _0 ^H f_G(v) a_G(v) \left(v-\frac{1-F_G(v)}{f_G(v)}\right)dv$$

The experienced reader will notice that $v-\frac{1-F_G(v)}{f_G(v)}$ is exactly Myerson's virtual value for the conditional distribution $F_G(\cdot)/q_G$, which we'll denote by $\varphi_G(v)$. At this point, we still have a continuous LP with only allocation variables, but still lots of truthfulness constraints.  \citet{FGKK} observe that many of these constraints are redundant, and in fact it suffices to only make sure that when the bidder has (value, interest) pair $(v, G)$ they:
\begin{itemize}
\item Prefer to tell the truth rather than report any other $(v', G)$. This is accomplished by constraining $a_G(\cdot)$ to be monotone non-decreasing (exactly as in the single-item setting).
\item Prefer to tell the truth rather than report any other $(v, G' \in \children(G))$. This is accomplished by constraining $\int_0^v a_G(w)dw \geq \int_0^v a_{G'}(w)dw$ (as the LHS denotes the utility of the buyer for reporting $(v,G)$ and the RHS denote the utility of the buyer for reporting $(v, G')$). 
\end{itemize}
All of these constraints together imply that $(v,G)$ also does not prefer to report any other $(v', G')$.\footnote{
For example, if $(v,G)$ prefers truthful reporting to reporting $(v,G')$ where $G' \succ G$, and $(v,G')$ prefers truthful reporting to reporting $(v',G')$, then since $(v,G)$ gets the same utility for reporting $(v,G')$ as type $(v,G')$ does for truthfully reporting, $(v,G)$ prefers truthful reporting to reporting $(v',G')$.} Below, we will now formulate the Primal LP and its Lagrangian relaxation. This derivation is not a new result, but important to understanding our approach. So we'll go through some of the steps to help provide some intuition for the reader, but omit any calculations and proofs.

\subsection{The Primal} \label{subsec:primal}
With the above discussion in mind, we can now formulate our primal continuous LP.

\begin{align*}
\noindent \text{Variables: }\quad\quad\quad &a_G(v),\ \forall G \in \G,\ v\in [0,H]\\
      \text{Maximize } \quad\quad\quad    &\sum_{G \in \G} \int _0 ^{H}  f_G(v) a_G(v) \varphi_G(v) dv\\ \\
          \text{subject to}\quad\quad\quad  &a_G'(v) \geq 0  &&\forall G \in \G\ \forall v \in [0,H] \ \mbox{\rm (dual variables $\lambda_G(v)\geq 0$)} \\            
        \int_0^v &a_{G}(x) dx - \int_0^v a_{G'}(x) dx  \geq 0   &&\forall G \in \G,\ G' \in \children(G)\ \forall v \in [0,H]  \ \mbox{\rm (dual vars $\alpha_{G,G'}(v)\geq 0$)} \\
&a_G(v) \in [0,1] &&\forall G \in \G,\ \forall v \in [0,H]\ \mbox{\rm (no dual variables)}
\end{align*}

The first constraint requires that $a_G(\cdot)$ is monotone non-decreasing for all $G$. If an allocation rule is not monotone, it cannot possibly be part of a truthful mechanism. As discussed above, Myerson's payment identity combined with monotonicity guarantees that $(v,G)$ will always prefer to report $(v,G)$ instead of $(v', G)$. The second constraint directly requires that the utility of $(v,G)$ for reporting $(v,G)$ is at least as high as for reporting $(v,G')$ (also discussed above). The final constraint simply ensures that the allocation probabilities lie in $[0,1]$.

\subsection{Derivation of the Partial Lagrangian Dual} \label{subsec:dualderiv}

Moving the first two types of constraints from the primal to the objective function with multipliers $\lambda_G(v)$ and $\alpha_{G,G'}(v)$ respectively gives the partial Lagrangian primal:

$$\max _{a: a_G(v) \in [0,1] \, \forall G \in \G, \forall v \in [0,H]} \min _{\lambda, \alpha \geq 0}  \L(a; \lambda, \alpha)$$
where
\begin{multline} \L(a; \lambda, \alpha) := \\ \sum_{G \in \G} \int _0 ^{H}  \left[ f_G(v) a_G(v) \varphi_G(v) + \sum_{G' \in N^+(G)} \alpha_{G,G'}(v) \cdot \left[ \int_0 ^v a_G(x) dx - \int _0 ^v a_{G'}(x)dx \right] + \lambda_G(v) a'_G(v) \right] dv. \nonumber \end{multline}

This gives the corresponding partial Lagrangian dual of 
$$\min _{\lambda, \alpha \geq 0} \max _{a: a_G(v) \in [0,1]\, \forall G \in \G, \forall v \in [0,H]} \L(a; \lambda, \alpha).$$
Note however that we can rewrite $\L(a; \lambda, \alpha)$ by using integration by parts on the $a'_G(v)$ term to get $a_G(v)$ terms, using that $a_G(0) = 0$ and $\lambda_G(H) = 0$ without loss:
$$\int _0 ^H \lambda_G(v) a'_G(v) dv = \lambda_G(v) a_G(v) \mid _0 ^H - \int _0 ^H \lambda'_G(v) a_G(v) dv = - \int _0 ^H \lambda'_G(v) a_G(v) dv$$
As in \citepalias{FGKK}, this uses the facts that $\lambda_G(\cdot)$ is continuous and equal to 0 at any point that $a'_G(v) = \infty$, which occurs at only countably many points.  Then, collecting the $a_G(v)$ terms gives:
\begin{align*}
\L(a; \lambda, \alpha) &=  \sum_{G \in \G} \int _0 ^{H}  \bigg[ f_G(v) a_G(v) \varphi_G(v) \\
& \quad + \sum_{G' \in N^+(G)} \alpha_{G,G'}(v) \cdot \left[ \int_0 ^v a_G(x) dx - \int _0 ^v a_{G'}(x)dx \right] - \lambda'_G(v) a_G(v) \bigg] dv \\
&= \sum_{G \in \G} \int _0 ^{H}  f_G(v) a_G(v) \Phi^{\lambda, \alpha}_G(v) dv
\end{align*}
where we define
$$\Phi^{\lambda, \alpha}_G(v) := \varphi_G(v) + \frac{1}{f_G(v)} \cdot \left[ \sum_{G' \in N^+(G)} \int_v ^H \alpha_{G,G'}(x) dx - \sum_{G': G \in N^+(G')}  \int _v ^H \alpha_{G',G}(v)dx \right] - \frac{1}{f_G(v)} \lambda'_G(v).$$
Then we can write that the Lagrangian dual problem is
$$\min _{\lambda, \alpha \geq 0} \quad \max _{a: a_G(v) \in [0,1] \, \forall G \in \G, \forall v \in [0,H]} \quad \sum_{G \in \G} \int _0 ^{H}  f_G(v) a_G(v) \Phi^{\lambda, \alpha}_G(v) dv.$$

\subsection{More Dual Terminology} \label{sec:moredualtermsapp}


Minimal dual terminology is first introduced in subsection~\ref{sec:dualtermsapp}.  Here, we add a few additional terms.

Dual best response (condition (\ref{CS4}))  implies the following.

\begin{itemize}
\item (Preferable Items) To satisfy complementary slackness, for any $x$ such that $\alpha_{G,G'}(x) > 0$, we must have $\u_{G'}(x) \geq \u_{G''}(x) \quad \forall G'' \in \children(G).$
This is because (a) $\u_G(x) = \u_{G'}(x)$ by complementary slackness and (b) $\u_{G}(x) \geq \u_{G''}(x) \quad \forall G'' \in \children(G)$ by incentive compatibility.

\item (Equally Preferable Items) \label{fact:equallypreferablegroups} By the above, to satisfy complementary slackness with any dual with $\alpha_{G,G'}(x) > 0$ and $\alpha_{G,G''}(x) > 0$, we must have $u_{G'}(x) = u_{G''}(x)$.

\end{itemize}


 \subsection{Review of Dual Properties} \label{subsec:regvars}
 
\begin{itemize}
\item (Rerouting Flow Among $\children(G)$) \label{fact:rerourtingflow} If $G', G'' \in \children(G)$ and we decrease $\alpha_{G,G'}(v)$ by $\vareps$ and increase $\alpha_{G,G''}(v)$ by $\vareps$, then $v' \leq v$, $f_{G'}(v') \Phi^{\lambda,\alpha}_{G'}(v')$ decreases by $\vareps$ and $f_{G''}(v') \Phi^{\lambda,\alpha}_{G''}(v')$ increases by $\vareps$.  All other virtual values, including all of those within $G$, remain the same.
\item (Utility based on the dual) We can often simplify how utility is written in terms of the dual and complementary slackness constraints. If $\underline{x}_G < x < y < \overline{x}_G$, then allocation in ironed intervals implies $u_G(y) = u_G(x) + a_G(y) (y-x)$.
\item (Allocation to Nonzero Virtual Values) As shown above in Subection~\ref{subsec:lagrangian}, the dual variables (1) determine the virtual welfare functions $\Phi^{\lambda,\alpha}(\cdot)$ and (2) are chosen to minimize the maximum virtual welfare under $\Phi^{\lambda,\alpha}(\cdot)$.  For an optimal dual solution, the optimal mechanism will simply be the corresponding virtual welfare maximizer that satisfies complementary slackness.  Parts of this mechanism are easy to predict if the virtual value functions are sign-monotone, which we will later ensure that they are.  Assuming this, we can talk about the virtual values in terms of three regions: positives, negatives, and zeroes.  
\item (Ironing and Proper Monotonicity.) We say that a dual satisfies \emph{proper monotonicity} if $f_G\cdot\Phi_G^{\lambda,\alpha}(\cdot)$ is monotone non-decreasing (note the multiplier of $f_G$). As shown in \citepalias{FGKK,DW}, for all $\alpha$, there exists a $\lambda$ such that $(\lambda, \alpha)$ is properly monotone.
\item (Boosting can only improve the dual.)  Given any dual with properly monotone virtual values, if there exists $v$ such that $f_G(v) \Phi^{\lambda,\alpha}_G(v) < 0$, then for any $G' \in \children(G)$, incrementing $\alpha_{G,G'}(v)$ by $f_G(v) \Phi^{\lambda,\alpha}_G(v)$ only improves the dual. By proper monotonicity, for all $v' \leq v$, $f_G(v') \Phi^{\lambda,\alpha}_G(v') < f_G(v) \Phi^{\lambda,\alpha}_G(v) < 0$, hence increasing $\alpha_{G,G'}(v)$ will not create any positives within $G$, not hurting the dual objective. Sending flow into an item $G'$ can only help by making positives less so, and does not increase any virtual values (but it's possible that it doesn't strictly help). This operation is coined \emph{boosting} in \citepalias{DW}.  While it is clear that $G$ should send the flow, the remaining question is \emph{which} $G' \in \children(G)$ should the flow be sent to. This is the bulk of our analysis.
\item By sign monotonicity, $v > \bar{r}_G$ has a positive virtual value, and thus the allocation rule must set $a_G(v) = 1$, otherwise it is not maximizing virtual welfare.  
\item Similarly, for values with negative virtual values, that is, $v < \underline{r}_G$, it must be that $a_G(v) = 0$. 
\end{itemize}

From these observations, we can conclude that the flow out of $C$ is identical to the flow out of the root node (day $n$) in the FedEx solution.  That is, 
$$\alpha_{C,A}(v) + \alpha_{C,B}(v) = \begin{cases}0 & v > \bar{r}_C \\ -\hat{R}''_C(v)/f_C(v) & v \leq \bar{r}_C. \end{cases}$$
where $R_C(\cdot)$ is defined as in Definition~\ref{def:revcurve}, $\hat{R}_C(\cdot)$ is the least concave upper bound on $R_C(\cdot)$, and $\hat{R}''_C(\cdot)$ is the second derivative of this function with respect to $v$.

We conclude with a fundamental result from~\citepalias{FGKK}.

\begin{theorem}[Proper Ironing \citepalias{FGKK}]Given all dual variables $\alpha$, suppose $\lambda_G(v) = 0$ for all $(v,G)$.  Then $f_G(v) \Phi^{\lambda,\alpha}_G(v)$ is defined for all $(v,G)$.  We define $\Gamma_G(v) = - \int _0 ^v f_G(x) \Phi^{\lambda,\alpha}_G(x) dx$, and $\hat\Gamma_G(\cdot)$ is the least concave upper bound on this function.  Then setting $\lambda_G(v) = \hat\Gamma_G(v) - \Gamma_G(v)$ defines a continuous and differentiable $\lambda_G(\cdot)$ that, with the update of $\Phi^{\lambda,\alpha}_G(\cdot)$ based on $\lambda_G(\cdot)$, results in the proper monotonicity of $f_G(\cdot) \Phi^{\lambda,\alpha}_G(\cdot)$. \end{theorem}

\section{Three Illustrative Examples}
\label{sec:example}


In this section, we use three example instances to understand how the optimal mechanisms become increasingly complex, blowing up from deterministic prices to unbounded randomization. We begin with some intuition before diving into examples.

\paragraph{Intuition: Why is single-minded more complex?} 
Consider first a one-item setting that only sells 2-day shipping.  Myerson's seminal work proves that the optimal way to sell 2-day shipping in isolation is to post the monopoly reserve price for it. Consider next retroactively adding 1-day shipping into the mix, perhaps because some customers demand 1-day shipping and aren't satisfied with 2-day shipping. Perhaps the distribution of customers demanding 1-day shipping has a higher Myerson reserve than the initial 2-day shipping distribution, in which case it is consistent to set both optimal reserves. Note, however, that a customer who wants their package within 2 days would be content with 1-day shipping. So if instead the 1-day shipping distribution has a \emph{lower} Myerson reserve than 2-day shipping, posting the pair of Myerson reserves is no longer incentive compatible. This complexity arises in the FedEx problem~\cite{FGKK}, and requires considering the constraints imposed on 2-day shipping by 1-day shipping (or vice versa).

Now consider the simplest single-minded valuation setting.  The internet service provider (ISP) sells three options: wifi, wifi/cable, and wifi/phone, where wifi/cable and wifi/phone dominate wifi but are incomparable with each other.
If it happens to be that the distribution of consumers who are interested in wifi/cable or wifi/phone both have a higher Myerson reserve than the distribution of consumers who are interested in only wifi,\footnote{Recall that a one-dimensional distribution $D$ can stochastically dominate $D'$ yet have a lower Myerson reserve. For example, if $D$ is uniform over the set $\{1,10\}$, the Myerson reserve is $10$. If $D$ is uniform over the set $\{9,10\}$, the Myerson reserve is $9$.} then again the seller can simply offer all three options at their Myerson reserve. However, if this is not the case, further optimization must be done. Importantly, in contrast to the FedEx setting, there's a circular dependency involving these three options which doesn't arise in the totally-ordered case (see examples for further detail). In this way, the IC constraints that govern the mechanism are much more complex in the single-minded setting than in the FedEx setting, and are the reason both for developing much richer techniques and for the much higher degree of randomization that is seen in our results.\\
\\
\indent Now, we explain what the optimal mechanism looks like for (1) the minimal partially-ordered (single-minded) instance under DMR, (2) the minimal totally-ordered (FedEx) instance without DMR, and (3) the minimal partially-ordered instance without DMR.

\paragraph{Three Partially-Ordered Items under DMR.} We begin with the special case where the marginal distributions for each item satisfy DMR.  Recall that this implies that the marginal revenue curves for each item are concave, and thus do not require ironing. We show how to derive the optimal item pricing (but a proof that this is indeed optimal is deferred to 
\Cref{app:DMR} as part of the general DMR case).  Our instance is again that where $C$ is the worst item (e.g. wifi) and $A$ and $B$ are incomparable (e.g. wifi/cable and wifi/phone).

Let's start by considering what price we would set for item $A$ if we had already set price $p_C$ for item $C$. (Note that whatever price we set for item $B$ has no effect, as $A$ and $B$ are incomparable.) Observe that our revenue from setting any price $p_A$ is just $p_A\cdot [1-F_A(p_A)]$, so ideally we would just set price $r_A:= \arg\max_p \{p \cdot [1-F_A(p)]\}$. If $r_A \geq p_C$, this doesn't violate any IC constraints. Indeed, consumers with interest $C$ will prefer to pay $p_C \leq r_A$ to get item $C$ rather than item $A$. If $r_A < p_C$, however, setting price $r_A$ will violate IC, as now consumers with interest $C$ would strictly prefer to report interest in item $A$ instead. This constrains us to set a price for $A$ that is at least $p_C$.  Observe that, because $R_A(\cdot)$ is concave, the revenue-maximizing price to set that is at least $p_C$ (which is $>r_A$)  is $p_A:= p_C$. Hence, we can define the revenue curve $\bar{R}_A(\cdot)$ to describe the revenue we can get from selling item $A$ as a function of $p_C$: $$\bar{R}_A(p_C) = \begin{cases} R_A(r_A) & p_C \leq r_A \\ R_A(p_C) & p_C > r_A \end{cases}.$$  

The same definition holds for $\bar{R}_B(\cdot)$. Now, we can find the price to set for item $C$ that optimizes the impact on all three items by simply finding the $p$ maximizing $R_{ABC}(p):= R_C(p) + \bar{R}_A(p) + \bar{R}_B(p)$ (depicted in Figure~\ref{fig:dmrcurves}). Picking $p_C$ as such, and then setting $p_A:= \max\{r_A, p_C\}$, $p_B:= \max\{r_B, p_C\}$ is the optimal pricing. The (challenging) remaining step is to prove that in fact this is optimal even among randomized mechanisms. The duality theory previously hinted at is key in this step, but we postpone these details for now. Importantly, note that this claim requires the DMR assumption (so proving it will certainly be technically involved)---without it, there might be a better randomized mechanism.

\vspace{-.5cm}

\begin{figure}[h!]
\centering
\includegraphics[scale=.28]{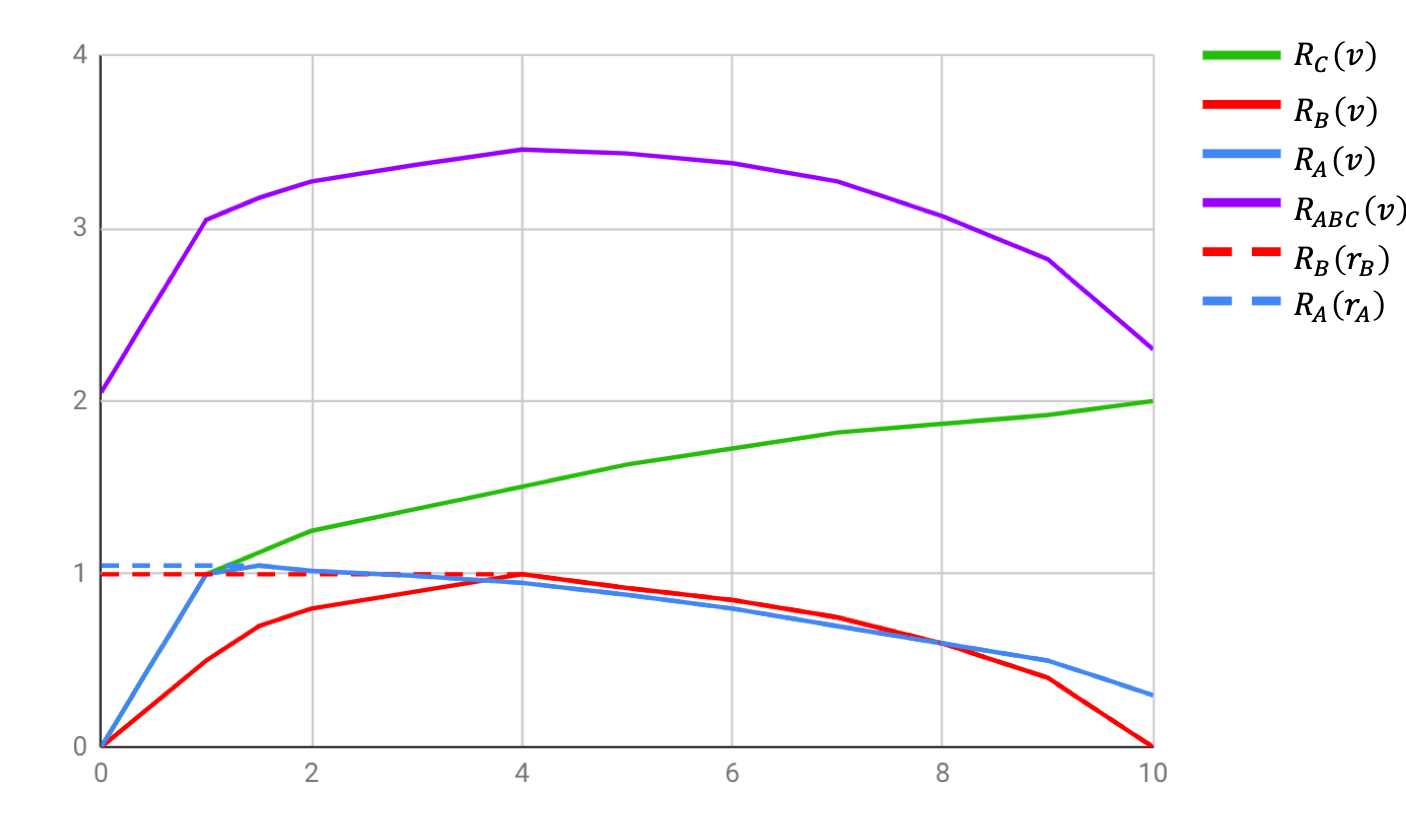} 
\caption{The construction of $R_{ABC}$ with $R_A, \bar{R}_A, R_B, \bar{R}_B,$ and $R_C$ illustrated as well.}
\label{fig:dmrcurves}
\end{figure}

\vspace{-.5cm}

\paragraph{Two Items without DMR (FedEx).} In this example, there are only two items, $A$ and $C$ with $A \succ C$. In this case, we'll think about first setting the price for $A$, and understanding how it constrains our choices for $C$. If we set price $p_A$ for item $A$, then we are constrained to give every type $(v,C)$ interested in item $C$ utility at least $v-p_A$. Again, if $r_C \leq p_A$, we should just set price $r_C$ on item $C$. However, if $r_C > p_A$, \emph{without the DMR assumption}, it's unclear what the best price to set should be. Indeed, it could be that some price $p_C \ll p_A$ generates more revenue than $p_A$ as $R_C(\cdot)$ is not necessarily concave. Note, however, that the ironed revenue curve $\hat{R}_C(\cdot)$ \emph{is} concave. So $\arg\max_{p_C \leq p_A}\{\hat{R}_C(p_C)\} = \min\{r_C, p_A\}$. It's unclear exactly what to make of this, but one hope (that turns out to be correct), is that the optimal scheme for item $C$, conditioned on $p_A$, is to set \emph{expected price} $p_C := \min \{r_C, p_A\}$ via the allocation rule defined as in Definition~\ref{def:ironprelim}. It is not obvious that such an allocation rule satisfies IC, but straight-forward calculations confirm that indeed it does. Similarly to the previous example, we can now define:
$$\bar{R}_C(p_A) = \begin{cases} \hat{R}_C(p_A) & p_A < r_C \\ R_C(r_C) & p_A \geq r_C \end{cases} \quad\quad \text{and} \quad\quad R_{AC}(p_A) = R_A(p_A) + \bar{R}_C(p_A).$$ 
This construction is depicted in Figure~\ref{fig:fedexcurves}. Figure~\ref{fig:splitting} gives some intuition as to why it is indeed incentive compatible to set the proposed allocation rule for item $C$ (but the goal of this section is not to provide complete proofs). It is now clear that, among all options which set a deterministic price for item $A$, and implement an expected price on the ironed revenue curve for item $C$, the above procedure is optimal. What is not clear is why this procedure is optimal over all possible menus for item $C$, or even why a randomized menu for item $A$ can't perform better. Indeed, the same duality theory referenced previously takes care of this.

This example perhaps also gives intuition for the menu complexity upper bound of $2^{m}-1$ for FedEx.  Repeating this process for another totally-ordered item, each option offered to buyers with interest $C$ could be ``split'' into at most two new options to be offered to buyers with interest $D \prec C$.

\begin{figure}[h!]
\begin{minipage}[t]{0.4\textwidth}
\centering
\includegraphics[width=.9\linewidth]{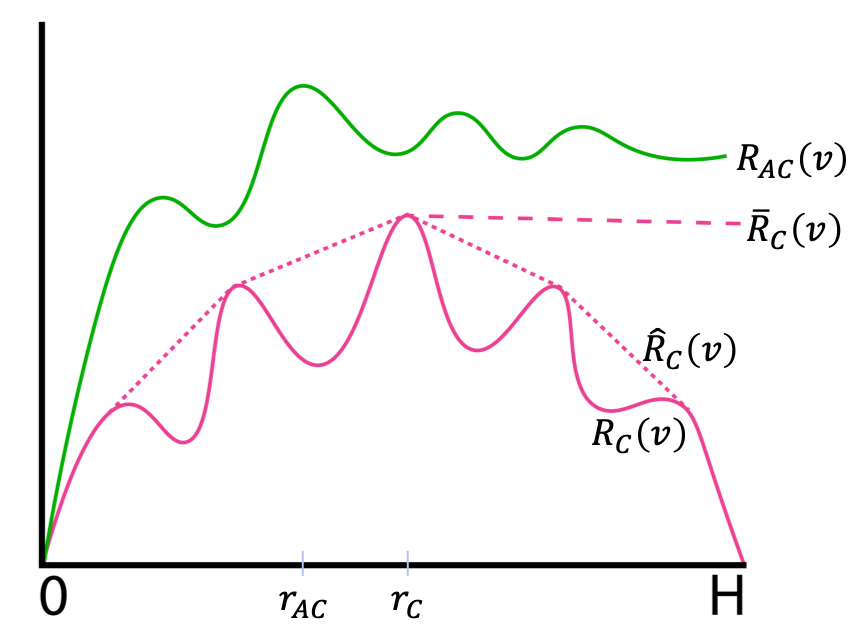}
\caption{A worse-to-better item revenue curve for the FedEx setting that determines the optimal mechanism even without DMR.}
\label{fig:fedexcurves}
\end{minipage}
\hfill
\begin{minipage}[t]{0.48\textwidth}
\centering
\includegraphics[scale=.42]{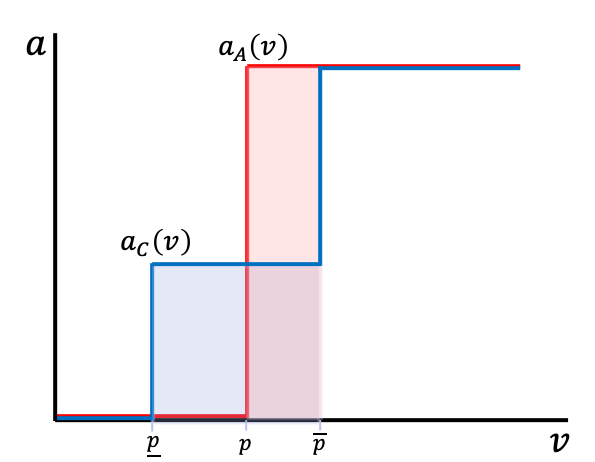}  
\caption{Utility for items $A$ and $C$ are equal for $v \leq \underline{p}$ and $v \geq \underline{p}$, but for $v \in (\underline{p}, \overline{p})$, the randomized option provides more utility.}
\label{fig:splitting}
\end{minipage}
\end{figure}

\vspace{-.3cm}

\paragraph{Three Partially-Ordered Items without DMR.} In our first example, we reasoned about how our decision for item $C$ constrains which prices to set for items $A$ and $B$. In our second example, we reasoned about how our decision for item $A$ constrains prices to set for item $C$. We presented the opposite direction (1) to present both types of arguments and (2) because this direction is necessary without the DMR assumption. For partially-ordered items, however, we really can only reason about how decisions for item $C$ constrain prices for $A$ and $B$. The reason is that in order to know how $p_A$ constrains our options for item $C$, we also need to know $p_B$. Indeed, only $\min\{p_A, p_B\}$ matters for constraining $C$. So we would need to know $p_B$ to know whether a proposed $p_A$ is imposing a new constraint or not.  This results in an impasse for this approach: this partial order requires us to reason about $C$'s price first, but without DMR, we must reason about $A$ and $B$ first.  However, this is only intuition as to why this setting becomes more complicated.  In \Cref{sec:highlevelLB}, we explain why it is that the IC constraints can cause the randomization to get so unwieldy, and \Cref{subsec:3nodmrex} cements this with an example.

Note, however, that we can still reason as we previously did about the optimal item pricing. If, as in the first example, we define $\bar{R}_A(p_C)$ to be the revenue from selling item $A$ at the optimal price that exceeds $p_C$, and $\bar{R}_B(p_C)$ similarly for item $B$, then $R_{ABC}(p_C):=R_C(p_C) + \bar{R}_B(p_C) + \bar{R}_A(p_C)$ accurately defines the revenue we get from all three items by setting price $p_C$ on item $C$, and setting the optimal prices for $A$ and $B$ conditioned on this.

\section{An Exact Characterization Under the Assumption of DMR} \label{app:DMR}

Recall from Subsection~\ref{subsec:rev-iron} that when the distributions satisfy DMR, $\lambda_G(v) = 0$ for all $(v,G)$. Our main result in this section is the following:

\begin{theorem} \label{thm:dmr}
Consider any partially-ordered preferences for items $\G, \succ$.  If the marginal distribution for each item satisfies DMR, the optimal mechanism is deterministic.
\end{theorem}

For a deterministic mechanism, we will set a take-it-or-leave-it price $p_G$ for each item $G$.

\subsection{Intuition}

It will turn out that the optimal mechanism is analogous to that in FedEx and will set prices as follows: 
\begin{itemize}
\item For items $G$ that are sink nodes in the DAG, set $p_G = r_G$.
\item Starting from the sink nodes and visiting nodes in reverse depth, we will define a least upper bound on each node's price based on the prices set for nodes that dominate it.  We define $\bar{p}_G = \min _{G' \in \N(G)} \bar{p}_{G'}$ to be the least upper bound on $G$'s price.  Then set a price of $p_G = \min\{ \bar{p}_G, r_G\}$ for $G$.
\end{itemize}

In our pricing algorithm, nodes $G$ are limited by the smallest $r_A$ for any $A$ that they have a directed path to.  From complementary slackness, every $r_A$ that a node $G$ has a path to is an upper bound on the price that can be set for $G$, so the smallest of these upper bounds is the most limiting. We define $\bar{p}_G$ to be this smallest upper bound, and we define $L_G$ to be the nodes from $\N(G)$ who are also constrained by this upper bound.  Thus, if we follow the sets $L_G$, we will find all of the limiting nodes with $r_A = \bar{p}_G$.

When we send flow out of $G$, we aim to send it along the paths to the nodes that limit $G$'s price the most. We do this recursively, sending from $G$ to the most limiting neighbor, and from there to that node's most limiting neighbor, splitting the flow equally if there are several limiting neighbors.  This raises the limiting reserve and never lowers it.  We update regularly to ensure that we are always sending flow to the now-limiting reserve, raising it, and thus relaxing the constraints on $G$.  This is almost exactly the construction: the only caveat is that we should never send flow out of an item $B$ at $v$ where $f_B(v)\Phi^{\lambda, \alpha}_B(v) > 0$.  If we send into a $B$ along the path where this is the case, we instead send flow out at $r_B < v$.

\subsection{Formal Pricing Algorithm}

Formally, we set the dual variables according to the following algorithm:
\vspace{.5cm}

\begin{algorithmic}
\STATE \textbf{Dual variable construction:}
\STATE Base case: For sink nodes $A$, there is nowhere to send flow. Set $\bar{p}_A = r_A$.
\FOR{all nodes $A$ starting from the sink nodes and in increasing reverse depth \footnote{i.e. \# edges from sink nodes}}
	\STATE $\bar{p}_A = \min _{B \in \N(A)} \bar{p}_B$ 

	\STATE For all $v$ from $r_A$ down to $0$, determine the minimal amount of flow out $\sigma_A$ such that $\varphi_A(v) = 0$.

	\FOR{$v$ from $0$ to $r_A$}
		\STATE Update($A, v, \sigma_A(v)$)
	\ENDFOR
\ENDFOR
\end{algorithmic}
\vspace{.5cm}

\begin{algorithmic}
\STATE \textbf{Update($A, v, \gamma)$:}
\STATE Let $L_A := \{\argmin _{B \in \N(A)} \bar{p}_B\}$.  
\FOR{all $B \in L_A$}
	\STATE Send $\alpha_{A,B}(v) = \frac{1}{|L_A|} \gamma$.
	\STATE Update($B, \min\{v, r_B\}, \gamma$).
\ENDFOR
\end{algorithmic}

\vspace{.5cm}

The key idea is that the price of a node $G$ is limited by the smallest $r_A$ where $A$ is some item better than $G$ (i.e. there is a path from $G$ to $A$ in the DAG). As we send flow along the path to $A$, we raise $r_A$ and it becomes less limiting.  Let $S_G$ be the set of the items that limit $G$ the most, which are precisely the items $A$ such that $r_A = \bar{p}_G$. Since we are in the continuous setting, sending flow is a continuous process. This means that the most limiting item never discretely jumps up higher and becomes no longer limiting.  Instead, all limiting items stay in the set $S_G$ and this set grows as the upper bounds raise and become less limiting.

Let $L_G \subseteq \children(G)$ to be the items such that, for all $B \in L_G$, there exists $v$ such that $\alpha_{G,B}(v) > 0$.  What this means is that $\bar{p}_G = \bar{p}_B$, and $B$ is on the path (if not the end of the path) from $G$ to a limiting item $A \in S_G$.  We will use the variable $\z$ to keep track of the updated $\bar{p}_G$.  If $A \in S_G$, then $f_A(\z)\Phi^{\lambda, \alpha}_A(\z)= 0$, and if $B$ is on a path to some limiting $A$, then $f_B(\z)\Phi^{\lambda, \alpha}_B(\z)\leq 0$. In every step we decrease the amount of flow to send and the algorithm will terminate when there is no flow left to send. Throughout this process the point $\z$ and the set $S_G$ both only increase.

First, we set the flow out of $G$:
$$\sum _{A \in \children(G)} \alpha_{G,A}(v) = \begin{cases}0 & v > \bar{r}_G \\ -\hat{R}''_G(v)/f_G(v) & v \leq \bar{r}_G. \end{cases}$$

\begin{lemma} \label{lem:alphaconstruct}
For every $G$, we can always send $\sigma_G$ out of $G$ distributed among $\children(G)$ such that 
\begin{enumerate}
\item \label{Sgrows} If $\alpha_{G,B}(x) > 0$ for any $x$, then $B \in L_G$.
\item \label{yesS} If $B \in L_G$, then $f_B(\z)\Phi^{\lambda, \alpha}_B(\z) = 0$.
\item \label{nonotS} If $B \in \children(G) \smallsetminus L_G$, then $\z < \overline{r}_B$ and thus $f_B(\z)\Phi^{\lambda, \alpha}_B(\z) \leq  0$.
\end{enumerate}
\end{lemma}

\begin{proof}
Suppose we have $\sigma_G(v)$ flow to send at $v$.  Let $Z = \argmin _{B \in \children(G) \smallsetminus L_G} \overline{p}_{B}$ be the next possible upper bound to hit.

Let $\vareps$ be such that  by sending $\sigma_G$ flow along paths to all items in $S_G$ with  correct proportions, we will maintain $S_G$ and raise $\z$ by $\vareps$.  That is, 
$$\sum _{A \in S_G} f_A(\z + \vareps)\Phi^{\lambda, \alpha}_A(\z+\vareps)= \sigma.$$  
If $\z + \vareps < \overline{p}_{Z}$, we can send this flow without growing $S_G$.  Let $\P(G,A)$ denote the edges forming every path from $G$ to $A$.  For every $(C,D) \in \P(G,A)$ for some $A \in S_G$, we set
$$\alpha_{C,D}(v) = \sum _{A \in S_G: (C,D) \in \P(G,A)} f_A(\z + \vareps) \Phi^{\lambda,\alpha}_A(\z + \vareps) \quad \forall A \in L_G.$$
This will ensure that after this update, $f_A(\z + \vareps)\Phi^{\lambda, \alpha}_A(\z + \vareps) = 0$ for all $A \in S_G$.  Update $\z \leftarrow \z + \vareps$.  Note that (\ref{yesS}) holds by construction, and (\ref{nonotS}) holds since $\z < \overline{p}_{Z} < \overline{p}_B$ for all $B \in \children(G) \smallsetminus L_G$.

Otherwise, suppose $\z + \vareps \geq \overline{p}_{Z}$ and $v \geq \overline{p}_{Z}$.  Then we  instead choose $\vareps = \overline{p}_{Z} - \z$ and make the same update described above, add $Z$ to $L_G$ and add the item $Y$ that is limiting $Z$, that is, $Y$ such that $\overline{p}_Z = R_Y$, to $S_G$.  Note that we have sent positive flow, but the flow sent is  $< \sigma$.  After the update, we will have $\z \leftarrow \z + \vareps =  \overline{p}_{Z}$ and $f_A(\z)\Phi^{\lambda, \alpha}_A(\z) = 0$ for all $A \in S_G$, including $Y$.  Then again (\ref{yesS}) holds, and (\ref{nonotS}) holds since $\z = \overline{p}_{Z} < \overline{p}_B$ for all $B \in \children(G) \smallsetminus L_G$.

Finally, (\ref{Sgrows}) holds in both cases as we only send flow to elements of $S_G$ and $S_G$ is non-decreasing.
\end{proof}

\begin{lemma} \label{lem:noironing1} For every $v$ and $G$, our choice of $\alpha_{G,A}(w)$ for all $w \in [0,H],\, A \in \N(G)$ maintains $\lambda_{G}(v) = 0$ for all $v$. \end{lemma}

\begin{proof} Since the flow out of $G$ is chosen exactly to bring all virtual values to 0 below $\bar{r}_G$, no non-monotonicities are caused.
\end{proof}

\begin{lemma} \label{lem:noironing2} For every $v$ and $G$, any choice of $\alpha_{A,G}(w)$ for all $w \in [0,H],\, A \in \parents(G)$ maintains $\lambda_{G}(v) = 0$ for all $v$. \end{lemma}

\begin{proof} Suppose we get flow $\alpha$ into $G$ at $x$.  Every value $v \leq x$ has $f_G(v)\Phi^{\lambda, \alpha}_{G}(v)$ decrease by $\alpha$ while this remains unchanged for $v > x$, causing no non-monotonicities.

\end{proof}

We are now ready to prove the main result of this section.

\begin{proof}[Proof of Theorem~\ref{thm:dmr}]
We claim the the following deterministic allocation rule always satisfies complementary slackness with the dual: set $p_G = \min\{r_G, r_A : A \in S_G\}$. 

From DMR and our setting of $\lambda$, we will have $\lambda_G(v) = 0$ for all $(v,G)$, automatically satisfying complementary slackness for these variables.  Further, even after sending $\alpha$ flow, $f_G(\cdot)\Phi^{\lambda, \alpha}_G(\cdot)$ will be properly monotone for all $G$ by Lemma~\ref{lem:noironing1} and Lemma~\ref{lem:noironing2}.  

First, we verify that the when we set a price, the virtual values are 0 at that price, so we have the freedom to do so. By Lemma~\ref{lem:alphaconstruct}, $f_A(\z)\Phi^{\lambda, \alpha}_A(\z) = 0$ for all $A \in S_G$.  Of course, by definition of $\bar{r}$, $f_A(\overline{r}_A)\Phi^{\lambda, \alpha}_A(\overline{r}_A) = 0$.  In addition, by definition of the flow out of $G$, $f_G(v)\Phi^{\lambda, \alpha}_G(v) = 0$ for all $v \leq \bar{r}_G$ so $f_G(\z)\Phi^{\lambda, \alpha}_G(\z) = 0$.  Then all of the prices posted are viable.  

It remains to choose a mechanism that satisfies complementary slackness with the $\alpha$ variables.  If $\alpha_{G,B}(v) > 0$ for some $v$ then we know that (1) $B \in L_G$ and (2) $v < \bar{r}_G$. By Lemma~\ref{lem:alphaconstruct}, the variable $\alpha_{G,B}(v) > 0$ for any $v$ if and only if $v \in L_G$, a monotone increasing set as $v$ increases.  In this case, then $\bar{p}_B = \bar{p}_G$ and both are set at this price, satisfying $u_G(v) = u_B(v)$ for all $v$ and automatically satisfying complementary slackness.
\end{proof}

\section{An Extension of FedEx: DAGs with Out-Degree At Most 1} \label{sec:fedexext}

In this section, we consider DAGs with out-degree at most 1.  That is, partial orders that are tree-like, where each item has at most one item that minimally dominates it.  In this case, we see that the FedEx solution applies.

\begin{theorem} \label{thm:fedexext}
Consider any partially-ordered preferences for items $\G, \succ$ such that for any $G$, there exists at most one $G''$ that minimally dominates $G$: that is, $G'' \succ G$ and there does not exist any $G'$ where $G'' \succ G' \succ G$.  Then a nearly identical construction to the FedEx Problem with a minor modification for partial orderings yields closed-form optimal dual variables and the optimal mechanism.
\end{theorem}

We use the notation and methods of \citepalias{FGKK}.  The proof is almost identical, provided for completeness, and much of the following is duplicated from their paper, with a slight modification to allow for the DAG structure with out-degree at most 1.  The key difference is the change in definition of the $\Gamma_{\geq G}$ curves.

We recall the following definitions from their paper:

\begin{itemize}
\item Let $\gamma_G(v) := \varphi_G(v) f_G(v)$.  Recall that $\varphi_G(\cdot) = v- \frac{1-F_G(v)}{f_G(v)}$.  

\item Let $\Gamma_G(v) = \int _0 ^v \gamma_G(x) dx$.  As shown in \citepalias{FGKK}, this function
is the negative of the marginal revenue curve for item $G$.  Thus,
$\Gamma_G(0)=\Gamma_G(H) = 0$ and $\Gamma_G(v) \le 0$ for $v\in [0, H]$.

\item For any function $\Gamma$, define $\hat\Gamma(\cdot)$ to be the lower convex envelope \footnote{ The {\em lower convex envelope} of function $f(x)$ is the supremum over convex functions $g(\cdot)$ such
that $g(x) \le f(x)$ for all $x$.  Notice that the lower convex envelope of $\Gamma(\cdot)$ is the negative of the ironed revenue curve $\hat{R}(v)$.} of
$\Gamma(\cdot)$. We say that $\hat\Gamma(\cdot)$ is \emph{ironed} at $v$ if
$\hat\Gamma(v)\ne \Gamma(v).$

Since $\hat\Gamma(\cdot)$ is convex, it is continuously differentiable except at countably many points and its derivative is monotone (weakly) increasing.

\item Let $\hat\gamma(\cdot)$ be the derivative of  $\hat\Gamma(\cdot)$  and let $\gamma(\cdot)$ be the derivative of  $\Gamma(\cdot)$.
\end{itemize}

As shown in \citepalias{FGKK}, the following  facts are immediate from the definition of lower convex envelope:
\begin{itemize}
\item $\hat\Gamma(v)\le \Gamma(v)\quad \forall v$.
\item $\hat\Gamma(v_{\text{min}})= \Gamma(v_{\text{min}})\text{ where }v_{\text{min}} = \text{argmin}_v \Gamma(v).$  (This implies that there is no ironed interval containing $v_{\text{min}}$.)
\item $\hat\gamma(v)$ is an increasing function of $v$ and hence its derivative $\hat\gamma '(v)\ge 0$ is non-negative for all $v$.
\item If $\hat\Gamma(v)$ is ironed in the interval $[\ell, h]$ , then $\hat\gamma(v)$ is linear
and $\hat\gamma '(v)=0$ in $(\ell, h)$.
\end{itemize}


Now, we redefine the functions $\Gamma_{\geq i}$, which are used to set all of the FedEx dual variables, and can be interpreted as negative combined revenue curves for deadlines $i$ through $m$.  Instead, we redefine them for an item $G$ and all dominated items.  In a DAG, we let the set of all source nodes, that is, items that dominate no other items in the partial order, be the set $S$.  Similarly, we call the set of sinks, items that are dominated by no other items, as the set $T$.

Note that, by assumption, every item has out-degree at most 1.  Then the set of items that minimally dominate an item $G$, $N^{+}(G)$, is of size 0 or 1.  If it is of size 0, then $G \in T$: $G$ is a sink node.  That is:
\begin{observation}For all $G \not \in T$, $|N^{+}(G)| = 1$. \label{obs:outdeg1}\end{observation}
For this reason, we define the following notation.
\begin{definition} For $G \not \in T$, let $D(G)$ refer to the single item that minimally dominates $G$: $G' \in N^+(G)$. \end{definition}

Now, for any source node $G \in S$, define $\Gamma_{\geq G} = \Gamma_G$.  We will define the curves $\Gamma_{\geq G}$ for $G \in \G \smallsetminus G$ inductively.  Let define $r_{\geq G} := \max \argmin \Gamma_{\geq G}(v)$.  Then define 
$$\bar{\Gamma}_{\geq G}(v) := \begin{cases} \hat\Gamma_{\geq G}(v) & v < r_{\geq G} \\ \Gamma_{\geq G}(r_{\geq G}) & v \geq r_{\geq G}. \end{cases}$$
Now, for all $G$ that are not source nodes, we can inductively define
$$\Gamma_{\geq G}(v) := \Gamma_G(v) + \sum _{G' \in N^{-}(G)} \bar{\Gamma}_{\geq G'}(v).$$
Note then that 
$$\bar{\gamma}_{\geq G} = \begin{cases} \hat\gamma_{\geq G} & v \leq r_{\geq G} \\ 0 & v > r_{\geq G} \end{cases} \quad\quad \text{and} \quad\quad \gamma_{\geq G} = \gamma_G + \sum_{G' \in N^{-}(G)} \bar{\gamma}_{\geq G'}(v).$$

\subsection{Primal, Dual, and Complementary Slackness}

We use the following primal and dual formulations with the noted complementary slackness conditions.  They are virtually identical to FedEx, modified for the DAG, and much of it is copied from \citepalias{FGKK}.

\subsection*{The Primal}
 Variables: $a_G(v)$, for all $G \in \G$, and all $v\in [0,H]$.

$$
     \begin{aligned}
      \text{Maximize }   \sum_{\G}  \int _0 ^{H} a_G(v) \gamma_G(v) dv \\
      &\\
            \text{Subject to}\quad\quad\quad\quad\quad\quad\quad\quad\quad\, &\\
       \int_0^v a_G(x) dx -\int_0^v a_{G'}(x) dx & \leq 0   \qquad \forall G \in \G\smallsetminus S, G' \in N^{-}(G) \quad\forall v \in [0,H] \qquad \mbox{\text{ (dual variables $\alpha_{G',G}(v)$)}} \\
          a_G(v) &\leq 1  \qquad \forall G \in\G \quad\forall v \in [0,H]\qquad \mbox{\text{ (dual variables $b_G(v)$)}} \\
          -a_G'(v) &\leq 0  \qquad \forall G\in\G \quad\forall v \in [0,H]\qquad \mbox{\text{ (dual variables $\lambda_G(v)$)}} \\
       a_G(v) &\geq 0  \qquad \forall G\in\G \quad\forall v\in [0,H].
     \end{aligned}
     $$
Note that $a'_G(v)$ denotes $\frac{d}{dv} a_G(v)$.

\noindent
\subsection*{The Dual}
 Variables: $b_G(v), \lambda_G(v)$, for all $G \in \G$, and all $v\in [0,H]$, $\alpha_{G',G}(x)$ for $G \in \G\smallsetminus S, G' \in N^{-}(G)$ and all $x\in [0,H]$.

\begin{equation*}
     \begin{aligned}
            \text{Minimize }   \int_0^H \sum _{G\in \G} b_G(v) dv \quad\quad\quad\ \\
      &\\
            \text{Subject to}\quad\quad\quad\quad\quad\quad\quad\quad\quad\quad\quad\quad\quad\quad\quad\,\,\, &\\
       b_G(v) + \lambda_G'(v) + \sum _{G' \in N^{-}(G)} \int_v^H \alpha_{G',G}(x) dx &\geq \gamma_G(v) \qquad \forall v\in [0,H], G \in S  \mbox{\text{ (primal var $a_G(v)$)}} \\
       b_G(v) + \lambda_G'(v) + \sum _{G' \in N^{-}(G)} \int_v^H \alpha_{G',G}(x) dx & \qquad\qquad\qquad \forall v\in [0,H], i \in \G \smallsetminus S, T \\ 
       -  \int_v^H
       \alpha_{G,D(G)}(x) dx &\geq \gamma_G(v) \hspace{3.8cm} \mbox{\text{ (primal var $a_G(v)$)}} \\
       b_G(v) + \lambda_G'(v) - \int_v^H \alpha_{G,D(G)}(x)dx &\geq \gamma_G(v) \qquad \forall v\in [0,H], G \in T \mbox{\text{ (primal var $a_G(v)$)}}\\
      \lambda_G(H)&=0 \qquad \forall G \in \G\\
       \alpha_{G',G}(v) &\geq 0  \qquad \forall v\in [0,H], G \in \G\smallsetminus S, G' \in N^{-}(G)\\
       b_G(v),\lambda_G(v) &\geq 0 \qquad \forall G \in \G \forall v\in [0,H].
     \end{aligned}
\end{equation*}
Note that $\lambda'_G(v)$ denotes $\frac{d}{dv} \lambda_G(v)$.

\subsection{Conditions for strong duality}
\label{sec:CS}
As long as there are feasible primal and dual solutions satisfying the following conditions, strong duality holds. Theorem 3 from \citepalias{FGKK} proves that these conditions are sufficient.
\begin{eqnarray}
 a_G(v) > 0  &\Rightarrow&  \lambda_G(v) \text{ continuous at }v \quad G \in \G \quad\quad\quad \label{FCS0}\\
  a_G(v) < 1 &\Rightarrow& b_G(v)=0 \qquad G\in\G \label{FCS1}\\
  a_G'(v) > 0 &\Rightarrow& \lambda_G(v)=0 \qquad G\in\G \label{FCS2}\\
  \int_0^v a_G(x)dx <\int_0^v a_{G'}(x) dx &\Rightarrow& \alpha_{G',G}(v)=0  \,\,G \in \G\setminus S, G' \in N^{-}(G)\label{FCS3}\\
    b_G(v) +\lambda_G'(v) + \sum _{G' \in N^{-}(G)} \int_v^H \alpha_{G',G}(x) dx \hspace{1.4cm}& \nonumber\\
    - \int_v^H \alpha_{G,D(G)}(x)dx  > \gamma_G(v) &\Rightarrow& a_G(v)=0 \qquad G \in \G\smallsetminus S,T\label{FCS4}\\
  b_G(v) +\lambda_G'(v) + \sum _{G' \in N^{-}(G)} \int_v^H\alpha_{G',G}(x)dx > \gamma_G(v) &\Rightarrow& a_G(v)=0 \qquad G \in T\label{FCS5}\\
    b_G(v) +\lambda_G'(v) - \int_v^H \alpha_{G,D(G)}(x)dx >\gamma_G(v) &\Rightarrow& a_G(v)=0 \qquad G \in S\label{FCS6}
\end{eqnarray}

We allow $a_G'(v)\in \R  \cup \{+\infty\}$.  It may have (countably many) discontinuities.  However, the proof of optimality in \citepalias{FGKK} handles this.

\subsection{Optimal Primal Variables}

We determine the allocation rules inductively, from sink nodes all the way to source nodes.  First, for sink nodes $G \in T$, set $$a_G(v) = \begin{cases} 0 & v < r_{\geq G} \\ 1 & v \geq r_{\geq G}. \end{cases}$$

Suppose that $a_{G'}$ has been defined for $G' = D(G)$,
with jumps at $v_1, \ldots, v_k$, and values
$0= \beta_0 < \beta_1 \le \beta_2 \ldots \le \beta_k=1$. That is,

$$ a_{G'}(v)=
\begin{cases}0 &\mbox{if }v <v_1, \\
\beta_j &v_j \le v < v_{j+1}\quad 1\le j < k\\
1 & v_k \le v.  \end{cases} $$
Thus, we can write
$$a_{G'}(v) = \sum_{j=1}^{ k}(\beta_j - \beta_{j-1})a_{G',j}(v)$$
where
$$a_{G',j}(v)=
\begin{cases}0 &\mbox{if }v <v_j  \\
1&v \ge v_j. \end{cases} $$
Next we define $a_G(v)$.

\begin{definition} \label{alloc curves}
Let $j^*$ be the largest $j$ such that $v_j \le r_{\ge G}$. For any $j \le j^*$,
consider two cases:
\begin{itemize}
\item $ \hat\Gamma_{\ge G}(v_j)= \Gamma_{\ge G}(v_j)$, i.e.
$\hat\Gamma_{\ge G}$ not ironed at $v_j$:
In this case,
define
$$a_{G,j}(v)=
\begin{cases}0 &\mbox{if }v <v_j  \\
1 & otherwise.  \end{cases}. $$
\item $ \hat\Gamma_{\ge G}(v_j)\ne \Gamma_{\ge G}(v_j)$:
In this case, let
\begin{itemize}
\item $\underline{v}_j:= \text{  the largest }v < v_j \text{ such that }
\hat\Gamma_{\ge G}(v)= \Gamma_{\ge G}(v)\text{ i.e., not ironed},$
and
\item $\overline{v}_j:= \text{  the smallest }v > v_j \text{ such that }
\hat\Gamma_{\ge G}(v)= \Gamma_{\ge G}(v)\text{ i.e., not ironed}$.
\end{itemize}
Let $0< \delta <1$ such that
$$v_j = \delta\underline{v}_j + (1- \delta)\overline{v}_j.$$
Then $\hat\Gamma_{\ge G}(\cdot)$ is linear between $\underline{v}_j$
and $\overline{v}_j$: $$\hat\Gamma_{\ge G}(v_j) = \delta\Gamma_{\ge G}(\underline{v}_j )+ (1- \delta)\Gamma_{\ge G}(\overline{v}_j).$$

Define
$$a_{G,j}(v)=
\begin{cases}0 &\mbox{if }v <\underline{v}_j  \\
\delta&\underline{v}_j \le v < \overline{v}_j\\
1 & otherwise.  \end{cases} $$
\end{itemize}
Finally, set $a_G(v)$ as follows:
\begin{equation}
\label{def:ai}
a_{G}(v)= \begin{cases}
\sum_{j=1}^{ j^*}(\beta_j - \beta_{j-1})a_{G,j}(v) &\quad\mbox{if } v < r_{\ge G}, \\
& \\
1& \quad v \ge r_{\ge G}.  \end{cases}
\end{equation}
\end{definition}

\noindent
\textbf {Remark:}
In order to continue the induction and define $a_{G''}(v)$ for $G = D(G'')$, we need to rewrite $a_G(v)$ in terms of
functions $a_{G,j}(v)$ that take only 0/1 values.   This is straightforward.

\subsection{Closed-Form Dual Variables and Proof of Optimality}

The following dual variables and proofs are again almost verbatim from \citepalias{FGKK} with very small modifications for the DAG structure.

\begin{align}
\lambda_G(v) &= \Gamma_{\geq G}(v) - \hat\Gamma_{\geq G}(v) \label{defLambda}\\
b_G(v) &= \begin{cases} 0 & v < r_{\geq G} \\ \hat\gamma_{\geq G}(v) &  v \geq r_{\geq G} \end{cases}\label{defbi} \\
\alpha_{G',G}(v) &= \begin{cases} \hat\gamma_{\geq G'}'(v) & v < r_{\geq G'} \\ 0 & v \geq r_{\geq G'} \label{defAlphai}\end{cases} & \text{for } G \in \G \smallsetminus S, G' \in N^{-}(G)
\end{align}

Taking the derivative of (\ref{defLambda}), and using the definition of $\Gamma$, we obtain:
\begin{align}
\gamma_G(v) - \lambda'_G(v) &= \hat\gamma_{\geq G}(v) - \sum _{G' \in N^{-}(G)} \bar{\gamma}_{\geq G'}(v) & \text{for } G \in \G \smallsetminus S \label{RewriteDual}\\
\gamma_G(v) - \lambda_G'(v) &= \hat\gamma_G(v) & \text{for } G \in S
\end{align}

Also,  using (\ref{defAlphai}) and the fact that $\hat\gamma_{\geq i+1}(r_{\geq i+1}) =0$, we get:
\begin{equation}
A_{G',G}(v):= \int_v^{H} \alpha_{G',G}(x)\,dx = \begin{cases} -\hat\gamma_{\geq G'}(v) & v < r_{\geq G'} \\ 0 & v \geq r_{\geq G'} \end{cases} = - \bar{\gamma}_{\geq G'}(v)
\label{Ai}
\end{equation}

\begin{claim} Conditions (\ref{FCS4})--(\ref{FCS6})  and dual feasibility: For all $G$ and $v$, $a_G(v) > 0 \implies$ the corresponding dual constraint is tight, and the dual constraints are always feasible. \end{claim}

\begin{proof} Rearrange the dual constraint $b_G(v) + \sum _{G' \in N^{-}(G)} A_{G',G}(v) - A_{G,D(G)}(v) + \lambda_G'(v) \geq \gamma_G(v)$ to  $$b_G(v) - A_{G,D(G)}(v) \geq \gamma_G(v) - \lambda_G'(v) - \sum _{G' \in N^{-}(G)} A_{G',G}(v).$$
\textbf{Fact 1:} For $G \not \in S$, 
$\gamma_G(v) - \lambda_G'(v) - \sum _{G' \in N^{-}(G)} A_{G',G}(v) = \hat\gamma_{\geq G}(v)$ 
for all $v$.  To see this use (\ref{RewriteDual}) and (\ref{Ai}):
$$ \gamma_G(v) - \lambda'_G(v) = \hat\gamma_{\geq G}(v) - \sum _{G' \in N^{-}(G)} \bar{\gamma}_{\geq G'}(v) \hspace{2cm} A_{G',G}(v) = -\bar{\gamma}_{\geq G'}$$
\\
\textbf{Fact 2:} For $G \not \in T$,  
$b_G(v) - A_{G,D(G)}(v) = \hat\gamma_{\geq G}(v)$ for all $v$.
$$b_G(v) = \begin{cases} 0 & v < r_{\geq G} \\ \hat\gamma_{\geq G}(v) & v \geq r_{\geq G} \end{cases} \hspace{2cm} -A_{G,D(G)}(v) = \begin{cases} \hat\gamma_{\geq G}(v) & v < r_{\geq G} \\ 0 & v \geq r_{\geq G} \end{cases}$$
Hence for $G \not \in T$, $b_G(v) -  A_{G,D(G)}(v) \geq \gamma_G(v) - \lambda_G'(v) - \sum _{G' \in N^{-}(G)} A_{G',G}(v)$ for all $v$.

For $G \in S$, since $\gamma_{\geq G} = \gamma_G$, and $\gamma_G(v) - \lambda_G'(v) = \hat\gamma_G(v)$. Combining this with Fact 2 above, we get that $b_G(v) - A_{G,D(G)}(v) + \lambda_G'(v) = \gamma_G(v)$ for all $v$.

Finally, for $G \in T$, using Fact 1, for $v < r_{\geq G}$, we get $$b_G(v) = 0 \geq \hat\gamma_{\geq G}(v) = \gamma_G(v) -\lambda_G'(v) - \sum _{G' \in N^{-}(G)} A_{G',G}(v)$$ which is true for $v < r_{\geq G}$.  For $v \geq r_{\geq G}$, we get $$b_G(v) = \gamma_{\geq G}(v) = \gamma_G(v) -\lambda_G'(v) - \sum _{G' \in N^{-}(G)} A_{G',G}(v),$$ so the dual constraint is tight when $a_G(v) > 0$ as this starts at $r_{\geq G}$.
\end{proof}

The proofs of conditions \Cref{FCS0}--\Cref{FCS3} are identical as in \citepalias{FGKK} using our slightly modified dual variables:

\begin{claim} Condition (\ref{FCS1}): For all $G$ and $v$, $a_G(v) < 1 \implies b_G(v) = 0$. \end{claim}

\begin{proof} If $a_G(v) < 1$, then $v < r_{\geq G}$, so by construction, $b_G(v) = 0$. \end{proof}

\begin{claim} Condition (\ref{FCS2}): For all $G$ and $v$, $a_G'(v) > 0 \implies \lambda_G(v) = 0$. \end{claim}

\begin{proof} From Subsection 4.2 of \citepalias{FGKK}, $a_G'(v) > 0$ only for unironed values of $v$, at which $\lambda_G(v) = 0$. \end{proof}

\begin{claim} \label{claim4} Condition (\ref{FCS3}): For all $G \in \G \smallsetminus S, G' \in N^{-}(G)$ and $v$, $\int _0 ^v a_G(x) dx < \int _0 ^v a_{G'}(x) dx \implies \alpha_{G',G}(v) = 0$. \end{claim}

\begin{proof}
As discussed at the end of the proof of Lemma 1 of \citepalias{FGKK}, $\int _0 ^v a_G(x) dx = \int _0 ^v a_{G'}(x) dx $ unless $\Gamma_{\ge G'}$ is ironed at $v$, or $v \ge r_{\ge G}$. In both of these cases $\alpha_{G',G}(v)=0$ (by our fourth fact about the lower convex envelope and equation~\Cref{defAlphai}, respectively).
 \end{proof}

The above claims prove that this dual solution satisfies feasibility and all complementary slackness and continuity conditions from Section~\ref{sec:CS} hold.
\section{Details from \Cref{sec:highlevelLB}: Construction of a Candidate Dual Instance}

\label{app:LB}



We extend the above examples from \Cref{sec:highlevelLB} to require any number of menu options.  As in the two examples, we can reason from the top downward that the allocation at the bottom of every ironed interval must be positive, and reason from the bottom upward that the allocation must strictly increase for each new overlapping ironed interval we encounter, yielding all different menu options.  We formally define this interleaving structure and call it a ``chain,'' depicted in Figure~\ref{fig:chaindef}.  As another sanity check: each new point in the chain induces a new equality that has to be satisfied. So if the chain is of length $M$, intuition suggests that we should need $M$ degrees of freedom to possibly satisfy complementary slackness (but this is just intuition). 


\begin{figure}[h!]
\centering
\includegraphics[scale=.4]{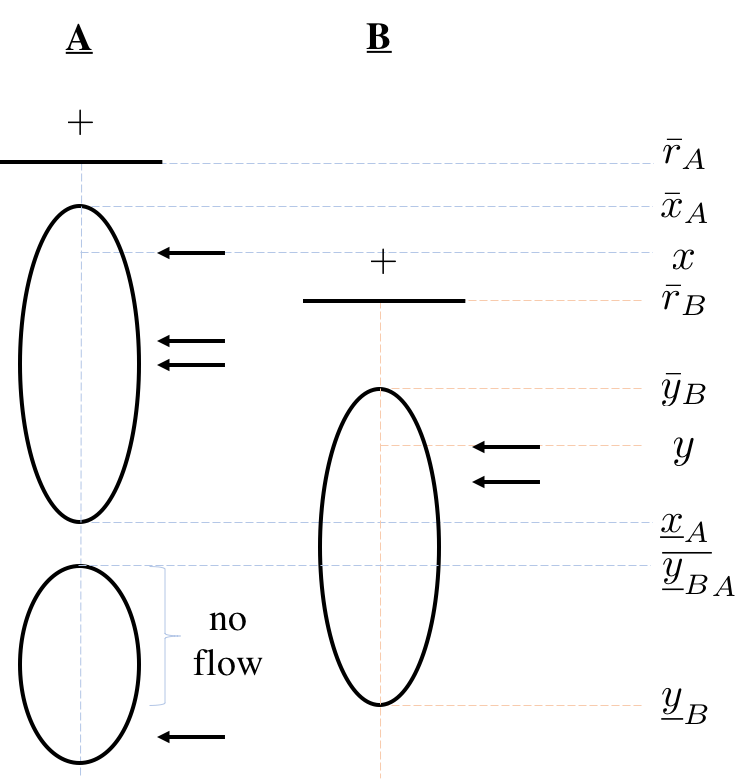}
\caption{This is an example of a chain that consists of the points $\{(x,A), (y,B)\}$.  It is a top chain as $x > \bar{r}_B$.  Note that $(y,B)$ is preceded by $(x,A)$ as there is flow into $B$ at $y$ and $y > \underline{x}_A$, and there is no flow into $B$ for any $v \in (y, \bar{y}_B]$.  The chain terminates at $(y,B)$ since there is no flow into $A$ for any $v \in [\underline{y}_B, \overline{\underline{y}_B}_A]$.}
\label{fig:chaindef}
\end{figure}

\begin{definition}[Top chain] \label{def:chain} 

A sequence $(x_1, A), (x_2, B), (x_3, A), \cdots$ of points that switch between items $A$ and $B$ is called a {\em chain} if the following hold:
\begin{itemize}
\item  $\Phi^{\lambda, \alpha}_A(x) =0 $ for all $(x,A)$ in the chain and  $\Phi^{\lambda, \alpha}_B(y) = 0$ for all $(y,B)$ in the chain.
\item $\alpha_{C,A}(x) > 0$ for all $(x,A)$ in the chain and $\alpha_{C,B}(y) > 0$ for all $(y,B)$ in the chain.
\item $\lambda_A(x) > 0$ for all $(x,A)$ in the chain and $\lambda_B(y) > 0$ for all $(y,B)$ in the chain.
\item $\underline{x_i}_A < x_{i+1} < x_i$ if $(x_i,A)$ is in the chain and $\underline{x_i}_B < x_{i+1} < x_i$ if $(x_i,B)$ is in the chain.
\end{itemize}  

We call a chain the \emph{top chain} if $x_1 > \overline{r}_B$. 
\end{definition}

Note that if any of these conditions do not hold, the mechanism has an easier solution.  If any point $v$ in the zero regions of both $A$ and $B$ were unironed, we could just set a price of $v$ for both.  If the chains did not interleave with flow alternating in, our series of constraints would end.  The top chain structure (and it is key that it is a top chain) prevents this.

We now provide a complete proof of Theorem~\ref{thm:unbounded}.  First, we provide a construction of our candidate dual, which is depicted in Figure~\ref{fig:candidatedual8}.  The instance uses definition~\ref{def:chain} of a top chain. \\
\\
\noindent \textbf{Construction of candidate dual instance:} 
\begin{itemize}
\item Let there exist no point at which $A$ and $B$ both have virtual value zero and both are unironed, that is, there is no $v$ such that $\Phi_A^{\lambda, \alpha}(v) = \Phi_B^{\lambda, \alpha}(v) = 0$ and $\lambda_A(v) = \lambda_{B}(v) = 0$.  
\item Let $\overline{r}_A > x_1 > \overline{r}_B > x_2 > x_3 > \cdots > x_M > \underline{r}_B > \underline{r}_A$. The dual has a top chain of length $M$ defined by $(x_1, A), (x_2, B), \ldots, (x_M,A)$.  
\item In addition, we have flow into the other item at each point in the chain: let $\alpha_{C,B}(x_i) > 0$ for all $(x_i, A)$ in the chain as well as $\alpha_{C,A}(x_i) > 0$ for all $(x_i, B)$ in the chain.  
\item Let $\lambda_C(v) = 0$ for all $v$, {\em i.e.}, item $C$ is unironed everywhere. 
\item For all $v$ where $\alpha$ has not already been defined, let $\alpha_{C, A}(v) = \alpha_{C, B}(v) = 0$.
\end{itemize}

We first make some remarks that follow directly from our construction. All the remarks below (only) talk about our dual and any feasible primal that satisfies complementary slackness with our dual.

\begin{remark}\label{rem:1}
 For all $i \in \{1, 3, \cdots , \menuseq -2\}$, we have $x_i, x_{i+1} \in [\underline{x_{i}}_A, \overline{x_{i}}_A] = [\underline{x_{i+1}}_A, \overline{x_{i+1}}_A]$. Since this interval is ironed, we have $\lambda_A(v) > 0 \implies a'_A(v) = 0$ for $v$ in this interval. Thus, $a_A(x_i) = a_A(x_{i+1})$.
\end{remark}

\begin{remark} \label{rem:2}
For all $i \in \{2, 4, \cdots , \menuseq -1\}$, we have $x_i, x_{i+1} \in [\underline{x_{i}}_B, \overline{x_{i}}_B] = [\underline{x_{i+1}}_B, \overline{x_{i+1}}_B]$. Since this interval is ironed, we have $\lambda_B(v) > 0 \implies a'_B(v) = 0$ for $v$ in this interval. Thus, $a_B(x_i) = a_B(x_{i+1})$.
\end{remark}

\begin{remark} \label{rem:3}
For all $i \in \{1,2, \cdots , \menuseq\}$, we have  $u_A(x_i) = u_B(x_i)$.
\end{remark}
 
We now prove a lemma that forms the backbone of our inductive argument:
 
\begin{lemma} \label{lem:3}
For all $i \in \{1, 2, \cdots , \menuseq-1\}$, we have  $a_A(x_i)  > a_B(x_i) \iff a_A(x_{i+1}) < a_B(x_{i+1})$. Similarly, we have $a_A(x_i) < a_B(x_i) \iff a_A(x_{i+1}) > a_B(x_{i+1})$
\end{lemma}
\begin{proof} Note that either $a_A(x_i) = a_A(x_{i+1})$ or $a_B(x_i) = a_B(x_{i+1})$ by Remark~\ref{rem:1} and Remark~\ref{rem:2}. We only prove $a_A(x_i)  > a_B(x_i) \iff a_A(x_{i+1}) < a_B(x_{i+1})$ for the case $a_A(x_i) = a_A(x_{i+1})$ and omit the other (symmetric) cases. Since  $a_A(x_i) = a_A(x_{i+1})$, we have 

\[u_A(x_i) = u_A(x_{i+1}) + a_A(x_{i}) \cdot (x_i - x_{i+1}) = a_A(x_{i+1}) \cdot (\underline{x_i}_B - x_{i+1}) + a_A(x_{i}) \cdot (x_i - \underline{x_i}_B).\]

We also have, by the structure of the ironed intervals for the item $B$, 

\[u_B(x_i) = u_B(x_{i+1}) + a_B(x_{i+1}) \cdot (\underline{x_i}_B - x_{i+1}) + a_B(x_{i}) \cdot (x_i - \underline{x_i}_B) \]

Now, since the utilities at all points $x_i$ is the same for both items $A$ and $B$ (Remark~\ref{rem:3}), the fact that $a_A(x_i)  > a_B(x_i)$ is equivalent to  $a_A(x_{i}) \cdot (x_i - \underline{x_i}_B) > a_B(x_{i}) \cdot (x_i - \underline{x_i}_B)$ which is equivalent to $a_A(x_{i+1}) \cdot (\underline{x_i}_B - x_{i+1}) < a_B(x_{i+1}) \cdot (\underline{x_i}_B - x_{i+1})$ which, in turn, is equivalent to $a_A(x_{i+1}) < a_B(x_{i+1})$.
\end{proof}
 
Finally, we prove Theorem~\ref{thm:unbounded}.
 

\begin{proof}[Proof of Theorem~\ref{thm:unbounded}]
At $x_\menuseq$, we have that

\[u_A(x_\menuseq) = a_A(x_\menuseq) \cdot (x_\menuseq - \underline{x_\menuseq}_A) \quad\quad \text{and} \quad\quad u_B(x_M) = a_B(x_\menuseq) \cdot (x_\menuseq - \underline{x_\menuseq}_B).\]

Since $\underline{x_\menuseq}_B > \underline{x_\menuseq}_A$ and $a_A(x_\menuseq) > 0$, then to ensure that $u_A(x_M) = u_B(x_M)$ (Remark~\ref{rem:3}), we must have $a_B(x_M) > a_A(x_\menuseq)$.  However, with this fact, Lemma~\ref{lem:3} says that $a_B(x_i) > a_A(x_i)$ and $a_B(x_{i+1}) > a_A(x_{i+1})$ in alternation. 

Since $a_A(\cdot)$ and $a_B(\cdot)$ are non-decreasing sequences, they can only alternate if they have $\Omega(\menuseq)$ distinct elements.

By Theorem~\ref{thm:allocalg}, there exists a feasible primal that satisfies complementary slackness.  The primal algorithm constructs a mechanism with menu complexity at least $\menu$ and satisfies complementary slackness, hence this dual is in fact optimal.
\end{proof}
%



\begin{corollary} This idea gives a lower bound for Multi-Unit Pricing as well. \end{corollary} We expand on this on \Cref{sec:MUPLB}.
\section{Menu Complexity is Finite: Characterizing the Optimal Mechanism via Duality}
\label{app:optvars}

In this section, we'll characterize the optimal mechanism for three items $\{A, B, C\}$ with structure $A \succ C$, $B \succ C$, and $A \not \succ B, B \not \succ A$. While our approach will be algorithmic, our focus isn't to actually run this algorithm or analyze its runtime. We'll merely use the algorithms to deduce structure of the optimal mechanism. We prove essentially that the interleaving of ironed intervals used in the construction of the previous section is the worst case (in terms of menu complexity of the optimal mechanism). Still, in order to possibly prove this, we need to at minimum find an optimal mechanism for every possible instance. 

Our approach is the following: we propose a \emph{primal recovery algorithm} that, given a dual $(\lambda, \alpha)$, produces a primal solution that (1) satisfies complementary slackness with the dual and (2) has finite menu complexity.
Obviously, the algorithm can't possibly succeed for every input dual (as some duals are simply not optimal for any instance). But we show that whenever the algorithm fails, the dual has some strange structure (elaborated below). We then show that the best dual (which is optimal and always exists, definition below) never admits these strange structures, and therefore the algorithm always succeeds when given the best dual as input. 

\begin{definition}[Best Dual]\label{def:best} We define the best dual of an instance with three partially-ordered items to be the $(\lambda, \alpha)$ satisfying the following:
\begin{enumerate}
\item First, $(\lambda, \alpha)$ is optimal: $(\lambda, \alpha) \in \arg\min\{\sum_{G \in \{A, B, C\}}\int_0^H f_G(v)\cdot \max\{0, \Phi^{\lambda, \alpha}_G(v)\} dv\}$. 
\item Among $(\lambda, \alpha)$ satisfying (1), $(\lambda, \alpha)$ has the \emph{fewest ironed intervals} of virtual value zero. That is, $(\lambda, \alpha)$ minimizes $|\mathcal{I}(\lambda,\alpha)| = |\{\underline{x}_G \mid (x,G) \in [0,H] \times \{A, B, C\}, \Phi^{\lambda, \alpha}_G(v) = 0\}|$.
\item Among $(\lambda, \alpha)$ satisfying (2), $(\lambda, \alpha)$ has the \emph{lowest positives} (lexicographically ordered). That is, $(\lambda, \alpha)$ minimizes $\bar{r}_A$, followed by $\bar{r}_B$, followed by $\bar{r}_C$.
\end{enumerate}
\end{definition}

\begin{definition} A \emph{double swap} exists when there are consecutive points $(x,A)$ and $(y,B)$ in a chain, and there is flow into $A$ for $v \in [\underline{x}_A, y)$. See Figure~\ref{fig:doubleswap}. \end{definition}

\begin{definition} An \emph{upper swap} occurs when there is flow into $(x,A)$ and $(y,B)$ where $x > \bar{r}_A > y > \bar{r}_B$. See Figure~\ref{fig:swaps}.\end{definition}

\begin{proposition}\label{prop:bestdual} The best dual has no double swaps or upper swaps. 
\end{proposition}

The full proof of Proposition~\ref{prop:bestdual} appears below. 
The high-level approach is that whenever a double swap or upper swap exists, we can exploit this structure to modify the dual variables.  This creates a better dual solution (with respect to definition~\ref{def:best}) and proves that (2) or (3) respectively must not have held for the original dual.


\begin{numberedtheorem}{\ref{thm:allocalg}} For any best dual solution, we can find a primal with finite menu complexity that satisfies complementary slackness (and is therefore optimal). \end{numberedtheorem}

A full proof appears below,
but the high-level approach is explained in the following.

\begin{proof}[Proof Sketch of Theorem~\ref{thm:allocalg}]
(No bad structures exist in best duals.)  First, we try to satisfy the necessary complementary slackness system of equations as in \Cref{app:LB}, and identify all possible barriers to solutions existing. These barriers are exactly double swaps or upper swaps, which are not found in best duals by Proposition~\ref{prop:bestdual}. 

(Inductive primal recovery algorithm.)  Without these barriers, an inductive argument shows that we can indeed find an allocation rule that satisfies all of the complementary slackness conditions.  Every dual has a (possibly empty) top chain, and each point in the chain has another set of preferability constraints for that item, along with the constraint that the allocation is constant.  We use induction to handle one point in the chain at a time. (See Figure~\ref{fig:ISrecovery} in \Cref{app:optvars}.)  We take the partially-constructed allocation that satisfies the constraints for the chain so far, scale it down (and thus continue to satisfy the constraints), and then solve for the allocation probability that will satisfy the new constraints given by this point in the chain.  As shown in \Cref{app:LB}, this requires choosing a different allocation probability at the bottom of each ironed interval in the chain, but we show that this is sufficient, giving menu complexity at most the length of the chain $+\,1$.  

(Finite menu complexity.) The other interesting part not addressed in \Cref{app:LB} is what to do if there is a chain of countably infinite length (which can certainly exist). Naively following our algorithm would indeed result in a primal of countably infinite menu complexity. But, because the sequence of chain points is monotonically decreasing (and lower bounded by zero), they must converge to some value $v$. If they converge, and the chain is indeed infinitely long, then neither $A$ nor $B$ can possibly be ironed at $v$, and we can simply set price $v$ for both items instead. 
\end{proof}

We begin below by reviewing properties of the dual previously observed in~\cite{FGKK, DW}. Throughout this section we'll reference the ``best'' dual. While multiple optimal duals might exist, we'll be interested in a specific tie-breaking among them (and refer to the one that satisfies these conditions as ``best"). 

\begin{theorem}[\cite{DW}] \label{thm:bestdualdw}The best dual $(\lambda, \alpha)$ satisfies the following:
\begin{itemize}
\item (Proper monotonicity) $(f_G \cdot \Phi^{\lambda, \alpha}_G)(\cdot)$ is monotone non-decreasing, for all $v$.
\item (No-boosting) $\Phi^{\lambda, \alpha}_G(v) \geq 0$ for all $G$ such that there exists a $G' \succ G$. 
\item (No-rerouting) $\Phi^{\lambda, \alpha}_G(v) > 0 \Rightarrow \alpha_{G, G'}(v) = 0$ for all $G'$. 
\item (No-splitting) $\lambda_G(v) > 0 \Rightarrow \alpha_{G, G'}(v) = 0$ for all $G'$. 
\end{itemize}

\end{theorem}

Returning to our three-item example, prior work nicely characterizes the flow coming out of $C$ in the optimal dual: No-boosting tells us that we must always send flow out of $(v, C)$ into somewhere whenever $\Phi^{\lambda, \alpha}_C(v) < 0$ (in order to bring it up to $0$). No-rerouting tells us that we can never send flow out of $(v, C)$ if $\Phi^{\lambda, \alpha}_C(v) > 0$. No-splitting tells us that we never send flow out of the middle of an ironed interval. But, we still need to decide whether to send this flow into $A$ or $B$. This is the novel part of our analysis.

\begin{proof}[Proof of Proposition~\ref{prop:bestdual}]  By Definition~\ref{def:best}, we know that a best dual has the minimum number of ironed intervals amongst all optimal duals. Similarly, a best dual has the lowest positives amongst all optimal duals. We prove the proposition using two lemmas. The first lemma proves that a best dual can't have double swaps:

\begin{lemma} The optimal dual that has the minimal number of ironed intervals does not contain any double swaps. \label{DS} \end{lemma}

First, we discuss why this structure would cause a problem for how we're used to satisfying complementary slackness conditions.  Complementary slackness forces that in the ironed intervals $[\underline{z}_B, \overline{z}_B]$ and $[\underline{y}_A, \overline{y}_A]$, the allocation is constant, and thus utility in these regions is linear.  However, no linear utility functions can satisfy the preferability constraints of having  utility that is higher for item $A$, then $B$, then $A$, as illustrated on the left in Figure~\ref{fig:doubleswap}.

\begin{figure}[h!]
\begin{center} \includegraphics[scale=.5]{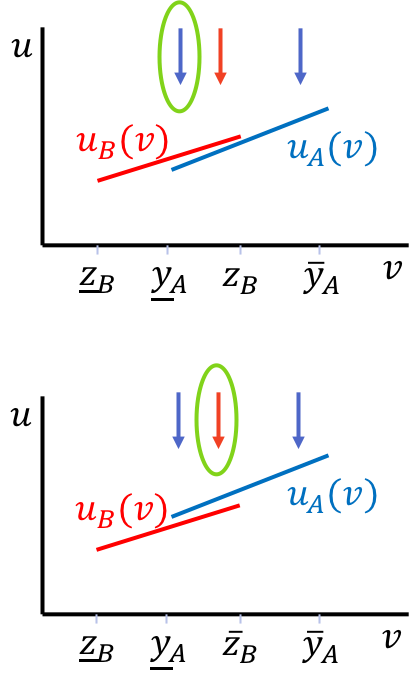} \hspace{1cm} \includegraphics[scale=.4]{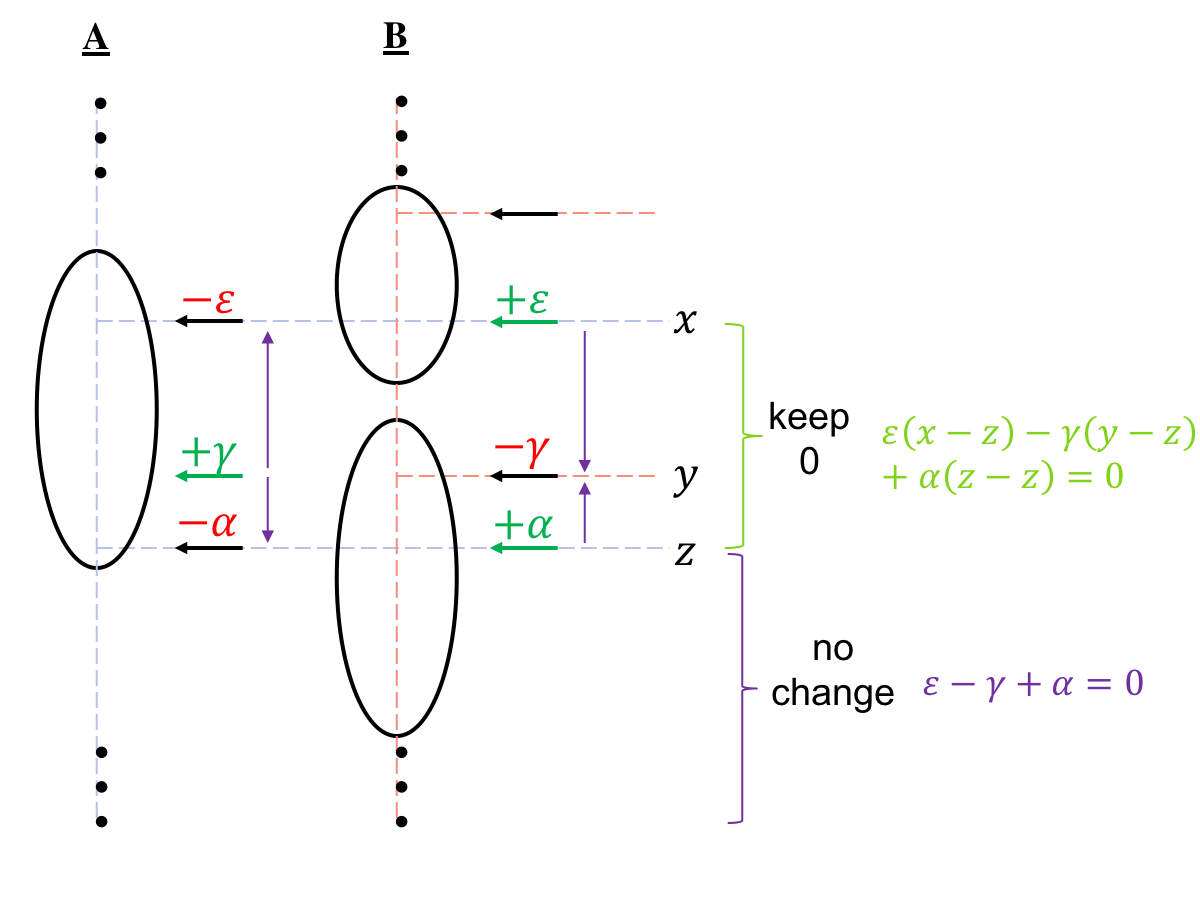}  \end{center}
\caption{Left: Complementary slackness forces linear utility in ironed intervals.  For any choice of linear utility functions, we cannot satisfy the preferability constraints imposed by the double swap for item $A$, then $B$, then $A$ in this region.  The violated constraint corresponds to the circled arrow. 
Right: The operation used in the proof of Lemma~\ref{DS}, using a double swap to maintain virtual welfare and create fewer ironed intervals.}\label{fig:doubleswap}
\end{figure}

\begin{proof}
Proof by contradiction. Suppose that somewhere in the top chain, some point in the chain $(x,A)$ is succeeded by $(y,B)$ and $\alpha_{C,A}(z) > 0$ for some $z \in (\underline{x}_A,y)$, creating a double swap.  We consider the following operation (depicted on the right in Figure~\ref{fig:doubleswap}) that pushes flow down within the ironed interval $[\underline{x}_A, \overline{x}_A]$ and does the reverse on $B$, yet negates the change in flow at $z$ to maintain the virtual values below here.  Move $\vareps$ flow from $(x,A)$ to $(x,B)$.  Move $\gamma$ flow from $(y,B)$ to $(y,A)$.  Move $\alpha$ flow from $(z,A)$ to $(z,B)$.  We will set
$$\alpha = \left(\frac{x-y}{y-z}\right) \vareps \quad\quad\quad\quad \text{and} \quad\quad\quad\quad \gamma = \left(1 + \frac{x-y}{y-z}\right) \vareps.$$
First, this ensures that $\vareps - \gamma + \alpha = 0$, and thus for $v \leq z$, $\hat\Phi^{\lambda,\alpha}_A(v) = \Phi^{\lambda,\alpha}_A(v)$ as well as $\hat\Phi^{\lambda,\alpha}_B(v) = \Phi^{\lambda,\alpha}_B(v)$.  Second, this ensures that $\vareps(x-z) - \gamma(y-z) = 0$, keeping the average virtual value from $z$ to $x$ the same for both items.
\begin{align*}
\int _z ^x f_A(v) \hat \Phi^{\lambda,\alpha}_A(v) dv &= \int _z ^y f_A(v) (\Phi^{\lambda,\alpha}_A(v) + \vareps - \gamma) dv + \int _y ^x f_A(v) (\Phi^{\lambda,\alpha}_A(v) - \gamma) dv \\
&= \int _z ^x f_A(v) \Phi^{\lambda,\alpha}_A(v) dv + \vareps(x-z) - \gamma(y-z) \\
&= \int _z ^x f_A(v) \Phi^{\lambda,\alpha}_A(v) dv 
\end{align*}
However, the virtual values in $[y,x]$ are increasing for item $A$ and decreasing for item $B$, and likewise those in $[z,x)$ are decreasing for item $A$ and increasing for item $B$.  If we choose $\vareps$ small enough as to not uniron the interval $[\underline{x}_A, \bar{x}_A]$, the change gets spread around the interval and the interval remains all zeroes.  However, for item $B$, the interval $[\underline{y}_B, \bar{y}_B]$ becomes positive while the region above becomes negative.  Since the average of both regions is the same and there is now a non-monotonicity, the regions will be ironed together, creating a larger ironed interval with virtual value zero.

Since the virtual welfare of the dual hasn't changed, but we have reduced the number of ironed intervals, then we did not start with an optimal dual with the fewest possible ironed intervals, deriving a contradiction.
\end{proof}




\begin{figure}
\begin{center}  \includegraphics[scale=.5]{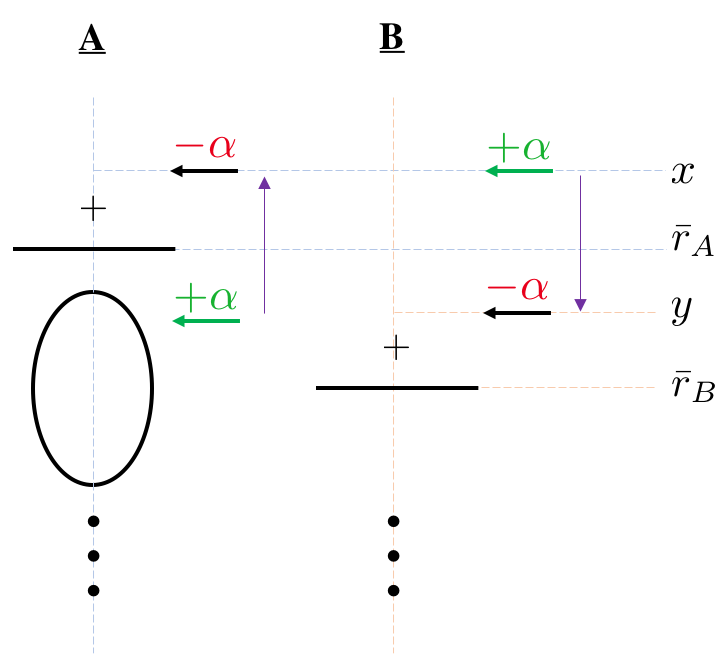} \end{center}
\caption{The operation used in the proof of Lemma~\ref{lowerpos}, using a upper swap to maintain virtual welfare and create lower positives.}\label{fig:swaps}
\end{figure}


The second lemma proves that a best dual can't have upper swaps:

\begin{lemma}The optimal dual that has the lowest positives does not contain any upper swaps. \label{lowerpos}\end{lemma}


\begin{proof}
Proof by contradiction. Suppose an upper swap exists.  Then (as depicted on the right in Figure~\ref{fig:swaps}) we can push up $\alpha$ flow from $(y,B)$ to $(x,B)$, causing $f_B(v) \hat \Phi^{\lambda,\alpha}_B(v) = f_B(v) \Phi^{\lambda,\alpha}_B(v) - \alpha$ for $v \in [y, x]$ and improving virtual welfare by $\alpha(x-y)$.  To leave the flow out of item $C$ unchanged, we balance this out by pushing $\alpha$ flow down from $(x,A)$ to $(y,A)$, causing $f_A(v) \hat \Phi^{\lambda,\alpha}_A(v) = f_A(v) \Phi^{\lambda,\alpha}_A(v) + \alpha$ for $v \in [y, x]$.  

If $y$ is unironed at $A$, that is, $\bar{y}_A = y$, or if $\bar{y}_A < \bar{r}_A$, then by choosing $\alpha = - f_A(\bar{y}_A) \Phi^{\lambda,\alpha}_A(\bar{y}_A)$, this will cause $\hat{\bar{r}}_A = \bar{y}_A$, lowering the positives.  

Alternatively, if $y$ is ironed up to $\bar{r}_A$ such that $\bar{y}_A = \bar{r}_A$, then we can choose a very small $\alpha$ to keep the interval $[\underline{y}_A, \bar{r}_A]$ ironed, making the whole interval positive and causing $\hat{\bar{r}}_A = \underline{y}_A$, lowering the positives.  The dual will only increase by $\alpha(x-y)$, even when the values are ironed around, as ironing preserves virtual welfare.  This is canceled out by the improvement in virtual welfare from item $B$.  Then we have maintained virtual welfare but lowered the positives, showing that this dual solution could not have had the lowest positives.
\end{proof}

Lemma~\ref{DS} and Lemma~\ref{lowerpos} comprise the proof of Proposition~\ref{prop:bestdual}.
\end{proof}

Now we prove that our primal recovery algorithm always succeeds in finding an optimal primal (that satisfies complementary slackness) when given a best dual.

\begin{proof}[Proof of Theorem~\ref{thm:allocalg}]
First, consider the case where there exists some point $v$ where $\Phi^{\lambda,\alpha}_A(v) = \Phi^{\lambda,\alpha}_B(v) = 0$, and $v$ is unironed both in $A$ and in $B$.  Then we simply set $v$ as a price for both $A$ and $B$, automatically satisfying the complementary slackness conditions of flow into $A$ or $B$, as both are equally preferable. Since both items $A$ and $B$, have the same allocation rule, the instance degenerates into a FedEx instance. Thus, an optimal allocation rule for the item $C$ can be determined.

Otherwise, the dual solution contains no point $v$ as described in the first case, meaning that ironed intervals interleave throughout the region of zero virtual values.  This means that, if without loss of generality $\bar{r}_A > \bar{r}_B$, that $\bar{r}_B = x$ must sit in an ironed interval $[\underline{x}_A, \bar{x}_A]$ on $A$.

If the top chain is empty, then we have $\bar{r}_A > \bar{r}_B > \underline{x}_A$ with no flow into $A$ for any $v \in [\underline{x}_A, \bar{x}_A]$.  Then, setting 
$$a_A(v) = \begin{cases} 1 & v \geq \bar{x}_A \\ \frac{\bar{r}_A - \bar{r}_B}{\bar{r}_A - \underline{x}_A} & v \in [\underline{x}_A, \bar{x}_A) \\ 0 & \text{otherwise} \end{cases} \quad \quad \text{and} \quad \quad a_B(v) = \begin{cases} 1 & v \geq \bar{r}_B \\ 0 & \text{otherwise} \end{cases}$$
makes both options equally preferable for all $v$ except for $v \in [\underline{x}_A, \bar{x}_A]$, where reporting $B$ is strictly preferable, but this does not violate complementary slackness by the assumption that the top chain is empty.

Otherwise, the top chain is non-empty.  A dual gives a system of utility inequalities via complementary slackness which the allocation rule must satisfy.  Instead, we can solve a system of utility equalities given by the chain via induction on the length of the top chain, and this will imply a solution that satisfies all of the inequalities.  More specifically, the following will hold for top chains of all lengths:
\begin{enumerate}
\item The allocation rule will only increase at the bottom of ironed intervals in the chain.  That is, if the allocation rule increases at $z$, so $a'_A(z) > 0$, then $z$ must be the bottom of an ironed interval for a point $(x,A)$ in the top chain, thus $z = \underline{x}_A$, and $a_A(x) = a_A(\underline{x}_A)$. \label{support}
\item We will fully allocate to all positive virtual values.  That is, $a_A(\bar{r}_A) = a_B(\bar{r}_B) = 1$.  \label{positives}
\item If $(x,A)$ is followed by $(y,B)$ in the chain, then $a_A(x) = a_A(\underline{x}_A) > a_B(y) = a_B(\underline{y}_B)$. \label{decreasing with chain}
\item At any point $(x,A)$ in the top chain, we will have $u_A(x) = u_B(x)$. \label{equality constraints}
\item An alternative solution can, for the first point in the chain $(x,A)$, vary $a_A(\underline{x}_A)$ such that the utility constraint is a strict inequality $u_A(x) > u_B(x)$, and instead we have equality at $\bar{r}_A$: $u_A(\bar{r}_A) = u_B(\bar{r}_A)$.  This gives an equal expected price for the two items, and equal utility for all values $v \geq \bar{r}_A$.  \label{first point} 
\end{enumerate}

To satisfy complementary slackness, for any type $(x,A)$ with flow in, it must be that $u_A(x) \geq u_B(x)$.  We now show why (\ref{decreasing with chain}-\ref{equality constraints}) imply that complementary slackness will be satisfied everywhere.

Consider a subsequence of points in the chain: $(x, B), (y,A), (z, B)$, hence $y > \underline{x}_B$ and $z > \underline{y}_A$.  Then $a_B(x) > a_A(y) > a_B(z)$ by (\ref{decreasing with chain}).  Since $u_A = u_B$ for every point in the chain and a larger allocation rule implies a larger change in utility, we can deduce that $u_A(v) \geq u_B(v)$ for all $v \in [z, y]$.

\begin{itemize}
\item For $v \in (\underline{y}_A, \underline{x}_{B})$, we have that $a_A(v) > a_B(v)$, and since $u_A(z) = u_B(z)$, then $u_A(v) \geq u_B(v)$ in this region.

\item For $v \in (\underline{x}_{B}, y)$, we have that $a_B(v) > a_A(v)$, and since $u_A(y) = u_B(y)$, then $u_A(v) \geq u_B(v)$ in this region.

\item By definition of a double swap, there is no $v \in [\underline{y}_A, z)$ such that there is flow into $(v,A)$.  Likewise, there is no $v \in [\underline{x}_{B}, y)$ such that there is flow into $(v,B)$.  
\end{itemize}
Hence all possible complementary slackness conditions are satisfied.


\begin{figure} [h!]
\begin{center} \includegraphics[scale=.4]{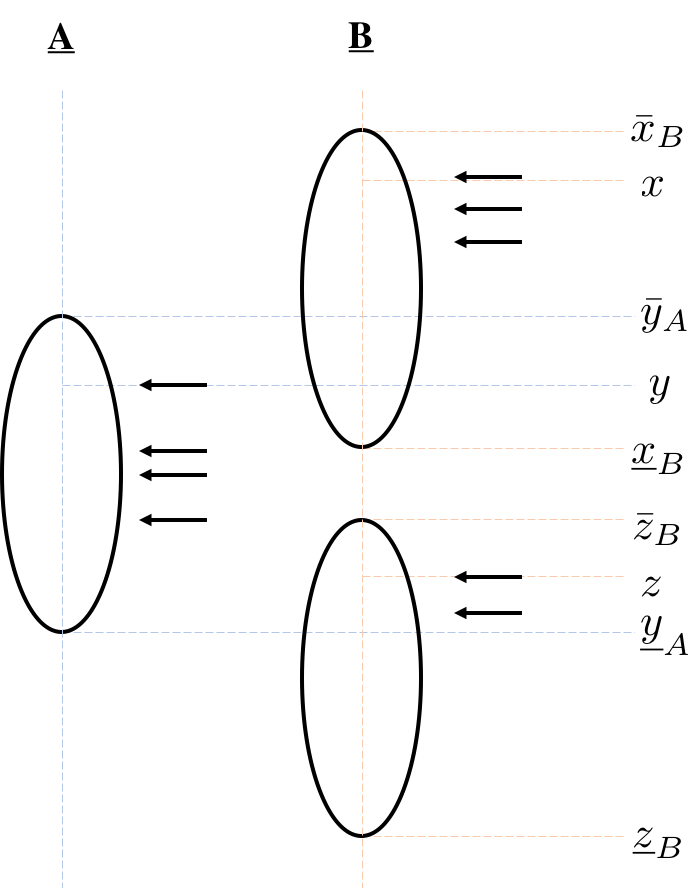} \quad\quad\quad\quad \includegraphics[scale=.4]{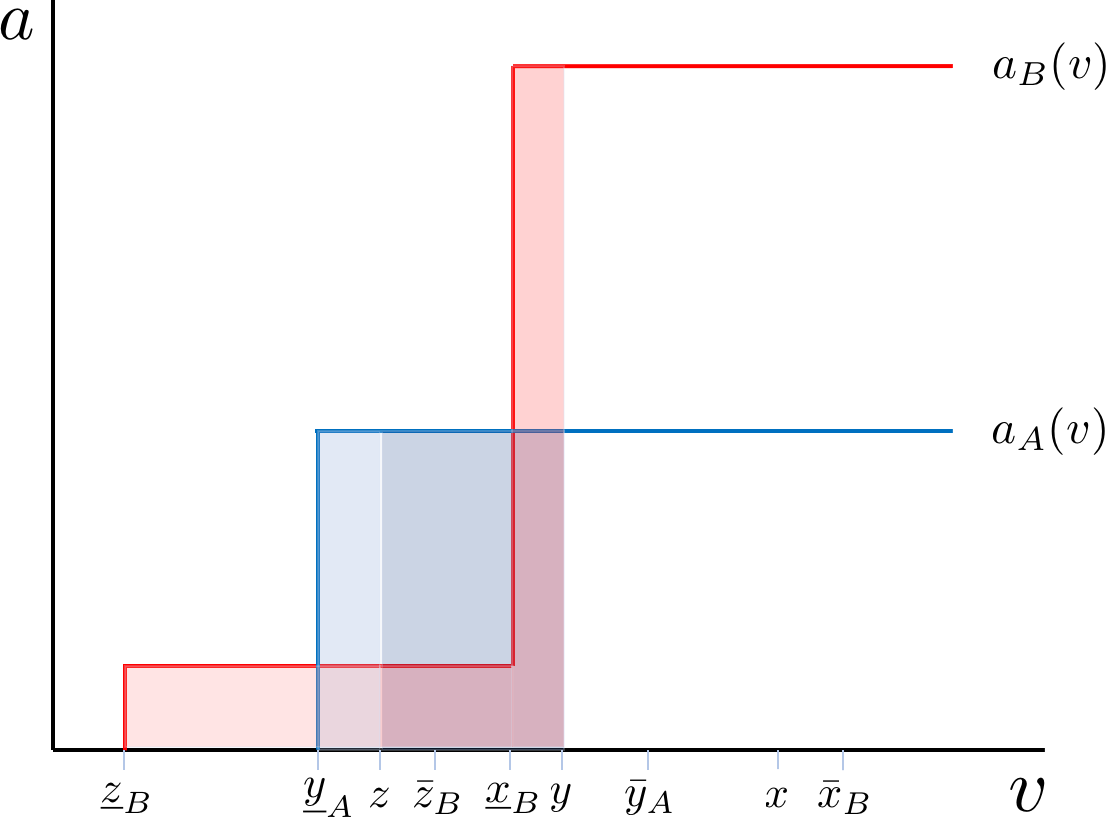} \end{center}
\caption{Left: A candidate dual (with no double or upper swaps); part of a chain.  Right: An allocation that satisfies complementary slackness up to value $y$, satisfying equal preferability at $z$ and $y$ and preferability at all points with flow in.}\label{fig:CSconstraints}
\end{figure}

We now show that these sufficient properties hold by induction.  As a base case, consider when there is one point in the top chain, which without loss is $(x, A)$.  By definition of the top chain, $\bar{r}_A > x > \bar{r}_B > \underline{x}_A$ and there is flow into item $A$ at $x$, which is in ironed interval $[\underline{x}_A, \bar{x}_A]$.  We can set $a_A(\underline{x}_A) = \frac{x - \bar{r}_B}{x - \underline{x}_A}$ and set $a_A(\bar{r}_A) = a_B(\bar{r}_B) = 1$.  Then $$u_A(x) = a_A(\underline{x}_A) \cdot (x - \underline{x}_A) = \frac{x - \bar{r}_B}{x - \underline{x}_A} \cdot (x - \underline{x}_A) = 1 (x - \bar{r}_B) = u_B(x).$$


\begin{figure} [h!]
\begin{center} \includegraphics[scale=.5]{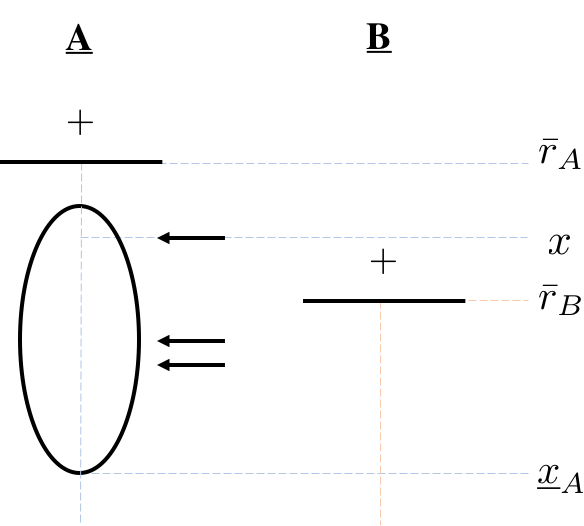} \quad\quad\quad\quad\quad\quad \includegraphics[scale=.3]{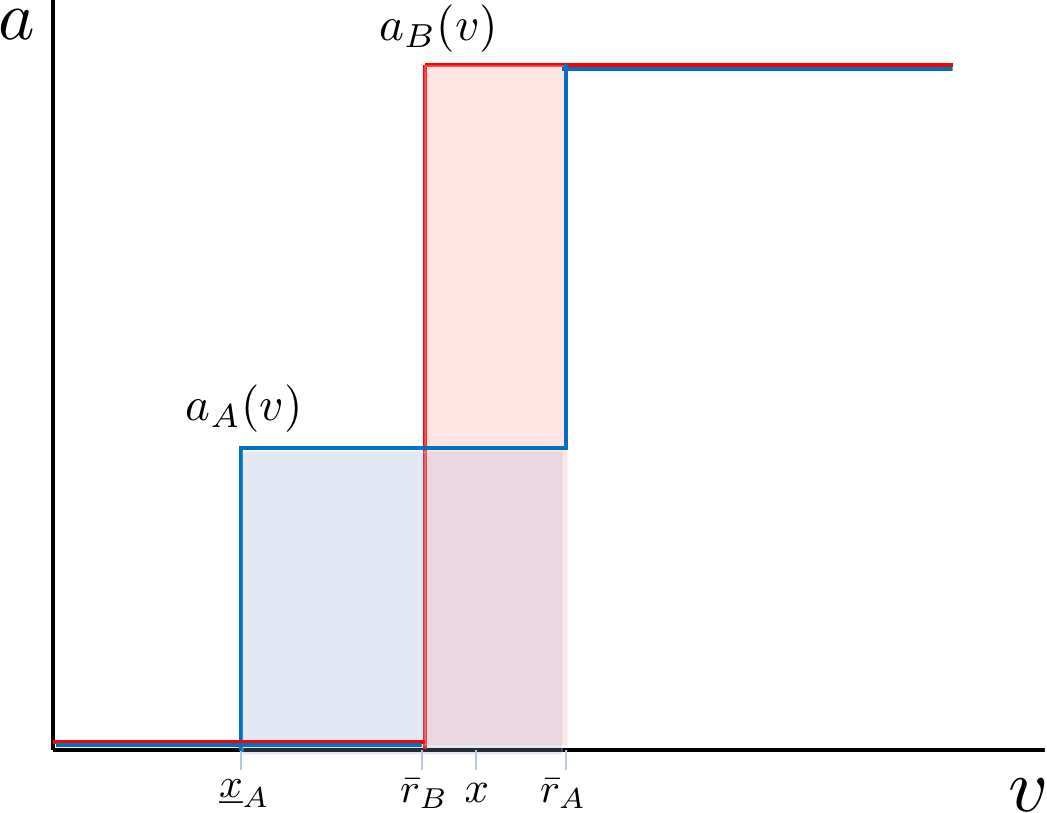} \end{center}
\caption{Left: The base case of a candidate dual with an empty chain.  Right: An allocation that satisfies complementary slackness.}\label{fig:basecaserecovery}
\end{figure}

Then conditions (\ref{support}-\ref{equality constraints}) are met.  To satisfy (\ref{first point}), we can instead set $a_A(\underline{x}_A) = \frac{\bar{r}_A - \bar{r}_B}{\bar{r}_A - \underline{x}_A}$.  Then $$u_A(\bar{r}_A) = a_A(\underline{x}_A) \cdot (\bar{r}_A - \underline{x}_A) = \frac{\bar{r}_A - \bar{r}_B}{\bar{r}_A - \underline{x}_A} \cdot (\bar{r}_A - \underline{x}_A) = 1 (\bar{r}_A - \bar{r}_B) = u_B(\bar{r}_A).$$

For the inductive hypothesis, suppose for any chain of length $n$, we have allocation rules such that (\ref{support}-\ref{first point}) hold. 



Now consider a chain of length $n+1$.  Without loss of generality, let $(x,A)$ be the top point in the chain, where $x$ sits in the ironed interval $[\underline{x}_A, \bar{x}_A]$, and this point is proceeded by $(y,B)$ which sits in $[\underline{y}_B, \bar{y}_B]$, hence $\bar{r}_A > x > \bar{r}_B$ and $y > \underline{x}_A$ by definition of the chain.


By the inductive hypothesis, we can come up with allocation rules $a_A(\cdot)$ and $a_B(\cdot)$ that satisfy complementary slackness to the same chain without the highest point $(x,A)$, and will have $a_A(\underline{x}_A) = a_B(\bar{r}_B) = 1$.  We construct an allocation rule $\hat{a}$ for the top chain of size $n+1$ as follows; this is depicted in Figure~\ref{fig:ISrecovery}.  Let $\lambda = \frac{x-\bar{r}_B}{x-y - a_B(\underline{y}_B)(\bar{r}_B - y)} < 1$.  Then let $$\hat a_A(v) = \begin{cases} 1 & v \geq \bar{r}_A \\ \lambda a_A(v) & \text{otherwise} \end{cases} \quad \quad \text{and} \quad \quad \hat a_B(v) = \begin{cases} 1 & v \geq \bar{r}_B \\ \lambda a_B(v) & \text{otherwise.} \end{cases}$$
This clearly satisfies (\ref{support}-\ref{decreasing with chain}).  To show that (\ref{equality constraints}) holds, we observe that at any previous point of concern $v < \bar{r}_B$, we had $u_A(v) = u_B(v)$.  Now at those points, we have $\hat u_A(v) = \int _0 ^v \hat a_A(v) dv = \lambda \int _0 ^v a_A(v) dv = \lambda u_A(v)$.  This holds for $\hat u_B(v) = \lambda u_B(v)$ as well.  Thus, complementary slackness is still satisfied at all previous points $v\leq \bar{r}_B$; we only need to check equal utility at $x$.
\begin{align*}
\hat u_A(x) &= \hat u_A(y) + \hat a_A(\underline{x}_A) (x-y) = \lambda u_A(y) + \lambda \cdot 1 \cdot (x-y) \\
\hat u_B(x) &= \hat u_B(y) + \hat a_B(\underline{y}_B) (\bar{r}_B-y) + \hat a_B(\bar{r}_B) (x -\bar{r}_B) = \lambda u_B(y) + \lambda \cdot a_B(\underline{y}_B) (\bar{r}_B-y) + 1 \cdot (x -\bar{r}_B)
\end{align*}
Then to have $\hat u_A(x) = \hat u_B(x)$, since $u_A(y) = u_B(y)$, we require that  $$\lambda (x-y) =  \lambda \cdot a_B(\underline{y}_B) (\bar{r}_B-y) + 1 \cdot (x -\bar{r}_B).$$ The solution here is exactly the $\lambda$ defined above.

Alternatively, by replacing $x$ with $\bar{r}_A$, thus setting $\lambda = \frac{\bar{r}_A -\bar{r}_B}{\bar{r}_A-y - a_B(\underline{y}_B) (\bar{r}_B-y)}$, we get a solution that has $u_A(x) > u_B(x)$ and $u_A(\bar{r}_A) = u_B(\bar{r}_A)$ as required in (\ref{first point}).

Thus we have ensured that for top chains of all lengths, we can give an allocation rule that satisfies complementary slackness for all values from the bottom to the top of the chain.  For $v$ below the chain, $u_B(v) = u_A(v) = 0$, so we automatically satisfy complementary slackness.  Above the chain, if we have used the alternate solution that (\ref{first point}) guarantees exists, we automatically satisfy complementary slackness for $v \geq \bar{r}_A$.  This would only fail if there is flow into item $B$ for $v \in [x, \bar{r}_A)$---that is, if the dual contains a upper swap, but by assumption it does not.  Then for any dual solution with no double swaps or upper swaps, this algorithm gives an allocation rule that satisfies complementary slackness. 


We prove that the menu complexity of the mechanism output by this algorithm is finite below:

\begin{claim}The menu complexity is always finite. \end{claim}

\begin{proof} Proof by contradiction. Suppose that there exists an instance such that the mechanism output by the algorithm has infinite menu complexity. 

Note that this can only happen if the length of the top chain is infinity. Thus, there exists a sequence of points $(x_1, A), (x_2, B), (x_3, A), \cdots$ such that the point $(x_i, A)$ is inside an ironed interval $[\underline{x_i}_A, \overline{x_i}_A]$ and $x_{i+1} \geq \underline{x_i}_A$. Analogous claims hold for an element $(x_{i+1}, B)$ in the chain.

Thus, we have

\[x_1 \geq \overline{r}_B \geq \overline{x_2}_B \geq x_2 \geq \underline{x_1}_A \geq \overline{x_3}_A \geq x_3 \geq \cdots\]

Since the infinite sequence $x_1, x_2, \cdots$ is monotone and bounded, it converges to a limit, say $x^*$.  Observe that $x^*$ satisfies $\Phi^{\lambda, \alpha}_A(x^*) = \Phi^{\lambda, \alpha}_B(x^*) = 0$ and is unironed. This is because points arbitrarily close to it are unironed and are zeroes of  $\Phi^{\lambda, \alpha}_A(\cdot)  $ and $\Phi^{\lambda, \alpha}_B(\cdot)$. However, in this case, our algorithm just sets the price $x^*$ and thus has constant menu complexity, a contradiction.

\end{proof}

\end{proof}

\begin{figure} [h]
\begin{center}\includegraphics[width=5cm]{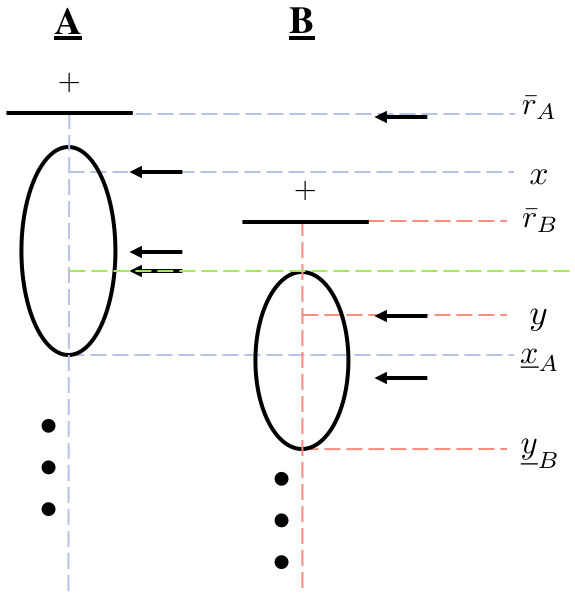} \hspace{1.5cm} \includegraphics[width=5.5cm]{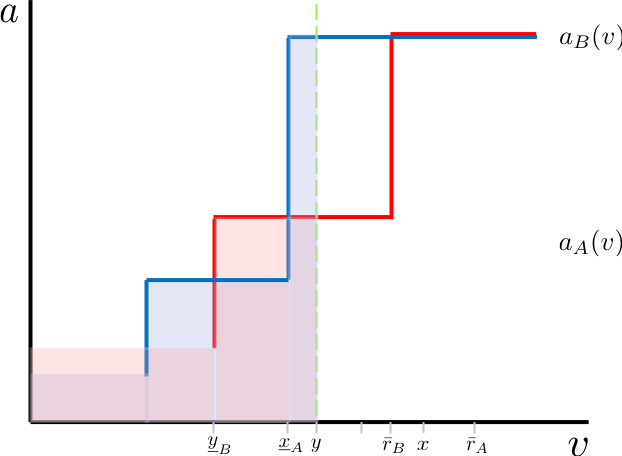} \\ 
\vspace{.5cm} \includegraphics[width=5.5cm]{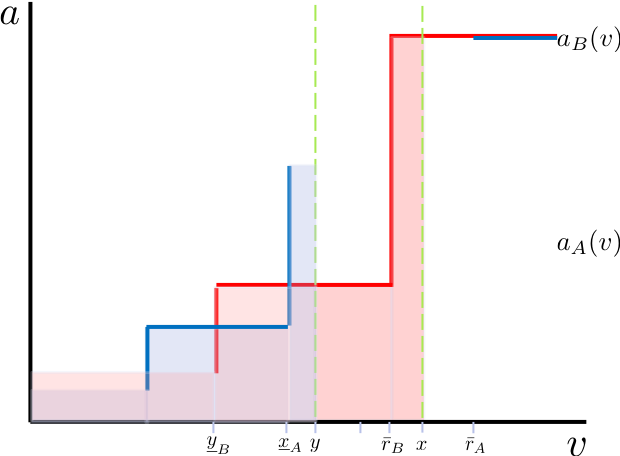} \hspace{1cm} \includegraphics[width=5.5cm]{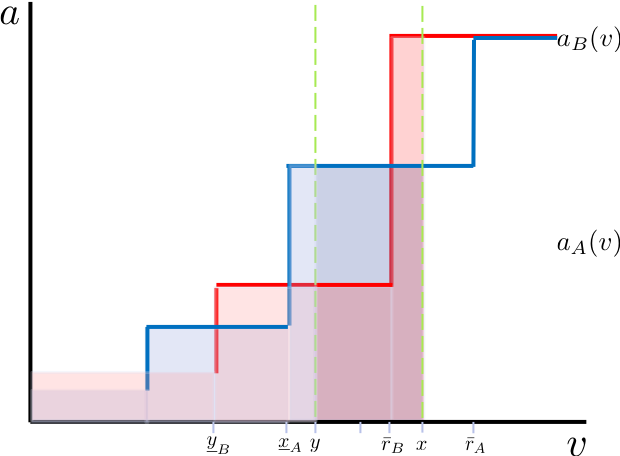}  \end{center}
\caption{Top Left: A top chain from a candidate dual.  We use the inductive hypothesis on the chain of one size smaller (below the green line).  Top Right: The allocation rule from the inductive hypothesis that satisfies all CS constraints on the smaller chain (below the green line).  Bottom Left: The scaled allocation rule, requiring preferability of $A$ between the green lines.  Bottom Right: The allocation rule that satisfies these preferability constraints.}\label{fig:ISrecovery}
\end{figure}

\section{Equivalence with Single-Minded Valuations} \label{sec:singleminded}

In the introduction, we note the following observation.

\begin{observation}\label{obs:SMequiv}
The partially-ordered setting is \emph{equivalent} to the single-minded setting.
\end{observation}

First, we define the single-minded setting.

\begin{definition}
In a \emph{single-minded setting}, a seller determines how to sell any bundle of $k$ items.  A buyer has a (value, bundle) pair $(v,B)$ where $B \in 2^{[k]}$ is any subset of items.  The pair $(v,B)$ is drawn from a joint probability distribution over $[0,H]\times 2^{[k]}$ where $H$ is the maximum possible value of any bidder for any item.
\end{definition}

Any single-minded setting can be represented as a partially-ordered setting: the set of possible interests $\G$ is just the set of possible bundles, $2^{[k]}$.  The relation is set containment: an interest $G$ dominates an interest $G'$, that is, $G \succ G'$, if $G \supset G'$.  The distribution $F$ is identical.

Any partially-ordered setting can be represented as a single-minded setting:  we can invent items such that every interest $G$ maps to some subset of items.  For any minimal interest $G$ (that is, $G$ which does not dominate any other interests), map $G$ to a new item $i$: $B(G) = \{i\}$.  For each successive interest $G' \in \children(G)$, map $G'$ to $B(G') = \{j\} \cup \bigcup _{G'' : G' \in \children(G'')} B(G'')$ where $j$ is a new additional item.  Repeat this process, completing a mapping from interests to subsets of some $m$ created items.  For all subsets $B \in 2^{[k]}$ which do not have an interest that maps onto it, assign measure $0$ to the event of drawing $(v,B)$ from $F$.  Otherwise, $f_{B(G)}(v) = f_G(v)$.


\section{The Master Theorem}
\label{sec:masterthmApp}

All of the analysis in the previous section started from a candidate dual solution, and showed that such duals are optimal (as in, there is a feasible primal satisfying complementary slackness). The missing step is ensuring that there exists an input distribution for which these duals are feasible. To save ourselves (and future work) the tedium of hand-crafting an actual distribution for which these duals are feasible, we prove a general Master Theorem, essentially stating that for a wide class of duals 
(essentially, anything dictated by ironed intervals, positive/negative regions, and flow in),
there exists a distribution for which this dual is feasible.  
\begin{theorem}[Master Theorem]
\label{thm:multimaster}
Suppose we are given a partial order over $\mathcal{G}$, for each item $G \in \mathcal{G}$ candidate endpoints of zero region (bounded away from $0$) $\bar{r}_{G}, \underline{r}_{G}$, a finite set of candidate ironed intervals (bounded away from zero) $[\underline{x}_{i,G},\overline{x}_{i,G}]$ with $\underline{r}_{G} \leq \underline{x}_{i,G} \leq \overline{x}_{i,G} \leq \bar{r}_{G}$, and for each pair of items $G' \succ G$ a finite set of candidate flow-exchanging points (bounded away from zero) $y_{i,G,G'}$ \emph{not in} $(\underline{x}_{i,G}, \overline{x}_{i,G}]$ for any candidate ironed interval. Then there exists a joint distribution over (value, interest) pairs with a feasible dual $(\lambda, \alpha)$ such that:
\begin{itemize}
	\item the endpoints of the zero region for $\Phi^{\lambda, \alpha}_G$ are $\underline{r}_G$ and $\bar{r}_G$.
	\item the ironed intervals of $\Phi^{\lambda, \alpha}_G$ are exactly to the intervals $[\underline{x}_{i,G},\overline{x}_{i,G}]$ (no others). 
	\item $\alpha_{G,G'}(y) > 0 \Leftrightarrow y = y_{i, G, G'}$ for some $i$.
\end{itemize}
\end{theorem}
Note that from the proof the Master Theorem, it is clear how to explicitly construct a distribution for the lower bound (although this is a tedious and unilluminating process).

In this section we provide a complete proof of Theorem~\ref{thm:multimaster}. On our way to prove this theorem, we generalize a result of~\cite{SSW17}, in which they show that for totally ordered preferences, one can always find a discrete distribution that produces a well-enough-behaved revenue curve. They use this result to show that there exist instances for which the menu complexity is the worst possible, exponential in the number of items. Here we extend their construction and show that for any well-enough-behaved set of continuous revenue curves for the partially ordered setting, there exist distributions that induce them.


The first step is to generalize the result of~\cite{SSW17} from discrete distributions to continuous distributions. 

\begin{lemma}[Revenue Theorem for Continuous Curves]
\label{lem:fromRtoF}
Given a continuous curve $R : [ 1, H ]$ differentiable everywhere except at countably many points, such that $R(1) = 1$ and $|R'(x)_{+}|, |R'(x)_{-}| \leq \frac{1}{2H} \forall x \in [1, H]$, there exists a distribution $\mathcal{F}$ such that $R$ is the revenue curve that arises from selling to a single bidder with a valuation drawn from $\mathcal{F}$. 
\end{lemma}

\begin{proof}
Consider the following distribution
$$F(x) = 1 - \frac{R(x)}{x}, x \in [1,H]$$
and $F(x) = 0$ for $x \leq 1$, $F(x) = 1$ for $x \geq H$. In order to show that this is a valid distribution, it suffices to show that it is monotonic non-decreasing. For that, we consider its derivative and show it is non-negative everywhere: 

$$F'(x) = \frac{-x R'(x) +R(x) }{x^2}.$$ 

It suffices to show that the numerator, $R(x) -R'(x)x$, is always non-negative. Note that for $x \geq 1$, $R(x) \geq \frac{1}{2}$ (since $R(1) = 1$ and the derivative doesn't change fast enough) and $|R'(x)_+| \leq \frac{1}{2H}$. Since $x \leq H$, the claim follows. 

It remains to show that indeed the revenue from this distribution matches the curve $R(x)$. Consider setting a price of $x$, then the revenue of selling at $x$ is exactly $x(1-F(x)) = R(x)$.  
\end{proof} 

Now we want to extend this to say we can find distributions for revenue curves with specific properties that will be useful.

\begin{theorem}[Master Theorem for Single Item]
\label{theorem:singlemaster}
Given candidate endpoints of zero region $x_{+}, x_{-}$ and candidate ironed interval endpoints $[\underline{x}_{i}, \overline{x}_{i}]_{i=1}^{k}$ (where $x_{-} \leq \underline{x}_{i} \leq \overline{x}_{i} \leq x_{+}$) there is a distribution $\mathcal{F}$ such that the revenue curve induced by a bidder whose valuation is drawn from $\mathcal{F}$ satisfies
\begin{itemize}
\item $\Phi^{\lambda,\alpha}(x) f(x)$ is negative for $x < x_{-}$ (i.e. $x_{-}$ is the lower endpoint of the zero region), 
\item $\Phi^{\lambda,\alpha}(x) f(x)$ is positive for $x > x_{+}$ (i.e. $x_{+}$ is the upper endpoint of the zero region) and, 
\item the ironed intervals correspond exactly to the intervals $[\underline{x}_{i}, \overline{x}_{i}]$ for $i =1$ to $k$.
\end{itemize}
\end{theorem}

\begin{proof}
We will reduce the problem of finding a valid distribution to that of constructing a revenue curve that will guarantee these properties and then apply Lemma~\ref{lem:fromRtoF}. Consider the following revenue curve:

\begin{equation*}
R(x) =\begin{cases}
      x & 0\leq x \leq 1, \\
      1+\frac{x}{2H} & 1 \leq x \leq x_{-}, \\
      1+ \frac{x_{-}}{2H} & x_{-} \leq x \leq \underline{x}_{1} \\ 
      1+ \frac{x_{-}+\underline{x}_{1}-x}{2H} & \underline{x}_{1} \leq x \leq \frac{\underline{x}_{1}+\overline{x}_{1}}{2} \\
      1+ \frac{x_{-}+x-\overline{x}_{1}}{2H} & \frac{\underline{x}_{1}+\overline{x}_{1}}{2} \leq x \leq \overline{x}_{1} \\
      \dots \\ 
       1+\frac{x_{-}}{2H} & \overline{x}_{i-1}  \leq x \leq \underline{x}_{i} \\ 
      1+ \frac{x_{-}+\underline{x}_{i}-x}{2H} & \underline{x}_{i} \leq x \leq \frac{\underline{x}_{i}+\overline{x}_{i}}{2} \\
      1+ \frac{x_{-}+x-\overline{x}_{i}}{2H} & \frac{\underline{x}_{i}+\overline{x}_{i}}{2} \leq x \leq \overline{x}_{i} \\
      \dots \\ 
      1+\frac{x_{-}}{2H} & \overline{x}_{k} \leq x \leq x_{+} \\
      1+\frac{x_{-}}{2H}-\frac{x-x_{+}}{2H(H-x_{+})}(x_{-}+1) & x_{+} \leq x \leq H.\\       
   \end{cases} 
\end{equation*}

This revenue curve is such that $R(1) = 1$ and $|R'(x)| \leq \frac{1}{2H}$ for $x \in [1, H]$. This allows us to claim that there is a distribution that induces this revenue curve. Moreover, from the way we constructed this revenue curve, the derivative is positive from $0$ to $x_{-}$, negative from $x_{+}$ to $H$, goes from negative to positive for the intervals $[\underline{x}_{i}, \overline{x}_{i}]$ and is $0$ elsewhere. We will show that these conditions are sufficient to make the virtual values take the signs we intend them to. 


It suffices to note that the sign of the derivative of the revenue at $x$ is the opposite of the sign of the virtual value at $x$ (noted in Definition~\ref{def:revcurve}). By construction, our revenue curve has negative slope on values higher than $x_{+}$ and positive slope on points below $x_{-}$. The intervals in between will be ironed and turn into $0$ slope intervals. 
\end{proof} 

\begin{remark}
It is possible to relax the condition that all ironed intervals are between $x_{-}, x_{+}$. It is not hard to see how to adapt the proof to have ironed intervals either in $[1, x_{-}]$ or $[x_+, H]$. It is sufficient to add dimpled intervals, like the ones in our construction, as the revenue curve is increasing or decreasing. We don't need them for our main result, hence don't worry about this more general result. Likewise, the revenue curve $R$ could be made differentiable everywhere if we used a smoother function to transition between the ironed and non-ironed intervals, as opposed to straight lines.
\end{remark}



\begin{proof}[Proof of Theorem~\ref{thm:multimaster}]
If the constraint over flows wasn't there, the problem would be a direct application of Theorem~\ref{theorem:singlemaster}. Unfortunately, the flow constraints may affect the virtual values of neighboring items. It is not hard to predict how outgoing and incoming flow will change the virtual values for the different items. From the study of duality in this context we know that if there is $\varepsilon$-flow leaving from $(y_i, G)$ to $(y_i, G')$ (where $G' \in N^{+}(G)$), then the virtual values of all points of item $G$ with $y \leq y_i$ will increase by $\varepsilon$ and all points of item $G'$ with $y \leq y_i$ will decrease by $\varepsilon$. Thus, given that we know what we want the revenue curves to look like after all flow has been sent, we can reverse engineer what they must look like in order to make that happen. In particular, since the flows shift the virtual values by a constant it will suffice to subtract a function whose value is $0$ before $y_i$ and becomes a line with small, negative slope at $x_i$ (say, slope $\varepsilon = \frac{1}{2^H}$) from the ``suggested'' (by Theorem~\ref{theorem:singlemaster}) revenue curve for item $G$ (since these will increase by $\varepsilon$ after the flow is sent) and add positive slope functions of the same value at $x_i$ on item $G_{i,G}$ from its suggested revenue curve (since these will decrease by $\varepsilon$ after the flow is sent). This is sufficient because of the connection between virtual values and revenue curves argued before: the derivative corresponds to changes in the virtual value. So for a constant change in virtual value, the matching change would be adding a linear term to the revenue curve of opposite sign. The order in which we do these changes is by processing items from leaves to the root (i.e. only process a node once all its children have been processed) and within an item $G$, address the flow-exchange values from smallest to largest. 
\end{proof}



We abuse this opportunity to prove a similar result for the multi-unit pricing setting. 

\begin{theorem}[Master Theorem for MUP]
\label{thm:mupmaster}
Suppose we are given a MUP instance where the buyer can get up to $n$ copies of an item. Let $G_i$ for $1 \leq i \leq n$ be the item corresponding to $i$ copies. For each item $G_i$ we are given candidate endpoints of the zero region $x_{-i}, x_{+i}$ and a set of candidate ironed interval endpoints $[\underline{x}_{j,i},\overline{x}_{j,i}]_{j=1}^{k_i}$ with $x_{-i} \leq \underline{x}_{j,i} \leq \overline{x}_{j,i} \leq x_{+i}$. Moreover, for each tuple $(i, i+1)$ and $(i,i-1)$, we are given a set of candidate flow-exchanging points $y_{j,i,i+1}$ and $y_{j,i,i-1}$ not in $(\underline{x}_{j,i},\overline{x}_{j,i}]$ for any candidate ironed interval. Then, there exists distributions $\mathcal{F}_G$ for all items $G$ such that: 
\begin{itemize}
	\item the endpoints of the zero region for $G_i$ correspond to $x_{-i}, x_{+i}$, 
	\item the ironed intervals correspond exactly to the intervals $[\underline{x}_{j,i},\overline{x}_{j,i}]_{j=1}^{k_i}$ (and no other), 
	\item the dual of the problem is such that there $\alpha_{G_i, G_{i+1}}(y_{j,i,i+1}) \geq 0$ (i.e. there is flow sent from $G_i$ at $y_i$ to $G_{i+1}$ into $y_{j,i,i+1}$ and no other flow from $i$ to $i+1$). 
	\item the dual of the problem is such that there $\alpha_{G_i, G_{i-1}}(y_{j,i,i-1}) \geq 0$ (i.e. there is flow sent from $G_i$ at $y_i$ to $G_{i-1}$ into $\frac{i-1}{i} y_{j,i,i-1}$ and no other flow from $i$ to $i-1$). 
\end{itemize}
\end{theorem}

\begin{proof}
This proof is similar to that of~\ref{thm:multimaster} with the exception that on the former, increasing the flow from $(v,G)$ to $(v,G')$ (with $G' \in N^{+}(G)$) by a little bit increases and decreases the virtual values below $v$ by the same amount. This is no longer true since we are moving from $(y_{j,i,i-1}, G_i)$ to $(\frac{i-1}{i} y_{j,i,i-1}, G_{i-1})$. In this case, sending $\varepsilon$ flow from $(y_{j,i,i-1}, G_i)$ to $(\frac{i-1}{i} y_{j,i,i-1}, G_{i-1})$ increases the virtual values below $(y_{j,i,i-1}, G_i)$ by $\varepsilon$ but decreases the ones on the other end by only  $\frac{i-1}{i} \varepsilon$. So, in order to reverse engineer the change in virtual value induced by this setting we need to add the same functions as in the proof of Theorem~\ref{thm:multimaster} to the revenue curve suggested for $G_i$ and add a $\frac{i}{i-1}$-scaled version of it for the receiving item at the point  $(\frac{i-1}{i} y_{j,i,i-1}, G_{i-1})$ on the revenue curve for $G_{i-1}$. The order in which these we do these changes is by processing items from leaves to root (i.e. from $G_n$ to $G_1$) and within a item $G_i$, address the flow-exchange points from smallest to largest.
\end{proof}

\section{A Candidate Dual for a Lower Bound on Menu-Complexity for the Multi-Unit Pricing Problem}\label{sec:MUPLB}

Consider an MUP instance where the buyer can get one, two, or three copies of a given item. The relevant complementary slackness constraints in this setting go from 

\begin{itemize}
\item {\bf Rightwards.} For all $v$, from $(v,1) \to (v,2)$ and $(v,2) \to (v,3)$. This is because a buyer can always misreport and get more items.
\item {\bf Leftwards.} For all $v$, from $(v,2) \to (v/2,1)$ and $(v,3) \to (2v/3,2)$. This is because a buyer would prefer getting fewer items if they are available for much cheaper.
\end{itemize}

As shown in \citepalias{DHP}, a buyer of type $(v,C)$'s utility for reporting $(v/2, A)$ is given by $u_A(v/2) = \int _0 ^{v/2} a_A(x) dx$.  The same buyer's utility for reporting $(v, B)$ is given by $u_B(v) = 2 \int _0 ^{v} a_B(x) dx$.

To construct a lower bound for the MUP instance, we adapt our construction from the partially ordered case. We describe our construction formally below, but note here all the relevant differences. Observe that the incentive compatibility constraints for the MUP instance described above hide a partially ordered instance inside them. Indeed, the `item' $2$ is analogous to the item $C$, while the items $A$ and $B$ are the items $1$ and $3$ respectively. Just like the partially ordered instance, there are incentive compatibility constraints from $(v,2) \to (v,3)$ for all $v$.  The only difference is that the constraints from $(v, 2) \to (v,1)$ have been replaced by those from $(v,2) \to (v/2,1)$. Also, there are `new' constraints from $(v,1) \to (v,2)$ and $(v,3) \to (2v/3,2)$.

We claim that, despite these changes, the essence of our argument there still holds. Roughly speaking, our argument there involved constructing a top-chain (see Definition~\ref{def:chain}) oscillating between items $A$ and $B$. For any value $x$ in this chain, we had flow coming from $C$ to {\em both} $A$ and $B$. Reasoning about complementary slackness constraints, then, gave us our lower bound.

For the MUP case, we can still do all the above things with the caveat that the value $(v/2,1)$ has to be treated as if it were $(v,1)$.  An analogous master theorem can still be proved as the effect of the `diagonal' flow on the virtual values is predictable. Using the master theorem, we can construct (essentially) any dual we want. Thus, we can have a feasible dual with a top-chain of an arbitrary length $\menuseq$ oscillating between items $1$ and $3$. Also, we have flow from the item $2$ to both $1$ and $3$ at all values in this top-chain. Chasing through the complementary slackness constraints in this dual again gives us a lower bound.

To highlight this analogy, in what follows, we use $C$ instead of $2$, $A$ instead of $1$, and $B$ instead of $3$.

Formally, we construct given an integer $\menu > 0$, a dual containing a top chain among $A$ and $B$ of length $\menuseq$. 
That is, a sequence of points $(x_1, A), (x_2, B), \ldots, (x_{\menuseq}, A)$ such that
$$\underline{x_{\menuseq}/2}_A < \underline{x_{\menuseq-1}}_B/2 < \cdots < \underline{x_2}_B/2 < \underline{x_1/2}_A.$$

In this dual, we have no extra space between the ironed intervals:
\begin{itemize}
\item $\underline{r}_A = \underline{x_\menuseq}_A$/2, $\overline{r}_A = \overline{x_1}_A/2$, and for $i$ such that $(x_i, A)$ and $(x_{i+2},A)$ are in the chain, $\underline{x_i/2}_A = \overline{x_{i+2}/2}_A$.
\item $\underline{r}_B = \underline{x_\menuseq}_B$, $\overline{r}_B = \overline{x_2}_B$, and for $i$ such that $(x_i, B)$ and $(x_{i+2},B)$ are in the chain, $\underline{x_i}_B = \overline{x_{i+2}}_B$.
\end{itemize}
Recall that by definition of the $\overline{\cdot}$ and $\underline{\cdot}$ operators, $(\underline{x}_G, \overline{x}_G]$ is ironed in $G$.  Also by our definitions, $f_G(v) \Phi^{\lambda,\alpha}_G(v) > 0$ for $v \geq \bar{r}_G$; $f_G(v) \Phi^{\lambda,\alpha}_G(v) = 0$ for $v \in [\underline{r}_G, \bar{r}_G]$;  $f_G(v) \Phi^{\lambda,\alpha}_G(v) < 0$ for $v \leq \underline{r}_G$.

We will also define $C$ to be DMR (and thus have no ironed intervals) with $\underline{r}_C = 2 \underline{r}_A$ and $\bar{r}_C = 2 \bar{r}_A$.

We adapt the flow from the partially ordered lower bound example: for any $(x,G)$ in the chain, $\alpha_{C,A}(x \to x/2) > 0$ and $\alpha_{C,B}(x \to x) > 0$. 

\begin{figure}[h!] \label{fig:MUPcandidatedual}
\centering{\includegraphics[scale=.5]{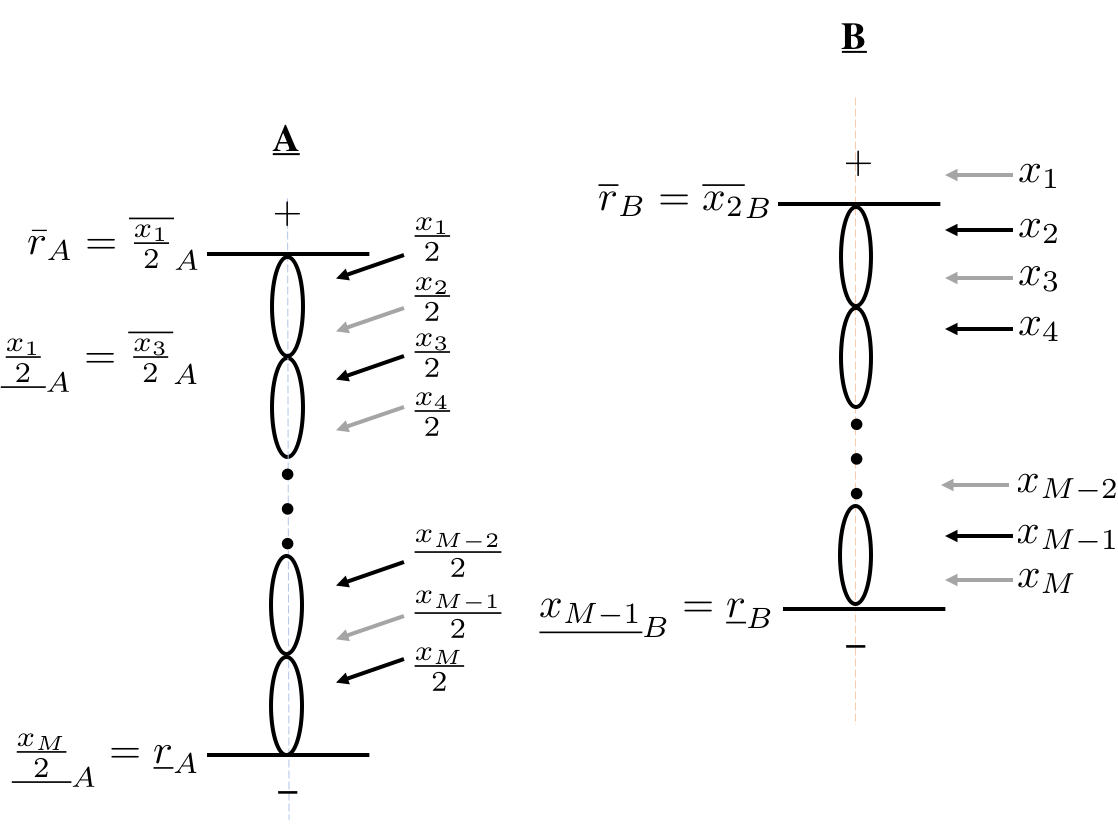}}
\caption{The analogue of the partially ordered candidate dual, adjusted for the Multi-Unit Pricing problem.}
\end{figure}

\begin{theorem} \label{thm:MUPunbounded} To satisfy complementary slackness with the candidate dual, the allocation requires $M$ distinct allocation probabilities; the menu complexity is at least $M$. \end{theorem}

\begin{proof}
The proof is almost identical to that of Theorem~\ref{thm:unbounded}.  Using the constraint that the allocation can't increase in the middle of an ironed interval and that $u_A(x/2) = u_B(x)$ for all $(x,G)$ in the chain, we show that the allocations must be non-zero throughout the chain.

Then, we show that for consecutive points in a chain $(x_i, A), (x_{i+1}, B)$ that $(1/2) a_A(x_i/2) > 2 a_B(x_{i+1})$, and similarly, for $(x_i, B), (x_{i+1}, A)$, that $2 a_B(x_{i}) > (1/2) a_A(x_{i+1}/2)$

This is enough to show that all of the menu options must be distinct, requiring meu complexity $\geq \menu$.
\end{proof}
\section{Coordinated Valuations} \label{sec:raghuvansh}

In this section and the following Appendix~\ref{sec:rrsappendix}, we examine the same minimal partial ordering with $\G = \{A,B,C\}$ where $A \succ C, B \succ C, A \not \succ B, B \not \succ C$.  However, a type $(v,C)$ now has a function $g_A(v)$ and $g_B(v)$ describe his valuations for $A$ and $B$ respectively, and $g_C(v) = v$.  That is, if a buyer with type $(v,C)$ gets item $G$, his utility is $g_G(v)$ times the probability that he is served minus his payment.

The main result is that even for $g_A$ that is piecewise linear with only two segments and for $g_B(v) = v$, the randomization required in the optimal mechanism jumps from unbounded but finite, as it was in the partially-ordered setting, to \emph{countably infinite}.  This further fills out the spectrum, placing this setting between partially-ordered items and two additive items.

However, if $g_A$ and $g_B$ are not piecewise linear, the menu complexity once again jumps up, becoming uncountably infinite, and matching the menu complexity for two additive items.

\subsection{Preliminaries}

Consider selling $3$ items $A$, $B$, and $C$ to one bidder. Define the set $\items = \{A,B,C\}$. We use $\bar{G}$ to refer to a general item $\in \items$. When we make claims about $\bar{G}$, we mean that the claim holds for each of the three items in $\items$. When referring only to items $A$, $B$, we use the symbol $G$. Thus, a claim that holds for $G$ holds for both $A$ and $B$.

In the setting we consider, the bidder has a type in the set $\types = \{A,B,C\} \times [0,H]$, where $H \in \mathbb{R}$ is some constant. The type $(\bar{G}, v)$ of the bidder is drawn from a distribution $\mathfrak{f}$ supported on $\types$. Denote by $q_{\bar{G}} = \int_0^H \mathfrak{f}(\bar{G}, t) dt$. Also, define $f_{\bar{G}}(v) = \frac{1}{q_{\bar{G}}}\int_0^v \mathfrak{f}(\bar{G}, t) dt$ as the distribution $\mathfrak{f}$ conditioned on $q = \bar{G}$. Almost exclusively, we refer to $\mathfrak{f}$ as $(q, f_{\bar{G}})$. We also omit the subscript when it is clear from context.

We now define the bidder's value function $\v: \types \times \items \to [0,H]$. This is defined as
\[\v(C, v,\bar{G}) = g_{\bar{G}}(v) \quad \quad \text{ and } \quad \quad \v(G, v, \bar{G}) = v\cdot \mathbbm{1}_{G = \bar{G}},\]
where $g_G :[0,H] \to [0,H]$ is an increasing, invertible function that is Lipschitz with parameter $L$ and $g_C(v) = v$ for all $v \in [0,H]$.
Intuitively, this means that if the bidder has type $(C, v)$, then their value for item $\bar{G}$ is $g_{\bar{G}}(v)$. A bidder with type $(G,v)$ has the value $0$ for any item $\bar{G} \neq G$ and value $v$ for the item $G$. 

\paragraph{Mechanisms.} A mechanism is defined by two functions $a$, $p$, where $a:\types \to [0,1]$ and $p: \types \to [0,H]$. The function $a$ is called the {\em allocation rule} and the function $p$ is called the {\em payment rule}. We will use $a_{\bar{G}}(v)$ to denote $a(\bar{G}, v)$ and  $p_{\bar{G}}(v)$ to denote $p(\bar{G}, v)$. A mechanism is said to be {\em incentive compatible} if, for all $\t = (\bar{G}, v), \t'=(\bar{G}', v') \in \types$
\begin{equation} \label{eq:ic}
\v(\t, \bar{G}) a(\t) - p(\t) \geq \v(\t, \bar{G}') a(\t') - p(\t') 
\end{equation}

\paragraph{Instance.} An instance $\I$ for the \setname setting is defined by a tuple $(q, f_{\bar{G}}, g_{\bar{G}})$. We will usually omit the subscript and simply write $(q,f,g)$. Our goal is find, for a given instance $\I$, the incentive compatible mechanism that maximizes the {\em revenue} $\E_{\t \sim (q,f)}[p(\t)]$. 

\subsubsection{A Linear Programming Formulation}
For an instance $\I$ for the \setname, finding the incentive compatible mechanism that maximizes $\E_{\t \sim (q,f)}[p(\t)]$ turns out to be equivalent to the linear program in \autoref{eq:primal}. Here and throughout, for a distribution $f$ supported on $[0,H]$, we use $\varphi(v) = v - \frac{1 - F(v)}{f(v)}$ to denote the Myerson's virtual value function, where $F(v) = \int_0^v f(t)dt$ is the cumulative distribution function for the distribution defined by $f$.

\begin{subequations}
\label{eq:primal}
\begin{alignat}{3}
 & \text{maximize} & \P(a) = \sum_{\bar{G}} \int_0^H  q_{\bar{G}}&f_{\bar{G}}(t)a_{\bar{G}}(t)\varphi_{\bar{G}}(t)dt& \\
 & \text{subject to} &&& \notag \\
 &\quad& \int_0^{g_G(v)} a_G(t)dt  - \int_0^v a_C(t)dt &\leq 0 \quad&,\forall v \label{eq:util} \\
 &\quad&   -a'_{\bar{G}}(v) &\leq 0  \quad&,\forall v \label{eq:allocmonotone}\\
 &\quad&  0 \leq a_{\bar{G}}(v) &\leq 1 \quad&,\forall v  \label{eq:allocrange}\\
  &\quad&   a_{\bar{G}}(0) &=0  \quad&
\end{alignat}
\end{subequations}

Even though the argument that \autoref{eq:primal} is equivalent to finding the incentive compatible mechanism that maximizes $\E_{\t \sim (q,f)}[p(\t)]$ is standard, we summarize it here. Observe that the non-trivial incentive compatibility constraints in \autoref{eq:ic} can be classified into two types:
\begin{itemize}
\item {\bf $\t = (\bar{G},v)$ and $\t' = (\bar{G}, v')$:} In this case, \autoref{eq:ic} reduces to $va_{\bar{G}}(v) - p_{\bar{G}}(v) \geq va_{\bar{G}}(v') - p_{\bar{G}}(v')$. This is attainable for all $v$, $v'$ if and only if the allocation rule $a_{\bar{G}}$ is monotone increasing (\autoref{eq:allocmonotone}) and $p_{\bar{G}}(v) = va_{\bar{G}}(v) - \int_0^v a_{\bar{G}}(t)dt$.
\item {\bf $\t = (C,v)$ and $\t' = (G, v')$:} In this case, \autoref{eq:ic} reduces to $va_{C}(v) - p_{C}(v) \geq g_G(v)a_{G}(v') - p_{G}(v')$. Due to the constraints in the previous paragraph, it is sufficient to have $va_{C}(v) - p_{C}(v) \geq g_G(v)a_{G}(g_G(v)) - p_{G}(g_G(v))$. This is equivalent to the constraint in \autoref{eq:util}.
\end{itemize}
Finally, it can be verified that if $p_{\bar{G}}(v) = va_{\bar{G}}(v) - \int_0^v a_{\bar{G}}(t)dt$, then  $\E_{\t \sim (q,f)}[p(\t)] = \P(a)$.

We will use $\p = (a_{\bar{G}})$ to denote a general solution to \autoref{eq:primal} and $\P(\I)$ to denote the optimal value for the instance $\I$.

\subsubsection{A Lagrangian Dual Formulation}

For any instance $\I$ of \setname, the revenue maximization problem is defined in \autoref{eq:primal}. Let $X \subset [0,H]$ be a discrete set of points in $[0,H]$.
Define the $X$-dual of \autoref{eq:primal} as:

\begin{subequations}
\label{eq:dual}
\begin{alignat}{3} 
 & \text{minimize} & \D_X(\lambda,\gamma,\Gamma) =  \sum_{\bar{G}} \int_0^H q_{\bar{G}}f_{\bar{G}}(t)&\max\left(0,\Phi_{\bar{G}}^{\lambda,\gamma,\Gamma, X}(t)\right)dt & \\
 & \text{subject to} &&&\notag\\
 &  \quad&  \lambda_{\bar{G}}(H) &=0 &\\
 &  \quad&  \lambda_{\bar{G}}(v),  \gamma_{G}(v) ,  \Gamma_{G}(x) &\geq 0 \quad&,\forall x \in X,v \in [0,H]
\end{alignat}
\end{subequations}

where 
\begin{equation*} 
\Phi_{\bar{G}}^{\lambda,\gamma,\Gamma, X}(v) = \varphi_{\bar{G}}(v) + \frac{1}{q_{\bar{G}}f_{\bar{G}}(v)}\left[-\lambda'_{\bar{G}}(v) - \sum_{X \ni x > g_{\bar{G}}^{-1}(v)} \Gamma_{\bar{G}}(x) - \int_{g_{\bar{G}}^{-1}(v)}^H \gamma_{\bar{G}}(s)ds\right]\\
\end{equation*}

\autoref{eq:dual} is obtained from \autoref{eq:primal} by Lagrangifying the constraints \autoref{eq:allocmonotone} using the variables $\lambda_{\bar{G}}(v)$, the constraints \autoref{eq:util} using the variables $\gamma_{G}(v)$, and the constraints \autoref{eq:util}, for $v \in X$ using the variables $\Gamma_{G}(v)$. Here and throughout this paper, we define $\gamma_C(v) = -\gamma_A(v) -\gamma_B(v)$ and $\Gamma_C(v) = -\Gamma_A(v) -\Gamma_B(v)$ for notational ease. We note that getting \autoref{eq:dual} from \autoref{eq:primal} requires integrating a certain term by parts. {\em Throughout this work, we assume that our functions are well-behaved enough to allow such standard operations.}

Note that the idea of using a Lagrangian dual is not new to this work. Indeed, a lot of recent advances in similar settings have been made using this technique. We will use $\d=  (\lambda_{\bar{G}}, \gamma_G, \Gamma_G)$  to denote a general solution to \autoref{eq:dual} and $\D_X(\I)$ to denote the optimal value for the instance $\I$. Often, we will abbreviate $\Phi_{\bar{G}}^{\lambda,\gamma,\Gamma, X}(v) $ to $\Phi_{\bar{G}}^{\d, X}(v)$ or even $\Phi(v)$  if the subscript and superscript are clear from the context.

We have the following `strong duality' result (proof in Appendix~\ref{app:SD}):

\begin{theorem}[Strong Duality] \label{thm:SD} Let $\I$ be an instance of \setname. 
\begin{enumerate}[label={(\alph*)}]
\item \label{thm:SD:WD} Let $X \subset [0,H]$ be a discrete set. For any feasible solution $\p = (a_{\bar{G}})$ of \autoref{eq:primal} and any feasible solution  $\d=  (\lambda_{\bar{G}}, \gamma_G, \Gamma_G)$ of the $X$-dual (\autoref{eq:dual}), it holds that:
\[\P(\p) \leq \D_X(\d).\]
Equality holds if and only if the following conditions are satisfied almost everywhere:
\begin{subequations} \label{eq:cs}
\begin{align}
\forall v:\Phi_{\bar{G}}^{\d}(v) > 0 &\implies a_{\bar{G}}(v) = 1.  \label{eq:pos}\\
\forall v:\Phi_{\bar{G}}^{\d}(v) < 0 &\implies a_{\bar{G}}(v) = 0. \label{eq:neg}\\
\forall v:\lambda_{\bar{G}}(v) > 0 &\implies a'_{\bar{G}}(v) = 0. \label{eq:iron}\\
\forall v:\gamma_{G}(v) > 0 &\implies \int_0^v a_C(t)dt = \int_0^{g_G(v)} a_G(t)dt. \label{eq:flowa}\\
\forall x \in X:\Gamma_{G}(x) > 0 &\implies \int_0^{x} a_C(t)dt = \int_0^{g_G(x)} a_G(t)dt. \label{eq:flowadisc}
\end{align}
\end{subequations}
\item \label{thm:SD:SD} There exists a set $X$ such that $\P(\I) = \D_X(\I)$. 
\end{enumerate}
\end{theorem}

The main reason we provide a proof for this `strong duality' result is that the variables are parametrized by a continuous variable. We could not find any results for such variables that subsume our setting.
Our proof of \autoref{thm:SD} works by showing, for all $\epsilon > 0$, a discrete linear program that has almost the same primal and dual value (up to terms depending on $\sqrt{\epsilon}$). Since strong duality holds for discrete systems, this gives us that the duality gap of our linear program is small. \autoref{thm:SD} then follows as the feasible region is closed.

\subsubsection{Our Framework}
Fix an instance $\I$ and discrete set $X \subseteq [0,H]$.
\paragraph{The Dual Framework.}  Let $\d = (\lambda_{\bar{G}}, \gamma_G, \Gamma_G)$ be a feasible solution for the $X$-dual, {\em i.e.}, \autoref{eq:dual}. Define
\[\underline{r}_{\bar{G}}(\d) = \inf \{v \mid \Phi^{\d,X}_{\bar{G}}(v) = 0\} \quad \quad\quad  \overline{r}_{\bar{G}}(\d) = \sup \{v \mid \Phi^{\d,X}_{\bar{G}}(v) = 0\}.\]
If the $\inf$ (resp. $\sup$) is over an empty set, we define it to be $H$ (resp. $0$).


We sometimes refer to $\gamma$ and $\Gamma$ as representing {\em flow}, {\em e.g.}, we say that there is  flow from $(C,v)$ to $A$ if $\gamma_A(v) >0$ or $\Gamma_A(v) > 0$.

An interval $[\underline{y}, \overline{y}] \subseteq [0,H]$ is said to be {\em ironed} on a item $\bar{G}$ if $\lambda_{\bar{G}}(\underline{y}) = \lambda_{\bar{G}}(\overline{y}) = 0$ and for all $v \in (\underline{y}, \overline{y})$, we have $\lambda_{\bar{G}}(v) > 0$.
 
\paragraph{The Primal Framework.} Let $\p = (a_{\bar{G}})$ be a feasible primal solution for \autoref{eq:primal}. Let $Y = [\underline{y}, \overline{y}] \subseteq [0,H]$ be an interval. Define $\MC_{\bar{G}}(Y, \p) = \lvert{\{\alpha \mid \exists v \in Y: a_{\bar{G}}(v) = \alpha\}}\rvert$ to be the number of distinct values taken by the function $a_{\bar{G}}$ over the interval $Y$.  Also, define $\MC(Y,\p) = \max_{\bar{G} \in \items} \MC_{\bar{G}}(Y,\p)$. We omit the argument $Y$ if it is $[0,H]$.

We define the {\em menu complexity} of the instance $\I$,  $\MC(\I) = \min_{\p : \P(\p) = \P(\I)} \MC(\p)$ to be the smallest menu complexity of any optimal solution to \autoref{eq:primal}.

\subsubsection{Formal Statements of our Results}

\begin{theorem} \label{thm:LB1} There exists an instance $\UC$ of \setname such that $\MC(\UC)$ is uncountably infinite. Furthermore, the instance $\UC$ satisfies $g_B(v) = v$ and all the distributions $f_{\bar{G}}$ are DMR\footnote{Recall that  a distribution is {\em DMR} if the Myerson's virtual value function is non-decreasing.}. 
\end{theorem}

\begin{theorem} \label{thm:UB} For any instance $\I$ such that the functions $g_G$ are piecewise linear, we have $\MC(\I)$ is at most countably infinite. 
\end{theorem}

\begin{theorem} \label{thm:LB2} There exists an instance $\C_1$ of \setname such that $\MC(\C_1)$ is countably infinite. Furthermore, the instance $\C_1$ satisfies $g_B(v) = v$ and all the distributions $f_{\bar{G}}$ are DMR, and the function $g_A$ is piecewise linear. 
\end{theorem}

\begin{theorem} \label{thm:LB-gen} There exists an instance $\C_2$ of \setname such that $\MC(\C_2)$ is countably infinite. Furthermore, the instance $\C_2$ satisfies $g_B(v) = v$ and the function $g_A$ is piecewise linear with only $2$ segments. 
\end{theorem}


\subsection{Master Theorem} \label{sec:RRSMT}

For our lower bounds, we will construct instances that have a large menu complexity. To show a lower bound on the menu complexity, we will define a feasible solution to the $X$-dual (\autoref{eq:dual}), for some $X$, show that  it is optimal, and then show that  any feasible primal that satisfies complementary slackness (\autoref{eq:cs}) with this dual must have a large menu complexity. Below, we prove a `Master Theorem' (\autoref{thm:MT}) (proof in Appendix~\ref{app:MT}) that constructs instances together with a feasible dual solution with some desirable properties.

\begin{theorem}[Master Theorem]\label{thm:MT} Let $H > 0$ be fixed. Suppose that, for all $\bar{G} \in \items$, points $\underline{\rho}_{\bar{G}} < \overline{\rho}_{\bar{G}} \in [1,H]$ are given. Let $X_A$, $X_B$ be discrete subsets of $[\underline{\rho}_C,\overline{\rho}_C]$. Consider finite or infinite sequences of disjoint intervals
\[\mathcal{Y}_G = \{(\underline{y}_{G,i}, \overline{y}_{G,i})\}_{i \geq 0} \quad\quad \mathcal{Z}_G = \{(\underline{z}_{G,i}, \overline{z}_{G,i})\}_{i \geq 0}\]
where $\underline{y}_{G,i},\overline{y}_{G,i}\in [\underline{\rho}_G,\overline{\rho}_G]$ and $\underline{z}_{G,i}, \overline{z}_{G,i} \in [\underline{\rho}_C,\overline{\rho}_C]$. Then, for any invertible functions $g'_G:[0,H]\to [0,H]$, there exists an instance $\I = (q, f_{\bar{G}}, g_G)$ and an $(X_A \cup X_B)$-dual $\d = (\lambda_{\bar{G}}, \gamma_{G}, \Gamma_{G})$ such that the following hold:
\begin{itemize}
\item $g_G(v) = g'_G(v)$ for all $v \in [0,H]$.
\item The value of $\Phi_{\bar{G}}^{\d, X_A \cup X_B}(v) f_{\bar{G}}(v)$ is non-decreasing and $\underline{r}_{\bar{G}}(\d) = \underline{\rho}_{\bar{G}} $ and $\overline{r}_{\bar{G}}(\d) = \overline{\rho}_{\bar{G}} $.
\item $\lambda_C(v) = 0$ throughout and $\lambda_G(v) > 0 \iff v \in (\underline{y}_{G,i}, \overline{y}_{G,i})$ for some $i$.
\item $\gamma_{G}(v) > 0 \iff v \in (\underline{z}_{G,i}, \overline{z}_{G,i})$ for some $i$.
\item $\Gamma_{G}(v) > 0 \iff v \in  X_G$ for some $i$.
\end{itemize}

Moreover, if $\mathcal{Y}_G$ is empty, then the distributions $f_{\bar{G}}$ are DMR.
\end{theorem}


\subsection{Lower Bounds} \label{sec:RRSLB}

\subsubsection{DMR Distributions} \label{sec:RRSLB1}
Let $\mathfrak{a}:[0,H] \to [0,1]$ be a non-decreasing function that is continuous except at countably many points. 
In this subsection, we construct an instance $\I  = (q,f,g)$ of \setname such any optimal solution of \autoref{eq:primal} for $\I$ satisfies that $a_A = \mathfrak{a}$\footnote{We abuse notation slightly here. What is meant is that $a_A$ takes the same values over the interval $[2, H+2]$ as $\mathfrak{a}$ takes over the interval $[0,H]$ (see \autoref{sec:RRSLB-DMR-inst} for the exact statement).}. In our construction, the distributions $f_{\bar{G}}$ are DMR and $g_B(v) = v$. 

Two important instantiations of this general procedure prove \autoref{thm:LB1} and \autoref{thm:LB2}. 
For \autoref{thm:LB1}, we construct $\UC$ by setting need $\mathfrak{a}(v) = v/H$. For \autoref{thm:LB2}, we construct $\C_1$ by setting $\mathfrak{a}$ to be a function that takes countably many values. A concrete example of such a function is one that has countably many ``steps'' as it moves from $0$ to $1$. We take care that the function $g_A$ is piecewise linear in $\C_1$.


\paragraph{The Instance}
\label{sec:RRSLB-DMR-inst}

For notational convenience, we work in the range $[0,H+3]$ in this subsubsection. Consider an increasing function $\mathfrak{a}:[0,H] \to [0,1]$ and let $\mathfrak{A}(v) = \int_0^v \mathfrak{a}(t)dt$. Let $\rho = H + 2 - \mathfrak{A}(H) \in [2,H+2]$. Note that $\mathfrak{A}^{-1}(v)$ is well defined for $v > 0$ and define: 

\[g'_A(v) = \begin{cases}
 v&, 0 \leq v \leq 1 \\
 \frac{v-1}{\rho-1} + 1&, 1 < v \leq \rho\\
\mathfrak{A}^{-1}\left(v - \rho\right) + 2&, \rho < v \leq H+2\\
 v&, H+2 < v \leq H+3\\
\end{cases}.\]

It is readily seen seen that $g'_A$ satisfies $\mathfrak{A}(H) - \mathfrak{A}\left(g'_A(v) - 2)\right) = H - (v - 2)$ in the interval $(\rho ,H+2]$. Let $g'_B(v) = v$ and apply \autoref{thm:MT} with $g'_G$ and
\begin{itemize}
\item $\underline{\rho}_A = 2$, $\overline{\rho}_A = H+2$, $\underline{\rho}_B = \overline{\rho}_B = \rho$,  $\underline{\rho}_C = 1$, $\overline{\rho}_C = H+2$.
\item $X_G, \mathcal{Y}_G$ are empty
\item $\mathcal{Z}_G(v) = \{(\rho, H+2)\}$.
\end{itemize}

This gives an instance $\I = (q,f,g)$ and an $(X_A \cup X_B)$-dual solution $\d = (\lambda_{\bar{G}}, \gamma_G, \Gamma_G)$ such that the distributions $f_{\bar{G}}$ are DMR and :
\begin{itemize}
\item For all $v$, $g_G(v) = g'_G(v)$ and $\underline{r}_{\bar{G}} = \underline{\rho}_{\bar{G}}$ and $\overline{r}_{\bar{G}} = \overline{\rho}_{\bar{G}}$.

\item $\lambda_{\bar{G}}(v), \Gamma_G(v)$ are $0$ throughout.

\item $\gamma_G(v) > 0 \iff v \in (\rho, H+2)$.
\end{itemize}

\paragraph{The Analysis}
Recall $\mathfrak{a}$, $\rho$, and $\I$ constructed above.
Define a feasible primal solution $\p^* = (a^*_{\bar{G}})$ of \autoref{eq:primal} for $\I$ as:
\[a^*_A(v) = \begin{cases}
0 &, 0\leq v \leq 2\\
\mathfrak{a}(v-2) &, 2< v \leq H+2\\
1 &, H+2 < v \leq H+3\\
\end{cases}\quad\quad
a^*_B(v) = a^*_C(v) = \begin{cases}
0 \hspace{1.3cm}&, 0\leq v \leq \rho\\
1 &, \rho < v \leq H+3\\
\end{cases}\]
In \autoref{lemma:compslack1} (proof in Appendix~\ref{app:lb1}), we show that $\p^*$ and $\d$ satisfy complementary slackness. To finish our menu complexity lower bound, we argue that {\em any} primal $\p = (a_{\bar{G}})$ that satisfies complementary slackness with $\d$ must have $a_A = a^*_A$. This proof is in \autoref{lemma:LB-DMR}. 

Together with strong duality (\autoref{thm:SD}), \autoref{lemma:compslack1} shows that $\d$ is optimal. Thus, any optimal primal for $\I$ must satisfy complementary slackness with $\d$. \autoref{lemma:LB-DMR} says that $a_A = a^*_A$ for this primal and thus, it has a high menu complexity.

\begin{lemma} \label{lemma:compslack1}The primal $\p^*$ is feasible and $\p^*$, $\d$ satisfy complementary slackness (\autoref{eq:cs}).
\end{lemma}
\begin{lemma} \label{lemma:LB-DMR} Consider any feasible primal $\p = (a_{\bar{G}})$ that satisfies complementary slackness with $\d$ must $a_A = a^*_A$, where $a^*_A$ is the allocation for item $A$ in $\p^*$.
\end{lemma}
\begin{proof}
We reason from \autoref{eq:cs}. Specifically, the constraints \eqref{eq:pos}, \eqref{eq:neg} for item $B$ imply that $a_B(v) = 0$ for $v \in [0,\rho)$ and $a_B(v) = 1$ for $v \in (\rho, H+3]$. Thus, we have:

\[\int_0^{g_B(v)}a_B(t)dt = (v - \rho)\mathbbm{1}_{v \geq \rho}.\]

Since $\gamma_{G}(v) > 0$ for all $v \in (\rho, H+2)$, we have by \eqref{eq:flowa} that $\int_0^{g_A(v)}a_A(t)dt = \int_0^{g_B(v)}a_B(t)dt = \int_0^v a_C(t)dt$ for all $v$ in this range. Thus, for all $v \in (\rho, H+2)$:

\[\int_0^{g_A(v)}a_A(t)dt = v -\rho.\]

Since the right hand side in the equation above is independent of the primal, we get for all $v \in (\rho, H+2)$ that $\int_0^{g_A(v)}a_A(t)dt = \int_0^{g_A(v)}a^*_A(t)dt$. Thus, for all $v \in (2, H+2)$, we have $\int_0^{v}a_A(t)dt = \int_0^{v}a^*_A(t)dt$ implying $a_A(v) = a^*_A(v)$ in this range. The constraints \eqref{eq:pos}, \eqref{eq:neg} for item $A$ fix the allocation $a_A$ outside $(2,H+2)$. Combining, we get that $a_A = a^*_A$.

\end{proof}

\subsubsection{Proof of \autoref{thm:LB-gen}}
\label{sec:RRSLB2}
In the last subsection, we showed that instances can have high menu complexity, even when all the distributions $f_{\bar{G}}$ are DMR. The reason for high menu complexity is the complexity in the functions $g_G$.
We now show that if the distributions $f_{\bar{G}}$ are not required to be DMR, even `simple' ({\em e.g.}, piecewise linear with only $2$ segments) functions $g_G$ can have countably infinite menu complexity. Two segments are required because of the arguments in \autoref{app:optvars}. This is tight due to our upper bounds in \autoref{sec:RRSUB}.

\paragraph{The Instance $\C_2$}
\label{sec:RRSLB-gen-inst}
For notational convenience, we work in the range $[0,H+1]$ in this subsubsection.
Define:
\[
g'_A(v)\footnote{Note that this function is piecewise linear with $3$ segments and not $2$. However, the first segment is just for notational ease and can be removed.} = \begin{cases}
v &, 0 \leq v \leq 1\\
\frac{v+1}{2} &, 1 < v \leq 2H/3 + 1\\
2V - H-1 &, 2H/3 +1< v \leq H +1\\
\end{cases}
\quad\quad\quad
g'_B(v) = v
\]

We define points $x_1 = \frac{4H}{5}+1$, $y_1 = \frac{3H}{4}+1$, and $x_i = \frac{8}{3}\frac{H}{2^i} +1$ , $y_i = \frac{8}{5} \frac{H}{2^i}+1 $ for $i > 1$. 
Note that the sequence $x_i$ converges to $x = 1$. 

In order to construct $\C_2$, we  apply \autoref{thm:MT} with $g'_G$ and
\begin{itemize}
\item $\underline{\rho}_A = g'_A(x)$, $\overline{\rho}_A = g'_A(x_2)$, $\underline{\rho}_B = g'_B(x)$, $\overline{\rho}_B= g'_B(x_1)$,  $\underline{\rho}_C = x$, $\overline{\rho}_C = x_1$.

\item $\mathcal{Y}_A = \{(g'_A(x_{2i+2}), g'_A(x_{2i}))\}$ for some $i>0$. $\mathcal{Y}_B = \{(g'_B(x_{2i+1}), g'_B(x_{2i-1}))\}$ for some $i>0$. 

\item $\mathcal{Z}_G$ is empty and $X_G = \{y_i\}_{i > 0}$.
\end{itemize}
This gives an instance $\C_2 = (q,f,g)$ and an $(X_A \cup X_B)$-dual solution $\d = (\lambda_{\bar{G}}, \gamma_G, \Gamma_G)$ such that :
\begin{itemize}
\item For all $v$, $g_G(v) = g'_G(v)$ and $\underline{r}_{\bar{G}} = \underline{\rho}_{\bar{G}}$ and $\overline{r}_{\bar{G}} = \overline{\rho}_{\bar{G}}$.

\item $\lambda_{C}(v), \gamma_G(v)$ are $0$ throughout. $\lambda_{A}(v) >0$ if and only if $v \in (g_A(x_{2i+2}), g_A(x_{2i}))$ for some $i > 0$ and $\lambda_{B}(v) >0$ if and only if $v \in (g_B(x_{2i+1}), g_B(x_{2i-1}))$ for some $i > 0$.

\item $\forall i > 0: \Gamma_G(y_i) > 0$.
\end{itemize}

\paragraph{The Analysis}

Recall the definitions of $x_i, y_i$ and the instance $\C_2$ above. Define a feasible primal solution $\p = (a^*_{\bar{G}})$ of \autoref{eq:primal} for $\C_2$ as:

\[a^*_A(v) = \begin{cases}
0 &, 0\leq v \leq g_A(x)\\
\frac{40}{13\cdot 4^{i}} &, g_A(x_{2i+2}) \leq v < g_A(x_{2i})\\
1 &,  g_A(x_2)\leq v \leq H+1\\
\end{cases}\quad\quad\quad a^*_B(v) =  \begin{cases}
0 &, 0\leq v \leq g_B(x)\\
\frac{40}{13\cdot 4^{i}} &, g_B(x_{2i+1}) \leq v < g_B(x_{2i-1}) \\
1 &,  g_B(x_1)\leq v \leq H+1\\
\end{cases}\]
\[a^*_C(v) =  \begin{cases}
0 &, 0\leq v \leq x \\
\frac{20}{13\cdot 2^i} &, y_{i+1} \leq v < y_{i} \\
1 &,  y_1 \leq v \leq H+1\\
\end{cases}\]

We proceed exactly as in \autoref{sec:RRSLB1}.
In \autoref{lemma:compslack2} (proof in Appendix~\ref{app:lb2}), we show that $\p$ and $\d$ satisfy complementary slackness. To finish our menu complexity lower bound, we argue that {\em any} primal that satisfies complementary slackness with $\d$ must have infinite menu complexity. 

Together with strong duality (\autoref{thm:SD}), \autoref{lemma:compslack2} shows that $\d$ is optimal. Thus, any optimal primal for $\C_2$ must satisfy complementary slackness with $\d$. \autoref{lemma:LB-gen} says that such a primal has infinite menu complexity

%

\begin{lemma} \label{lemma:compslack2} There primal $\p$ is feasible for $\C_2$ and $\p$,$\d$ satisfy complementary slackness.
\end{lemma}

\begin{lemma} \label{lemma:LB-gen} Any feasible primal $\p$ for $\C_2$ with a finite menu complexity does not satisfy complementary slackness with $\d$.
\end{lemma}
\begin{proof} Proof by contradiction. Let $\p = (a_{\bar{G}})$ be a feasible primal with finite menu complexity that satisfies complementary slackness with $\d$.

Let $i^*$ be the largest $i$ such that $a_A(g_A(y_{i})) > a_A(g_A(y_{i + 1}))$ or $a_B(g_B(y_{i})) > a_B(g_B(y_{i + 1}))$. If no such $i^*$ exists, we define $i^* = 0$. Since $\MC(\p)$ is assumed to be finite and the allocation is monotone (by constraint \eqref{eq:allocmonotone}), $i^*$ is well defined and $a_A(g_A(y_{i}))$ and $a_B(g_B(y_{i}))$ is constant for all $i > i^*$. Let the constant values be $\pi_A$ and $\pi_B$ respectively.

Since $\p$, $\d$ satisfy \autoref{eq:cs}, we have, by \autoref{eq:iron}
\begin{equation} \label{eq:constG}
  a_G(v) = \pi_G\quad\quad,\forall v \in (g_G(x), g_G(y_{i^* + 1})]
\end{equation}

Also, for all $i >0$ we have by the constraint \eqref{eq:flowadisc} (applied once to both $y_i$ and $y_{i+1}$) 
\begin{equation}\label{eq:steputil}
\int_{g_A(y_{i+1})}^{g_A(y_{i})}a_A(t)dt = \int_{y_{i+1}}^{y_{i}}a_C(t)dt = \int_{g_B(y_{i+1})}^{g_B(y_{i})}a_B(t)dt
\end{equation}

We derive a contradiction in two steps. First, we prove 
\begin{claim} \label{claim:ratio} $ \pi_A = 2\pi_{B}$.
\end{claim}
\begin{proof}
 Consider an $i$ larger than $i^*+10$. Using \autoref{eq:constG} and \autoref{eq:steputil}, we get
\begin{align*}
 \pi_A\left(g_A(y_{i}) - g_A(y_{i+1})\right) &= \pi_B\left(g_B(y_{i}) - g_B(y_{i+1})\right)&\\
 \pi_A\left(y_{i} - y_{i+1}\right) &= 2\pi_B\left(y_{i} - y_{i+1}\right)& (\text{Definition of $g_A$})\\
 \pi_A &= 2\pi_B.
\end{align*}
\end{proof}

\begin{claim} $a_A(g_A(y_2)) \neq 2a_B(g_B(y_2))$.
\end{claim}
\begin{proof} Proof by contradiction. Suppose that $a_A(g_A(y_2)) = 2a_B(g_B(y_2)) = 2\pi$. By \autoref{eq:pos}, we have $a_A(g_A(v)) = 1$ for all $v > x_2$. Since  $y_2, y_1$ are in the same ironed interval on $B$ and $y_2, x_2$ are in the same ironed interval on $A$, we have using \autoref{eq:steputil} that
\[g_A(y_{1}) - g_A(x_{2})  + 2\pi\left(g_A(x_{2}) - g_A(y_{2})\right) = \pi\left(g_B(y_{1}) - g_B(y_{2})\right).\]
Plugging in the values of $y_1, x_2, y_2$, we have 
\[\frac{H}{2}  - \frac{H}{3} + 2\pi\left(\frac{H}{3} - \frac{H}{5}\right) = \pi\left(\frac{3H}{4} - \frac{2H}{5}\right),\]
a contradiction to $\pi < 1$.
\end{proof}

These two claims together with \autoref{eq:constG} establish that $i^* > 1$. We now give a contradiction assuming $i^*$ is even. A similar argument works if $i^*$ is odd. Since $g_A(y_{i^*}), g_A(y_{i^* + 1})$ lie in the same ironed interval in $A$, we have $a_A(g_A(y_{i^*})) = a_A(g_A(y_{i^* + 1})) = \pi_A$.  By choice of $i^*$, we have $\pi'_B = a_B(g_B(y_{i^*})) > \pi_B$. By \autoref{eq:steputil}, we have
 \[\frac{\pi_A}{2}(y_{i^*} - y_{i^*+1}) =  \pi_A(g_A(y_{i^*}) - g_A(y_{i^*+1})) = \pi'_B(y_{i^*} - x_{i^*+1}) + \pi_B(x_{i^*+1} - y_{i^*+1}) > \pi_B(y_{i^*} - y_{i^*+1}),\]
 a contradiction to \autoref{claim:ratio}.
 \end{proof}


 \subsection{Upper bounds} \label{sec:RRSUB}

In this subsection, we prove that for any instance $\I$ such that the functions $g_A(\cdot)$ and $g_B(\cdot)$ are piecewise linear, we have $\MC(\I)$ is at most countably infinite (\autoref{thm:UB}). This result is tight by our arguments in \autoref{sec:RRSLB}

Our line of argument is as follows: Fix an instance $\I$. By \autoref{thm:SD}, there exists an $X$, a primal solution $\p = (a_{\bar{G}})$ and an $X$-dual solution $\d = (\lambda_{\bar{G}}, \gamma_G, \Gamma_G)$ that satisfy complementary slackness.
From $\p, \d$, we construct another primal solution $\hat{\p}$ such that $\MC(\hat{\p})$ is small and $\hat{\p}$, $\d$ satisfy complementary slackness.
Thus, the primal $\hat{\p}$ also defines an optimal revenue mechanism. The menu complexity of this new mechanism gives us our upper bound on $\MC(\I)$.

We note that this technique is markedly different from that employed in \autoref{app:optvars} where we assume an optimal dual and describe a recovery algorithm that reads an optimal primal from the optimal dual. Here, we assume {\em both}, an optimal dual and an optimal primal\footnote{Note that we need \autoref{thm:SD} to assume that there exists an optimal primal-dual pair that satisfies complementary slackness.}, and prove that such a primal can be improved to have a lower menu complexity. 

\subsubsection{Splitting}
 
Our procedure to improve the primal makes extensive use of the following ``splitting'' operation:
 
\begin{definition} \label{def:split}  Let $a:\mathbb{R} \to \mathbb{R}$ be a function and consider the interval $s = [x,y]$, where $x,y \in \mathbb{R}$. Define the function $a \circ s$ as  :
\[a\circ s (v) = \begin{cases}
a(v) &, v \not \in s\\
\frac{1}{y-x} \int_x^y a(t)dt &, v \in s\\
\end{cases}\]
\end{definition}
The following is easily observed for all $v \notin s$:
\begin{equation}\label{eq:split-util}
\int_{-\infty}^v a\circ s(t)dt = \int_{-\infty}^v a(t)dt
\end{equation}

If $s_1$ and $s_2$ are two disjoint intervals, then $a \circ s_1 \circ s_2 = a \circ s_2 \circ s_1$. We will use $a \circ s_1 s_2$ to denote this common value.

\begin{remark} \label{rem:split-util} Let $a$ be a function and $s$ be an interval. For any $x,y,z \in s$, we have 
\[\int_{-\infty}^z a\circ s(t)dt = \frac{1}{y-x}\left[(y-z)\int_{-\infty}^x a(t)dt  + (z-x)\int_{-\infty}^y a(t)dt\right] .\]
\end{remark}

\begin{lemma} \label{lemma:split-util} Let $a$ be an increasing function and $s$ be an interval. We have $\int_{-\infty}^z a(t)dt \leq \int_{-\infty}^z a\circ s(t)dt$. Moreover, equality holds if $v \not \in s$ or $a$ is constant over the interval $s$.
\end{lemma}
\begin{proof} Let $s = [x,y]$. Since $a$ is increasing, we have
\begin{align*}
\int_{x}^y a(t)dt &\geq  \int_{x}^z a(t)dt  + \frac{y-z}{z-x}\int_{x}^z a(t)dt=  \frac{y-x}{z-x}\int_{x}^z a(t)dt.
\end{align*}
Rearranging gives the result. The moreover part can be using \autoref{eq:split-util} and \autoref{def:split}.
\end{proof}

\subsubsection{Proof of \autoref{thm:UB}}
Fix an instance $\I$ and let $\p, \d$, be an optimal primal dual pair ($\d$ is $X$-dual for some $X$). Without loss of generality, we can assume that the product $\Phi^{\d,X}_{\bar{G}}(v) f_{\bar{G}}(v)$ is non-decreasing (see \autoref{sec:addtlprelims}). We define the function $\sgn(x)$ to be $1$ if $x> 0$, $0$ if $x = 0$, and $-1$ if $x < 0$.
Using $\d$, define the following notion of  a strip.
\begin{definition}[Strip] A strip $s$ is an interval $[x,y]$, where $x,y \in [0,H]$, such that the following hold:
\begin{itemize}
\item For all $\bar{G}$, there functions $g_{\bar{G}}$ are linear over $s$
\item For all $\bar{G}$, the function $\sgn\left(\Phi_{\bar{G}}^{\d, X}(g_{\bar{G}}(v))f_{\bar{G}}(g_{\bar{G}}(v))\right)$ is constant for $v \in s$. 
\item For all $\bar{G}$, the function $\sgn\left(\lambda_{\bar{G}}(g_{\bar{G}}(v))\right)$ is constant for $v \in s$.
\end{itemize}
\end{definition}

The following holds for any strip.

\begin{lemma}\label{lemma:strip} For $\I, \p, \d$ as above, let $s$ be any strip. There exists another primal solution $\hat{\p}$ of $\I$ such that 
\begin{itemize}
\item $\hat{a}_{\bar{G}}(g_{\bar{G}}(v)) = a_{\bar{G}}(g_{\bar{G}}(v))$ for all $v \not\in s$.
\item $\hat{\p}$, $\d$ satisfy complementary slackness (\autoref{eq:cs}).
\item For all $\bar{G}$, it holds that $\MC_{\bar{G}}\left([g_{\bar{G}}(x), g_{\bar{G}}(y)], \hat{a}_{\bar{G}}\right) \leq 10$.
\end{itemize}
\end{lemma}
The proof of this lemma spans the rest of this subsection.
\begin{proof}
Let $s = [x,y]$. Define the function $u_{\bar{G}}(v) = \int_0^{g_{\bar{G}}(v)} a_{\bar{G}}(t)dt$\footnote{Note that this is just the utility of a bidder with type $(\bar{G}, g_{\bar{G}}(v))$.}. Since the functions $g_{\bar{G}}$ are continuous, and $a_{\bar{G}}$ is continuous except at countably many points, we have that $u_{\bar{G}}$ is also continuous. Define points $\underline{z}$, $\overline{z}$ as follows\footnote{Throughout this subsection, we define several infimums and supremums. In case the argument to any of these is empty, we simply drop those terms from where they are used. For example, if $\underline{z}$, $\overline{z}$ are not defined, we simply use $a_G \circ [g_G(x), g_G(y)]$ instead of $a_G \circ [g_G(x), g_G(\underline{z})][g_G(\underline{z}), g_G(\overline{z})][g_G(\overline{z}), g_G(y)]$ below.}:
\[ \underline{z} = \inf_{v \in [x,y]} \{v \mid u_A(v) = u_B(v)\} \quad\quad \overline{z} = \sup_{v \in [x,y]} \{v \mid u_A(v) = u_B(v)\}. \]
We use $\underline{z}$ and $\overline{z}$ to define:
\begin{align*}
\hat{a}_G &= a_G \circ [g_G(x), g_G(\underline{z})][g_G(\underline{z}), g_G(\overline{z})][g_G(\overline{z}), g_G(y)].\\
\hat{u}_G(v) &= \int_0^{g_G(v)} \hat{a}_G(t)dt.
\end{align*}
Now, define:
\begin{align*}
z_1 = \inf_{v \in [x, \underline{z}]} \{v \mid \max(\hat{u}_A(v), \hat{u}_B(v)) = u_C(v)\} \quad&\quad z_2 = \sup_{v \in [x, \underline{z}]} \{v \mid \max(\hat{u}_A(v), \hat{u}_B(v)) = u_C(v)\}.\\
z_3 = \inf_{v \in [\underline{z}, \overline{z}]} \{v \mid \max(\hat{u}_A(v), \hat{u}_B(v)) = u_C(v)\} \quad&\quad z_4 = \sup_{v \in [\underline{z}, \overline{z}]} \{v \mid \max(\hat{u}_A(v), \hat{u}_B(v)) = u_C(v)\}.\\
z_5 = \inf_{v \in [\overline{z}, y]} \{v \mid \max(\hat{u}_A(v), \hat{u}_B(v)) = u_C(v)\} \quad&\quad z_6 = \sup_{v \in [\overline{z}, y]} \{v \mid \max(\hat{u}_A(v), \hat{u}_B(v)) = u_C(v)\}.
\end{align*}
Finally, we define:
\begin{align*}
\hat{a}_C&= a_C \circ [x, z_1][z_1, z_2][z_2, \underline{z}][\underline{z}, z_3][z_3, z_4][z_4, \overline{z}][\overline{z}, z_5][z_5, z_6][z_6, y].\\
\hat{u}_C(v) &= \int_0^{v} \hat{a}_C(t)dt.
\end{align*}
Our primal $\hat{\p}$ is defined by the allocations $(\hat{a}_{\bar{G}})$.
Note that item $1$ and item $3$ are straightforward from \autoref{def:split}. We only concentrate on item $2$.

For item $2$, we verify each of the constraints in \autoref{eq:cs}. For \autoref{eq:pos}, observe that if $\Phi_{\bar{G}}(g_{\bar{G}}(v)) > 0$ for some $v \in s$ and some $\bar{G}$, then, since $s$ is a strip, $\Phi_{\bar{G}}(g_{\bar{G}}(v)) > 0$ throughout $s$.  Thus, the allocation $a_{\bar{G}}$ is $1$ throughout $s$ and our operations have no effect. If $v \notin s$, then the result follows as $\p$, $\d$ satisfied complementary slackness. A similar argument verifies \autoref{eq:neg} and \autoref{eq:iron}.

For \autoref{eq:flowa}, consider a $v, G$ such that $\gamma_G(v) > 0$. If $v \not \in s$, then the result follows because of \autoref{rem:split-util} and the fact that $\p,\d$ satisfy complementary slackness. If $v \in s$, then, since $\p,\d$ satisfied complementary slackness, we have $u_G(v) = u_C(v)$. Thus, by \autoref{lemma:split-util}, $\hat{u}_G \geq  u_C(v)$ implying that $v \in [z_1, z_2] \cup [z_3, z_4] \cup [z_5, z_6]$. Suppose that $v \in [z_1, z_2]$. The other cases are similar. We have:
\begin{align*}
\hat{u}_G(v) &= \frac{\hat{u}_G(z_1) (g_G(z_2)-g_G(v)) + \hat{u}_G(z_2)(g_G(v)-g_G(z_1))}{g_G(z_2)-g_G(z_1)} &\text{(\autoref{rem:split-util})}\\
&= \frac{\hat{u}_G(z_1) (z_2-v) + \hat{u}_G(z_2)(v-z_1)}{z_2 -z_1} &\text{($g_G(v) = m_G v + c_G$ over a strip)}\\ 
&= \frac{\hat{u}_C(z_1) (z_2-v) + \hat{u}_C(z_2)(v-z_1)}{z_2 -z_1} &\text{(Definition of $z_1$, $z_2$)}\\ 
&= \hat{u}_C(v) &\text{(\autoref{rem:split-util})}
\end{align*}

\autoref{eq:flowadisc} is verified similarly.


\end{proof}
\begin{proof}[Proof of \autoref{thm:UB}]
Since the functions $g_G$ are assumed to be piecewise linear, we can partition the range $[0,H]$ into countably many disjoint strips. \autoref{thm:UB} follows by applying \autoref{lemma:strip} to each of the strips.
\end{proof}

\section{Missing Proofs From Appendix~\ref{sec:raghuvansh}} \label{sec:rrsappendix}

 \subsection{The Proof of \autoref{thm:SD}} \label{app:SD}
 \subsubsection{The Proof of \autoref{thm:SD:WD}}
\begin{proof} Fix $X \subset [0,H]$. Observe that :
\[\D_X(\d) =  \sum_{\bar{G}} \int_0^H q_{\bar{G}}f_{\bar{G}}(t)\max\left(0,\Phi_{\bar{G}}(t)\right)dt\geq \sum_{\bar{G}} \int_0^H q_{\bar{G}}f_{\bar{G}}(t)\Phi_{\bar{G}}(t)a_{\bar{G}}(t)dt\]
Using the definition of $\Phi$, we get 
\[
\D_X(\d) \geq \P(\p) +  \sum_{\bar{G}} \int_0^H a_{\bar{G}}(t)\left(-\lambda'_{\bar{G}}(t) - \sum_{X \ni x > g_{\bar{G}}^{-1}(t)} \Gamma_{\bar{G}}(x) - \int_{g_{\bar{G}}^{-1}(t)}^H \gamma_{\bar{G}}(s)ds\right)dt
\]
First, note that, integrating by parts, we have $ \sum_{\bar{G}} \int_0^H a_{\bar{G}}(t)\lambda'_{\bar{G}}(t)dt = -\sum_{\bar{G}} \int_0^H a'_{\bar{G}}(t)\lambda_{\bar{G}}(t)dt \leq 0$ by \autoref{eq:allocmonotone}.
Fix any $v \in[0,H]$. Grouping all the terms with $\gamma_G(v)$, we observe that  $\gamma_G(v)$ is multiplied by $\int_0^v a_C(t)dt - \int_0^{g_G(v)} a_G(t) dt \geq 0$ by \autoref{eq:util}.
Similarly, fix any $x \in X$. Grouping all the terms with $\Gamma_G(x)$, we observe that  $\Gamma_G(x)$ is multiplied by $\int_0^x a_C(t)dt - \int_0^{g_G(x)} a_G(t) dt \geq 0$ by \autoref{eq:util}.

Thus, we have $\D_X(\d) \geq \P(\p)$. To finish the proof, observe that the conditions in \autoref{eq:cs} are exactly those needed to make these inequalities tight.
\end{proof}

  \subsubsection{The Proof of \autoref{thm:SD:SD}}
  
 \begin{proof} We prove that for all $\epsilon >0$, there is a set $X_{\epsilon}$ such that $\P(\I) \geq \D_{X_{\epsilon}}(\I) - \epsilon$. The statement then follows because the the union of the region in \autoref{eq:dual} for all possible sets $X$ is closed. Recall that the functions $g_{G}$ are $L$-Lipschitz.
 
 Fix $\epsilon > 0$ and let $\delta >0$ be sufficiently small.
 Our proof proceeds by defining a {\em discrete} linear program (\autoref{eq:primaldisc}) and its dual (\autoref{eq:dualdisc}) for $\I, \delta$. Let the optimal value of \autoref{eq:primaldisc} be $\P_{\delta}(\I)$ and the optimal value of \autoref{eq:dualdisc} be $\D_{\delta}(\I)$. Since strong duality holds for discrete linear programs, we have $\P_{\delta}(\I) = \D_{\delta}(\I)$. We also ensure that $\D_{\delta}(\I)  = \P_{\delta}(\I) \leq H$ for all $\delta$.
 
 In \autoref{thm:primalapprox}, we show that $\P(\I) \geq (1-\sqrt{\delta})\P_{\delta}(\I) - L\sqrt{\delta}$. In \autoref{thm:dualapprox}, we show that there is a set $X_{\delta}$ such that $ \D_{X_{\delta}}(\I)  \leq \D_{\delta}(\I)$. Combining, we get 
 \[\P(\I) \geq (1-\sqrt{\delta})\P_{\delta}(\I) - L\sqrt{\delta} = (1-\sqrt{\delta})\D_{\delta}(\I) - L\sqrt{\delta} \geq (1-\sqrt{\delta})\D_{X_{\delta}}(\I) - L\sqrt{\delta} >   \D_{X_{\delta}}(\I) - \epsilon, \]
 for small enough $\delta$.
 \end{proof}
 
 The rest of this subsection is devoted to defining and analyzing the discrete linear program, in order to prove \autoref{thm:primalapprox} and \autoref{thm:dualapprox}.
 
 \subsubsection{The Discrete Linear Program}
 
We describe a discrete linear program for the instance $\I$. The instance $\I$ is fixed for the rest of this subsection. Without loss of generality, let $\delta > 0$ be such that $H/\delta = H'$ is an integer. 
 \paragraph{The Primal} 
 Consider the following optimization problem with the variables $a_{\bar{G}}(i)$ and $p_{\bar{G}}(i)$ for $0 \leq i \leq H'$. In this subsection,  we abuse notation and write $g_{\bar{G}}(i)$ instead of $g_{\bar{G}}(i\delta)$. Define $\int_{g_{\bar{G}}(i-1)}^{g_{\bar{G}}(i)} f_{\bar{G}}(v)dv = \hat{f}_{\bar{G}}(i)$ and $\hat{f}_{\bar{G}}(0) = 0$.
\begin{subequations}
\label{eq:primaldisc}
\begin{alignat}{3}
 & \text{maximize} & \P_{\delta}(a) = \sum_{\bar{G}} \sum_{i = 1}^{H'}  q_{\bar{G}}\hat{f}_{\bar{G}}(i)p_{\bar{G}}(i)&& \\
 & \text{subject to} &&&\notag\\
 && g_C(i) a_C(i) - p_C(i)  &\geq g_G(i)a_G(i) -p_G(i) \quad&,\forall i \label{eq:c1}\\
 && g_{\bar{G}}(i) a_{\bar{G}}(i) - p_{\bar{G}}(i)  &\geq g_{\bar{G}}(i)a_{\bar{G}}(i+1) -p_{\bar{G}}(i+1) \quad&,\forall i\label{eq:c2}\\
 && g_{\bar{G}}(i) a_{\bar{G}}(i) - p_{\bar{G}}(i)  &\geq g_{\bar{G}}(i)a_{\bar{G}}(i-1) -p_{\bar{G}}(i-1) \quad&,\forall i \label{eq:c3}\\
   &\quad&  p_{\bar{G}}(0) = a_{\bar{G}}(0) &= 0 \quad &  \\
 &\quad&  a_{\bar{G}}(i) &\in [0,1] \quad&,\label{eq:c5}\forall i  
\end{alignat}
\end{subequations}

It is easy to see why  $\P_{\delta}(\I) \leq H$. We also have:

\begin{theorem}\label{thm:primalapprox} The optimal value $\P_{\delta}(\I) $ of \autoref{eq:primaldisc} satisfies $(1-\sqrt{\delta}) \P_{\delta}(\I)\leq \P(\I) + L\sqrt{\delta}$.
\end{theorem}
\begin{proof}
Let $\eta = \sqrt{\delta}$ and $(a_{\bar{G}}, p_{\bar{G}})$ be the optimal solution for \autoref{eq:primaldisc}.  Consider the set $\mathcal{T} = \{(G_j, a_j, p_j) \mid \exists i, G_j : a_{G_j}(i) = a_j, (1-\eta) p_{G_j}(i) = p_j\}$. Using the set $\mathcal{T}$, we now define a mechanism $M$ for the continuous revenue optimization problem. Consider a type $\t  = (\bar{G}, v) \in \types$. Define: 
\[(G_{\t}, a_{\t}, p_{\t}) = \argmax_{(G_j, a_j, p_j) \in \mathcal{T} } \v(\t, G_j) a_j - p_j.\]
Let $M$ be the mechanism that allocates item $G_{\t}$ with probability $a_{\t}$ and charges price $ p_{\t}$ to a bidder who reports $\t$\footnote{We slightly deviate from our definition of a mechanism and allow $G_{\t}$ to be different from $\bar{G}$.}. Observe that $M$ is a truthful mechanism. By a standard argument, there exists a feasible solution $\p$ to \autoref{eq:primal} such that $\P(\I) \geq \P(\p) = \E_{\t \sim (q, f)}[p_{\t}]$ is the expected revenue of $M$. We now prove that $\E_{\t \sim (q, f)}[p_{\t}] \geq (1-\eta)\P_{\delta}(\I) -L\eta $.

Couple a bidder in $M$ with (continuous) type $\t = (\bar{G},v)$ with a bidder with (discrete) type $(\bar{G},i)$ where $i$ is the smallest value such that $g_{\bar{G}}(i)  > v$ . The coupling is valid as $\int_{g_{\bar{G}}(i-1)}^{g_{\bar{G}}(i)} f_{\bar{G}}(v)dv = \hat{f}_{\bar{G}}(i)$. We show that $p_{\t} \geq (1-\eta)p_{\bar{G}}(i) - L\eta$. Taking the expectation on both sides gives the result.

Observe that:
\begin{align*}
\v(\t, G_{\t}) a_{\t} - p_{\t} &\geq v a_{\bar{G}}(i) - (1-\eta)p_{\bar{G}}(i)\\
g_{\bar{G}}(i) a_{\bar{G}}(i) - p_{\bar{G}}(i) &\geq \v(\bar{G}, g_{\bar{G}}(i), G_{\t}) a_{\t} - \frac{p_{\t}}{1-\eta}
\end{align*}

 Adding, we get that 
 
 \[\v(\bar{G}, v, G_{\t}) a_{\t}+ g_{\bar{G}}(i) a_{\bar{G}}(i)  + \frac{\eta p_{\t}}{1-\eta} \geq v a_{\bar{G}}(i) + \eta p_{\bar{G}}(i) + \v(\bar{G},g_{\bar{G}}(i), G_{\t}) a_{\t} \]
 Since $a_{\t}, a_{\bar{G}}(i) \in [0,1]$, we have 
 \begin{align*}
 p_{\t} &\geq (1-\eta)p_{\bar{G}}(i) + \frac{1-\eta}{\eta}\left((v - g_{\bar{G}}(i))a_{\bar{G}}(i) + \left(\v(\bar{G}, g_{\bar{G}}(i), G_{\t}) - \v(\bar{G}, v, G_{\t})\right)a_{\t}\right)\\
&\geq (1-\eta)p_{\bar{G}}(i) + \frac{1}{\eta}\left(g_{\bar{G}}(i-1) - g_{\bar{G}}(i)\right)a_{\bar{G}}(i)  \tag{$g_{\bar{G}}(i) > v \geq g_{\bar{G}}(i-1)$}\\
&\geq (1-\eta)p_{\bar{G}}(i) - L\eta \tag{$g_{\bar{G}}$ is $L$-Lipschitz} \
 \end{align*}

\end{proof}

 \paragraph{The Dual}
 Consider the following Lagrangian relaxation of \autoref{eq:primaldisc}, where we Lagrangify the constraints \autoref{eq:c1} using the variables $\Gamma_{G}(i)$,  the constraints \autoref{eq:c2} using the variables $\lambda^+_G(i)$, and the constraints \autoref{eq:c3} using the variables $\lambda^-_G(i)$. We use  the convention that $\Gamma_C(i) = -\Gamma_A(i) - \Gamma_B(i)$
 \begin{subequations}
\label{eq:dualdisc}
\begin{alignat}{4} 
&  \text{minimize} & \D_{\delta}(\I)(\lambda, \Gamma) &= \sum_{\bar{G}} \sum_{i=1}^{H'} q_{\bar{G}}\hat{f}_{\bar{G}}(i)\max\left(0,\Phi_{\bar{G}}^{\lambda,\Gamma}(i)\right) &\\
&  \text{subject to} &&&\notag\\
&   & \lambda^+_{\bar{G}}(i) + \lambda^-_{\bar{G}}(i)  - \Gamma_{\bar{G}}(i) &= q_{\bar{G}}\hat{f}_{\bar{G}}(i) + \lambda^+_{\bar{G}}(i-1) + \lambda^-_{\bar{G}}(i+1) \quad&,\forall i \label{eq:flowg}\\
  & &  \lambda^+_{\bar{G}}(i), \lambda^-_{\bar{G}}(i) &\geq 0 \quad&,\forall i\\
  &&  \gamma_{G}(i) &\geq 0 \quad&,\forall i
 \end{alignat}
 \end{subequations}
where terms like $\lambda^+{\bar{G}}(H+1)$ are defined to be $0$ and 
\[\Phi_{\bar{G}}^{\lambda,\gamma}(i) = \left(g_{\bar{G}}(i) - \frac{1}{q_{\bar{G}}\hat{f}_{\bar{G}}(i)}\left(\lambda^+_{\bar{G}}(i-1)\left(g_{\bar{G}}(i-1) - g_{\bar{G}}(i) \right)+ \lambda^-_{\bar{G}}(i+1)\left(g_{\bar{G}}(i+1) - g_{\bar{G}}(i)\right)\right) \right)\]
 
 \begin{theorem}\label{thm:dualapprox} The optimal value $\D_{\delta}(\I) $ of \autoref{eq:dualdisc} satisfies $\D_{\delta}(\I) \geq \D_{X}(\I)$ where $X = \{i \delta \mid i \in \mathbb{Z} \cap [0,H']\}$.
\end{theorem}
\begin{proof} We proceed by defining a feasible solution of \autoref{eq:dual} from the optimal solution of \autoref{eq:dualdisc} with the same value. Throughout this proof, we denote the variables of the discrete linear program using a `hat'. Let $( \hat{\lambda}^+_{\bar{G}}, \hat{\lambda}^-_{\bar{G}}, \hat{\Gamma}_{G})$ be the optimal solution to \autoref{eq:dualdisc}. 
Define $\hat{F}_{\bar{G}}(v) = \sum_{i : g_{\bar{G}}(i) < v} \hat{f}_{\bar{G}}(i)$. Define an $X$-dual solution $\d = (\lambda_{\bar{G}}, \gamma_G, \Gamma_G)$ for $\I$ as follows

\begin{align*}
\lambda_{\bar{G}}(v) &= vq_{\bar{G}}(F_{\bar{G}}(v) - \hat{F}_{\bar{G}}(v) ) +(v-g_{\bar{G}}(i))\lambda^+_{\bar{G}}(i) + (g_{\bar{G}}(i+1) - v)\frac{g_{\bar{G}}(i) - g_{\bar{G}}(i-1)}{g_{\bar{G}}(i+1) - g_{\bar{G}}(i)}\lambda^+_{\bar{G}}(i-1)\\
\gamma_{G}(v) & = 0\\
\Gamma_{G}(i \delta) &= \hat{\Gamma}_{G}(i)
\end{align*}
where $i$ is the largest integer such that $g_{\bar{G}}(i) < v$.
Observe that $\d$ is feasible and
\begin{align*}
q_{\bar{G}}f_{\bar{G}}(v)\Phi_{\bar{G}}^{\d, X}(v) &= - q_{\bar{G}}\frac{g_{\bar{G}}(i+1)\left(1 - \hat{F}_{\bar{G}}(v) \right) - g_{\bar{G}}(i)\left(1 + \hat{f}_{\bar{G}}(i) - \hat{F}_{\bar{G}}(v) \right)}{g_{\bar{G}}(i+1) - g_{\bar{G}}(i)} \\
&- \lambda^+_{\bar{G}}(i)+ \frac{g_{\bar{G}}(i) - g_{\bar{G}}(i-1)}{g_{\bar{G}}(i+1) - g_{\bar{G}}(i)}\lambda^+_{\bar{G}}(i-1) - \sum_{j = i+1}^{H/\epsilon} \Gamma_{{\bar{G}}}(j) 
\end{align*}
Since $\Phi_{\bar{G}}^{\d, X}(v)$ is constant in $[g_{\bar{G}}(i), g_{\bar{G}}(i+1)]$, we have:

\[\int_{g_{\bar{G}}(i)}^{g_{\bar{G}}(i+1)} q_{\bar{G}}f_{\bar{G}}(t) \max\left(0, \Phi_{\bar{G}}^{\d, X}(t)\right)dt  =  \max\left(0, \int_{g_{\bar{G}}(i)}^{g_{\bar{G}}(i+1)} q_{\bar{G}}f_{\bar{G}}(t)\Phi_{\bar{G}}^{\d, X}(t)dt \right)\]

But, 
\begin{align*}
&\int_{g_{\bar{G}}(i)}^{g_{\bar{G}}(i+1)} q_{\bar{G}}f_{\bar{G}}(t)\Phi_{\bar{G}}^{\d, X}(t)dt \\
&= g_{\bar{G}}(i) q_{\bar{G}}\hat{f}_{\bar{G}}(i) - \lambda^+_{\bar{G}}(i-1)\left(g_{\bar{G}}(i-1) - g_{\bar{G}}(i)\right) - \int_{g_{\bar{G}}(i)}^{g_{\bar{G}}(i+1)}\left( \lambda^+_{\bar{G}}(i) + q_{\bar{G}}\left(1 - \hat{F}_{\bar{G}}(t) \right) + \sum_{j=i+1}^{H'}\Gamma_{{\bar{G}}}(j)\right) dt\\
&= g_{\bar{G}}(i) q_{\bar{G}}\hat{f}_{\bar{G}}(i) - \lambda^+_{\bar{G}}(i-1)\left(g_{\bar{G}}(i-1) - g_{\bar{G}}(i)\right) - \lambda^-_{\bar{G}}(i+1)\left(g_{\bar{G}}(i+1) - g_{\bar{G}}(i)\right)  \tag{\autoref{eq:flowg}}\\
&= q_{\bar{G}}\hat{f}_{\bar{G}}(i) \hat{\Phi}^{\hat{\lambda},\hat{\Gamma}}_{\bar{G}}(i) 
\end{align*}
Thus,  
$\int_{g_{\bar{G}}(i)}^{g_{\bar{G}}(i+1)} q_{\bar{G}}f_{\bar{G}}(t)\max\left(0, \Phi_{\bar{G}}^{\d, X}(t)\right)dt  = q_{\bar{G}}\hat{f}_{\bar{G}}(i)\max\left(0,  \hat{\Phi}^{\hat{\lambda},\hat{\Gamma}}_{\bar{G}}(i) \right)$. Summing over $i$, $\bar{G}$, we get that $\D_{\delta}(\I) = \D_X(\d)\geq \D_X(\I)$.
\end{proof}

\subsection{Proof of \autoref{thm:MT}} \label{app:MT}

We need the following technical lemma:

\begin{lemma} \label{lemma:curves} Let $H>1$ and $R(x) : [1,H] \to \mathbb{R}$ be a function such that $R(1) = 1$ and for all $x,y \in [1,H]$ such that $x < y$, we have $-\frac{y-x}{2H} < R(y) - R(x) < \frac{y-x}{2H}$. Then, there exists a distribution $F$ supported on $[0,H]$ such that 

\[x ( 1- F(x)) = \begin{cases} x &, 0 \leq x < 1\\R(x) &, 1 \leq x \leq H\\\end{cases}.\] 
\end{lemma}
\begin{proof}
Define:
\[F(x) = \begin{cases} 0 &, 0 \leq x < 1\\1 - \frac{R(x)}{x} &, 1 \leq x < H\\1 &, x= H\\\end{cases}.\]
Observe that $F$ satisfies the requirements of the theorem. To prove that $F$ is a valid distribution, it is sufficient to show $F(x) < F(y)$ for all $x< y$. We first note that $1 - \frac{R(x)}{x} \in [0,1]$. Thus, the only case left is when $x, y \in [1,H]$. In this case,
\begin{align*}
xy\left(F(y) - F(x)\right) &= yR(x)- xR(y) = yR(x) - xR(x) + x(R(x)- R(y))\\
&> yR(x) - xR(x) - x\frac{y-x}{2H}\\
&= (y -x)\left(R(x) - \frac{x}{2H}\right) =(y -x)\left(R(x) - 1 +\frac{2H-x}{2H}\right) \\
&> (y -x)\left(R(x)  - R(1) + \frac{x-1}{2H}\right) > 0.
\end{align*}
\end{proof}

\begin{proof}[Proof of \autoref{thm:MT}] Let $\I = (q, f_{\bar{G}}, g_G)$ be such that $q_{\bar{G}} = \frac{1}{3}$ and $g_G(v) = g'_G(v)$ for all $v$. In order to define $f_{\bar{G}}$, define
\begin{align*}
Q_{\bar{G}}(v) &= \begin{cases}
1 + \frac{v  - 1}{100H} &, 1 \leq v < \underline{\rho}_G\\
1 + \frac{\underline{\rho}_G  - 1}{100H} &, \underline{\rho}_G \leq v < \overline{\rho}_G\\
1 + \frac{\overline{\rho}_G + \underline{\rho}_G - 1 - v}{100H} &, \overline{\rho}_G \leq v \leq H\\
\end{cases}\\
\Lambda_G(v) &= \frac{1}{100H^2}\sum_i  \min\left((v - \underline{y}_{G,i})(v - \overline{y}_{G,i}), 0\right)
\end{align*}

Let $\Lambda_C(v) = 0$ throughout. Also, define:
\begin{align*}
Z_{G,i}(v) = \begin{cases}
\frac{1}{4^i} &, v < \underline{z}_{G,i}\\
\frac{1}{4^i} \frac{\overline{z}_{G,i} - v}{\overline{z}_{G,i} - \underline{z}_{G,i}} &, v \in [\underline{z}_{G,i}, \overline{z}_{G,i}]\\
0 &, v > \overline{z}_{G,i}\\
\end{cases}\quad\quad&\quad\quad Y_{G,i}(v) = \begin{cases}
\frac{1}{4^i} &, v < \underline{x}_{G,i}\\
0 &, v > \overline{x}_{G,i}\\
\end{cases}\\
P_{1,G}(v) = \frac{1}{100H}\sum_iZ_{G,i}(v)\quad\quad\quad\quad\quad\quad\quad&\quad\quad P_{2,G}(v) = \frac{1}{100H}\sum_iY_{G,i}(v)
\end{align*}

Finally, define $P_{1,C}(v) = -P_{1,A}(v) - P_{1,B}(v)$ and $P_{2,C}(v) = -P_{2,A}(v) - P_{2,B}(v)$ and ,
\begin{align*}
R_{\bar{G}}(v) &= Q_{\bar{G}}(v) + \frac{1}{q_{\bar{G}}}\left[ \Lambda_{\bar{G}}(v) - \int_1^v P_{1,\bar{G}}(g_{\bar{G}}^{-1}(t)) dt - \int_1^v P_{2,\bar{G}}(g_{\bar{G}}^{-1}(t)) dt \right]\\
\end{align*}

Observe that the function $R_{\bar{G}}$ satisfies all the requirements of \autoref{lemma:curves}. Thus, there exist distributions $f_{\bar{G}}$ such that 
\[x ( 1- F_{\bar{G}}(x)) = \begin{cases} x &, 0 \leq x < 1\\R_{\bar{G}}(x) &, 1 \leq x \leq H\\\end{cases}.\] 

Set $\lambda_{\bar{G}} = -\Lambda_{\bar{G}}$, $\gamma_{G}(v) = -P'_{1,G}(v)$ and $\Gamma_{G}(\cdot):X_G \to \mathbb{R}$ to be the unique function such that $\sum_{ x > v} \Gamma(x) = P_{2,G}(v)$.

Observe that the dual defined by $(\lambda_{\bar{G}}, \gamma_G, \Gamma_G)$ satisfies all the requirements of the theorem.

\end{proof}

\subsection{Omitted Proofs in \autoref{sec:RRSLB1}} \label{app:lb1}

\begin{proof}[Proof of \autoref{lemma:compslack1}]  We verify each of the constraints in \autoref{eq:cs} and leave verifying feasibility using \autoref{eq:primal} to the reader. The constraints \eqref{eq:pos}, \eqref{eq:neg} are verified easily. The constraint \eqref{eq:iron} is true because $\lambda_G(v) = 0$ throughout. The constraint \eqref{eq:flowa} holds because for all $v \in (\rho, H+2)$, we have
\begin{align*}
\int_0^v a^*_C(t) dt &= v - \rho\\
\int_0^{g_B(v)} a^*_B(t) dt &=  \int_0^{v} a^*_B(t) dt  = v - \rho\\
\int_0^{g_A(v)} a^*_A(t) dt =  \int_2^{g_A(v)} \mathfrak{a}(t-2) dt &= \int_0^{g_A(v) - 2} \mathfrak{a}(t) dt = \mathfrak{A}(g_A(v) - 2) = v- \rho.
\end{align*}
and the three quantities are equal. Finally, the constraint \eqref{eq:flowadisc} is satisfied because $\Gamma$ is zero throughout.

\end{proof}

\subsection{Omitted Proofs in \autoref{sec:RRSLB2}} \label{app:lb2}

Recall that $x_1 = \frac{4H}{5} + 1$, $y_1 = \frac{3H}{4} + 1$, and $x_i = \frac{8}{3}\frac{H}{2^i} + 1$ , $y_i = \frac{8}{5} \frac{H}{2^i}  + 1$ for $i > 1$.  This implies the equations  
\begin{subequations} \label{eq:xy}
\begin{align}
\forall i > 1: x_i  - y_i &= \frac{16H}{15\cdot 2^i}.  \label{eq:xy1}\\
\forall i > 1: y_i - x_{i+1} &= \frac{4H}{15\cdot 2^i}.  \label{eq:xy2}\\
\forall i > 1: \frac{1}{2} \left( x_{i+1} - y_{i+1}\right) + 2 \left( y_{i} - x_{i+1}\right) &= \left( y_{i} - y_{i+1}\right) \label{eq:xy3}\\
 5y_2 + 16y_1 &= 21x_2. \label{eq:xy4}
\end{align}
\end{subequations}

We will need the following lemma:
\begin{lemma}\label{lemma:doublearrow} For all $i > 0$ and $v \in [y_{i+1}, y_i]$, 
\[\int_{g_G(y_{i+1})}^{g_G(v)} a^*_G(x) dx \leq \int_{y_{i+1}}^{v} a^*_C(x) dx.\]
Moreover, equality holds if $v = y_i$.
\end{lemma}
\begin{proof} We calculate the three quantities:

\[\int_{g_A(y_{i+1})}^{g_A(v)} a^*_A(t) dt = \begin{cases}
\frac{10}{13}\left( g_A(v) - g_A(y_{2})\right)   &, i = 1, v \leq x_2\\
\frac{10}{13}\left( g_A(x_{2}) - g_A(y_{2})\right) + \left( g_A(v) - g_A(x_{2})\right)  &, i = 1, v > x_2\\
\frac{40}{13\cdot 2^{i}}\left( g_A(v) - g_A(y_{i+1})\right) &, i \text{ is even}\\
\frac{40}{13\cdot 2^{i+1}}\left( g_A(v) - g_A(y_{i+1})\right)  &, i \text{ is odd}, v \leq x_{i+1}\\
\frac{40}{13\cdot 2^{i+1}}\left( g_A(x_{i+1}) - g_A(y_{i+1})\right) + \frac{40}{13\cdot 2^{i-1}}\left( g_A(v) - g_A(x_{i+1})\right)  &, i \text{ is odd}, v > x_{i+1}\\
\end{cases}\]

\[\int_{g_B(y_{i+1})}^{g_B(v)} a^*_B(t) dt = \begin{cases}
\frac{40}{13\cdot 2^{i+2}}\left( v - y_{i+1}\right)  &, i \text{ is even}, v \leq x_{i+1}\\
\frac{40}{13\cdot 2^{i+2}}\left( x_{i+1} - y_{i+1}\right) + \frac{40}{13\cdot 2^{i}}\left( v - x_{i+1}\right)  &, i \text{ is even}, v > x_{i+1}\\
\frac{40}{13\cdot 2^{i+1}}\left(v- y_{i+1}\right) &, i \text{ is odd}\\
\end{cases}\]

\[\int_{y_{i+1}}^{v} a^*_C(t) dt = \frac{20}{13\cdot 2^{i}}\left( v- y_{i+1}\right)\]
We now prove the result for $i=1$. In this case, as the expressions for $B$ and $C$ are the same, it is sufficient to show that
\begin{align*}
\int_{g_A(y_{i+1})}^{g_A(v)} a^*_A(t) dt  &= \begin{cases}
\frac{5}{13}\left(v - y_{2}\right) &, v \leq  x_2\\
\frac{5}{13}\left(x_2 - y_{2}\right) +  2v - 2x_2  &, v > x_2\\
\end{cases}\\
&= \begin{cases}
\frac{5}{13}\left(v - y_{2}\right) &, v \leq  x_2\\
2v - \frac{10y_2}{13} - \frac{16y_1}{13}  &, v > x_2\\
\end{cases}\tag{\autoref{eq:xy4}}\\
&\leq \frac{10}{13}\left(v - y_{2}\right)  = \int_{y_{i+1}}^{v} a^*_C(t) dt \tag{As $v < y_1$}
\end{align*}

 We now prove for even $i  > 1$. The case for odd $i$ is similar. Note that $g_A(x_i) = \frac{x_i+1}{2}$ and $g_A(y_i) = \frac{y_i+1}{2}$ for all $i > 1$.  In this case, as the expressions for $A$ and $C$ are the same, it is sufficient to show that 
 \begin{align*}
\int_{g_B(y_{i+1})}^{g_B(v)} a^*_B(t) dt  &= \begin{cases}
\frac{40}{13\cdot 2^{i+2}}\left( v - y_{i+1}\right)  &, v \leq x_{i+1}\\
\frac{40}{13\cdot 2^{i+2}}\left( x_{i+1} - y_{i+1}\right) + \frac{40}{13\cdot 2^{i}}\left( v - x_{i+1}\right)  &, v > x_{i+1}\\
\end{cases}\\
&= \begin{cases}
\frac{10}{13\cdot 2^{i}}\left( v - y_{i+1}\right)  &, v \leq x_{i+1}\\
\frac{20}{13\cdot 2^{i}}\left( y_{i} - y_{i+1}\right) + \frac{40}{13\cdot 2^{i}}\left( v - y_{i}\right)  &, v > x_{i+1}\\
\end{cases}\tag{\autoref{eq:xy3}}\\
&\leq \frac{20}{13}\left(v - y_{i}\right)  = \int_{y_{i+1}}^{v} a^*_C(t) dt \tag{As $v < y_1$}
\end{align*}

The moreover part can be observed by putting $v = y_i$ in our equations.
\end{proof}

\begin{proof}[Proof of \autoref{lemma:compslack2}]~\vspace{-0.5cm}
\paragraph{Feasibility} We verify the feasibility constraints in \autoref{eq:primal}. It is straightforward to verify that $a^*_{\bar{G}}$ are monotone (constraint \eqref{eq:allocmonotone}) and take values in $[0,1]$ (constraint \eqref{eq:allocrange}). Finally, the constraint \eqref{eq:util} holds because of  \autoref{lemma:doublearrow}.\vspace{-0.5cm}
\paragraph{Optimality}
We verify all the constraints in \autoref{eq:cs}. The constraints \eqref{eq:pos}, \eqref{eq:neg} are straightforward to verify from the definition of $a^*_G$. The constraint \eqref{eq:iron} follows from the fact that $\lambda_C(v)$ is $0$ throughout and 
\begin{align*}
\lambda_A(v) > 0 &\implies \exists i > 0 : v \in (g_A(x_{2i+1}),g_A(x_{2i-1})) \implies a'^*_A(v) = 0\\
\lambda_B(v) > 0 &\implies \exists i > 0 : v \in (g_B(x_{2i+2}),g_B(x_{2i})) \implies a'^*_B(v) = 0
\end{align*}
The constraint \eqref{eq:flowa} holds because $\gamma_{G}(v)$ is $0$ throughout. For the constraint \eqref{eq:flowadisc}, we need to prove that for all $y_i$

\[\int_0^{g_A(y_i)}a^*_A(x)dx = \int_0^{g_B(y_i)}a^*_B(x)dx = \int_0^{y_i}a^*_C(x)dx\]

which holds because of (the moreover part of) \autoref{lemma:doublearrow}
 
 \end{proof}

\end{document}